%% file: main.tex
\documentclass{article}
\usepackage[utf8]{inputenc}

\usepackage[margin=1in]{geometry}
\usepackage{amsmath}
\usepackage{enumerate}
\usepackage{amsthm}
\usepackage{amssymb}
\usepackage{graphicx}
\usepackage{tikz}
\usepackage{pgfplots}
\usetikzlibrary{arrows}
\usepackage{caption}
\usepackage{subcaption}
\usepackage{bbm}
\usepackage{enumitem}
\usepackage[colorlinks=true]{hyperref}
\usepackage{cleveref}
\usepackage{float}
\usepackage{thmtools, thm-restate}

\usepackage[noline,noend,linesnumberedhidden,boxed,nokwfunc]{algorithm2e}
\DontPrintSemicolon
\usepackage{algorithmic}

\usepackage[T1]{fontenc}

\tikzstyle{vertex}=[circle, draw, fill, inner sep=0pt, minimum size=5pt]

\newcommand\abs[1]{\left\lvert{#1}\right\rvert}
\newcommand\ceil[1]{\left\lceil{#1}\right\rceil}

\newcommand\norm[1]{\left\lVert{#1}\right\rVert}
\newcommand\wtd[1]{\widetilde{#1}}

\DeclareMathOperator{\gap}{\mathrm{gap}}

\DeclareMathOperator{\Null}{\mathrm{Null}}

\def\Am{\mathbf{A}}
\def\Bm{\mathbf{B}}
\def\Lm{\mathbf{L}}

\def\R{\mathbb{R}}

\def\sE{\mathcal{E}}
\def\sB{\mathcal{B}}

\def\trans{{^\top}}
\def\E{\mathbb{E}}

\def\bv{\mathbf{b}}
\def\xv{\mathbf{x}}
\def\fv{\mathbf{f}}
\def\gv{\mathbf{g}}
\def\rv{\mathbf{r}}

\def\wv{\mathbf{w}}
\def\uv{\mathbf{u}}
\def\vv{\mathbf{v}}
\def\yv{\mathbf{y}}
\def\zv{\mathbf{z}}

\def\one{\mathbbm{1}}

\def\st{\mathrm{st}}

\def\ALGNAME{\textsf{Dual KOSZ}}
\def\OALGNAME{\textsf{KOSZ}}
\mathcode`l="8000
\begingroup
\makeatletter
\lccode`\~=`\l
\DeclareMathSymbol{\lsb@l}{\mathalpha}{letters}{`l}
\lowercase{\gdef~{\ifnum\the\mathgroup=\m@ne \ell \else \lsb@l \fi}}%
\endgroup

\newtheorem*{claim*}{Claim}

\newtheorem{lemma}{Lemma}
\newtheorem*{lemma*}{Lemma}

\newtheorem{prop}{Proposition}
\newtheorem*{prop*}{Proposition}
\newtheorem{theorem}{Theorem}
\newtheorem*{theorem*}{Theorem}
\newtheorem{defn}{Definition}
\newtheorem*{defn*}{Definition}

\newtheorem*{convention*}{Convention}

\newtheorem{fact}{Fact}

\newcommand*\samethanks[1][\value{footnote}]{\footnotemark[#1]}

\theoremstyle{plain}

\newenvironment{remark}{\noindent{\bf Remark}\hspace*{1em}}{\bigskip}
\usepackage{multicol}

\begin{document}

\title{Cut-Toggling and Cycle-Toggling for Electrical Flow and Other $p$-Norm Flows}
\author{Monika Henzinger\thanks{Address: Fakult\"at f\"ur Informatik, Forschungsgruppe Theorie und Anwendung von Algorithmen, W\"ahringer Strasse 29/6.32, A-1090 Wien, Austria.  Email: {\tt monika.henzinger@univie.ac.at}} \\University of Vienna \and  Billy Jin\thanks{Address: School of Operations Research and Information Engineering, Cornell University, Ithaca, NY 14853, USA.  Email: {\tt \{bzj3, davidpwilliamson\}@cornell.edu}.  Supported in part by NSF grant CCF-2007009 and NSERC fellowship PGSD3-532673-2019.}\\Cornell University \and Richard Peng\thanks{Address:  School of Computer Science, Georgia Tech, Atlanta, GA 30332, USA.  Email: {\tt rpeng@cc.gatech.edu}. }\\Georgia Tech \& \\University of Waterloo \and David P.\ Williamson\samethanks[2]\\Cornell University}
% \author{{\color{red} Anonymous Author(s)}}

%\thispagestyle{empty}
\maketitle
% \pagenumbering{gobble}

\begin{abstract}

We study the problem of finding $p$-norm flows in undirected graphs so as to minimize the weighted $p$-norm of the flow for any $p > 1$. 
When $p=2$, the problem is that of finding an electrical flow, and its dual is equivalent to solving a Laplacian linear system. The case $p = \infty$ corresponds to finding a min-congestion flow, which is equivalent to max-flows.
A typical algorithmic construction for such problems considers dual variables $x(i)$ corresponding to the flow conservation constraints for each $i \in V$, and
has two simple types of update steps: {\em cycle toggling}, which modifies the flow along a cycle, and {\em cut toggling}, which modifies all potentials on one side of a cut.
Both types of steps are typically performed relative to a spanning tree $T$; then the cycle is a fundamental cycle of $T$, and the cut is a fundamental cut of $T$.  
In this paper, we show that these simple steps can be used to 
give a novel
efficient implementation for the $p = 2$ case and to
find near-optimal $p$-norm flows in a low number of iterations for all values of $p > 1$.
Compared to known faster algorithms for these problems,
our algorithms are simpler, more combinatorial, and also
expose several underlying connections between these algorithms
and dynamic graph data structures that have not been formalized previously.
\begin{itemize}
\item For $p=2$, we give a cut-toggling algorithm that is dual to
the randomized cycle-toggling algorithm of
[Kelner-Orrechia-Sidford-Zhu STOC'13].
Their algorithm performs a near-linear number of cycle-toggling steps and
uses a data structure to implement each such step in logarithmic time, giving a near-linear time algorithm overall.
While our dual algorithm also runs in a near-linear number of cut-toggling steps,
we show that if we abstract the needed cut-toggling step as a natural data structure problem, this problem can be reduced to the online matrix-vector (OMv) problem, which has been conjectured to be hard
[Henzinger-Krinninger-Nanongkai-Saranurak STOC'15].
This implies that it is unlikely for a cut-toggling step to
be implementable in sublinear time,
but we then circumvent this difficulty via batching, sparsification, and recursion,
obtaining an overall almost-linear running time.
\item For general $p$-norm flows, we show that
$\widetilde{O} (\frac{1}{p-1}2^{\frac{p+1}{p-1}} m + nm^{p-1}R)$ cut-toggling iterations are sufficient to find a near-optimal flow when $1 < p \leq 2$, with $R$ the ratio between the maximum and minimum values of the edge weights. 
When $p \geq 2$,
we show that $\widetilde{O}((p2^{2p-1}+ (nR)^{\frac{1}{p-1}}) m)$ cycle-toggling iterations are sufficient.
This exposes a separation between cut and cycle toggling steps once $p$ moves away from $2$.
It also represents a starting point toward getting faster and more robust algorithms for $p$-norm flows,
but also leads to a significantly more difficult problem where
the tree (from which the fundamental cuts are picked) is 
dynamically changing as the algorithm progresses.
\end{itemize}
\end{abstract}

% \newpage

\pagenumbering{arabic}
\input{intro}

\input{notation}

\input{laplacian}

\input{pnorm}

\input{lowerbound}

\input{batching}

\input{sparsifyrecurse}

\appendix

\input{appendix}

\bibliographystyle{alpha} 
\bibliography{references}

\end{document}

%% file: intro.tex
\section{Introduction}

We study the problem of finding flows in undirected graphs so as to minimize the weighted $p$-norm of the flow for $p > 1$.  In particular we are given an undirected graph $G=(V,E)$, weights $r(i,j) > 0$
%\todo{is this resistance or weights? for $p = 2$ the `weights' are $1 / resistances$}
for each $(i,j) \in E$, and supplies $b(i)$ for each $i \in V$ such that $\sum_{i \in V} b(i)=0.$
% (We will call the $r(i,j)$ values resistances in the rest of the paper, even though in prior work this name has only be used for $p=2$.)
Let $\vec{E}$ be some arbitrary orientation of $E$, and let $\Am$ be the vertex-arc incidence matrix for $(V, \vec{E})$. The minimum-weighted $p$-norm flow problem and its dual are given below. (See Appendix \ref{app:dual_deriv} for how the dual is derived.)
\begin{align*}
(P) \quad &\min  \quad \frac1p\sum_{(i,j)\in\vec{E}} r(i,j)\abs{f(i,j)}^p  
\qquad\qquad (D)  &\max \quad \bv^T\xv - \left(1-\frac1p\right)\sum_{(i,j)\in\vec{E}} \left(\frac{\abs{x(i) - x(j)}^p}{r(i,j)}\right)^{\frac{1}{p-1}} \\
&\text{s.t.} \quad \Am\fv = \bv 
&\text{s.t.} \quad \xv \in \R^{V}~\text{unconstrained} 
\end{align*}
Various values of $p$ give classical flow problems: The case $p=1$ corresponds to an uncapacitated and undirected minimum-cost flow problem, in which $r(i,j)$ is the per-unit cost of shipping flow on edge $(i,j)$; if the graph were directed and we added capacity constraints, the problem would be the general minimum-cost flow problem.  The case $p=2$ in which $r(i,j)$ is the resistance of the edge corresponds to finding an electrical flow; that is, a flow that minimizes the total energy in the network.    The case $p=\infty$ corresponds to minimizing the congestion in a flow in an undirected graph when $r(i,j) = u(i,j)^p$, for $u(i,j)$ the capacity of the edge, which is equal to the maximum flow problem in undirected graphs.

A very typical algorithmic approach for both the minimum-cost flow problem and the electrical flow problem is to consider dual variables (or \emph{potentials}) $x(i)$ corresponding to the flow conservation constraints for each node $i \in V$, and to find flows $\mathbf{f}$ and potentials $\mathbf{x}$ that meet an optimality condition for the flow problem in question.  For instance, for the general (directed, capacitated) minimum-cost flow problem, complementary slackness guarantees that a feasible flow has minimum cost if there exist potentials $\mathbf{x}$ such that  $r(i,j) + x(i) - x(j)$ is nonnegative for all directed edges $(i,j)$ such that $f(i,j)$ is strictly less than the capacity, and $r(i,j) + x(i) -x(j)$ is nonpositive for all $(i,j)$ such that $f(i,j)$ is positive. Similarly, for electrical flow a feasible flow is optimal if there exist potentials $\mathbf{x}$ such that Ohm's Law is obeyed, and $f(i,j) = (x(i)-x(j))/r(i,j)$ for all $(i,j) \in \vec{E}$.  
In the case of electrical flow such potentials $\mathbf{x}$ are the solution to the linear system $\mathbf{Lx = b}$, where $\mathbf{L}$ is the weighted Laplacian of the graph with weight on each edge $(i,j)$ of $1/r(i,j)$, and $\mathbf{b}$ is the supply vector.  Algorithms manipulating potentials $\mathbf{x}$ tend either to be {\em primal-feasible} algorithms that maintain a  feasible flow $\mathbf{f}$ while finding potentials $\mathbf{x}$ that meet the optimality conditions, or {\em dual-feasible} algorithms that maintain the optimality condition on the potentials with respect to a current infeasible flow, and update the potentials to drive towards flow feasibility.

There are two very simple update steps, one for each type of algorithm.  For primal-feasible algorithms, a natural update step is {\em cycle toggling}: we push flow around a cycle so as to maintain primal feasibility.  For dual-feasible algorithms, a natural update step is  {\em cut toggling}: given a current set of potentials $\mathbf{x}$, we update $\mathbf{x}$ by setting $x(i) \gets x(i) + \delta$ for all $i \in S$ for some set $S \subset V$ and some value $\delta$. It is furthermore typical that such steps are made with reference to some spanning tree $T$ in the graph.  Then a cycle-toggling step is performed with respect to the fundamental cycle closed by adding some non-tree edge to $T$, and a cut-toggling step is performed with respect to a fundamental cut in the tree, which is a cut induced by removing an edge of some spanning tree $T$ of the graph.
%\todo{emphasize that the cuts/cycles are coming from trees?}
These two styles of algorithm are well-known for the (directed, capacitated) minimum-cost flow problems, and correspond to the primal and dual network simplex algorithms respectively, in which the tree $T$ corresponds to the current simplex basis.
More generally, there are both cycle-canceling algorithms and cut-canceling algorithms for the minimum-cost flow problem that choose appropriate cycles or cuts in the graph and perform a cycle-toggling or a cut-toggling iteration (for cycle toggling see, for instance, Klein \cite{Klein67} and Goldberg and Tarjan \cite{GoldbergT89}; for cut toggling, see  Hassin \cite{Hassin83} and Ervolina and McCormick \cite{ErvolinaM93}).  

In the case of electrical flow, Kelner, Orrechia, Sidford, and Zhu \cite{KOSZ13} present a randomized near-linear time cycle-toggling algorithm  that finds a near-minimum energy flow $\mathbf{f}$, and also an approximate solution to $\Lm\xv = \bv$.  Their algorithm finds a low-stretch spanning tree $T$ with respect to $\rv$, and performs a near-linear number of iterations, each of which modifies flow on a fundamental cycle with respect to $T$.  An appropriate choice of data structure allows them to implement each cycle-toggling iteration in logarithmic time, leading to the overall near-linear running time.  However, no corresponding cut-toggling algorithm exists in the literature, leading immediately to the following open question:

{\emph{Open question: Does there exist a cut-toggling algorithm for computing near-minimum energy flows / approximately solving Laplacian linear systems, and how efficiently can it be implemented?}}

%\todo{Add paragraph(s) to motivate why this problem is interesting / worth working on.}
%\begin{itemize}
%    \item Community thinks p norms should be harder than LPs
%    \item Applications to max flow
%\end{itemize}

%Another interesting research direction is to extend the cycle-toggling and cut-toggling 
%approach to general values of $p$.

There is a vast literature on solving Laplacian linear systems, and the current  fastest algorithm is the algorithm of \cite{JS21}, which runs in $O(m (\log \log n)^{O(1)} \log \frac{1}{\epsilon})$ time.
These works have motivated the Laplacian paradigm of graph algorithms~\cite{T10}:
solving problems on graphs and networks
using (a sequence of) linear systems of Laplacians. 
Further improvements of these new graph algorithms 
have increasingly emphasized the following question: 
\textit{Which class of algorithmic problems can be solved
using tools from linear systems solvers?}
The combined formulation of Laplacians and $p$-norm flows
given at the start of the introduction is directly motivated
by this connection: the recent development of
almost linear time solvers for $O(\log^{2/3}n)$-norm flows
and their duals~\cite{KPSW19,AS20} have already led to improvements
to extensively studied problems such as unit-capacity
flows and bipartite matchings~\cite{LS20,KLS20,AMV20}.
While many algorithms are known for different values
of $p$ (see Appendix~\ref{sec:previous} for a summary),
they use a host of methods based on continuous optimization,
and there is no clear winner.
The exponent of $m$ for the current best running times for different values of $p$,
for sparse (i.e. $m = n^{1 + o(1)}$),
unit-weighted graphs is given in  plot in Figure~\ref{fig:runtimes}, e.g., the 
best algorithm for $p=2$ takes time $\wtd{O}(m)$, and thus, we plot 1 for $p=2$.
As the figure shows, the complexity of $p$-norm flows is not well understood: For different values of $p$, different algorithms perform best, and there are currently two ``local minima'': for $p=2$ and for $p=\infty$.
%for large value of $p$ the running times of the currently best algorithm is even getting smaller.

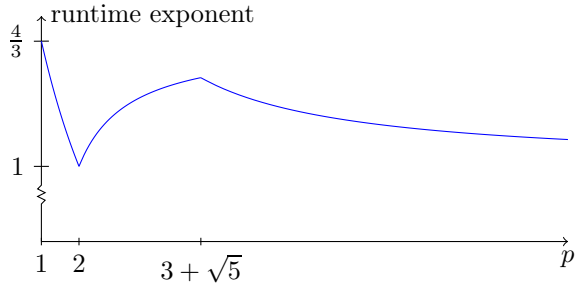
\begin{figure}
\begin{center}
\begin{tikzpicture}[xscale = 0.5, yscale = 5]
  \draw[->] (1, 0) -- (15, 0) node[below] {$p$};
  \draw[-]  (1, 0) -- (1, 0.1) -- (1.1, 0.11) -- (0.9, 0.12) -- (1.1, 0.13) -- (0.9, 0.14) -- (1, 0.15) -- (1, 0.2);
  \draw[->] (1, 0.2) -- (1, 0.6) node[right] {runtime exponent};
  \draw[-] (1, 0.01) -- (1, -0.01) node[below] {$1$};
  \draw[-] (2, 0.01) -- (2, -0.01) node[below] {$2$};
  \draw[-] (5.24, 0.01) -- (5.24, -0.01) node[below] {$3 + \sqrt{5}$};
  \draw[-] (1.2, 0.5333) -- (0.8,0.533) node[left] {$\frac{4}{3}$};
  \draw[-] (1.2, 0.2) -- (0.8,0.2) node[left] {$1$};
  \draw[domain=1:2, smooth, variable=\x, blue] plot ({\x}, {0.2 + (2 - \x)/(2 + \x)});
  \draw[domain=2:5.24, smooth, variable=\x, blue] plot ({\x}, {0.2 + (\x - 2)/(3 * \x - 2)});
  \draw[domain=5.24:15, smooth, variable=\x, blue] plot ({\x}, {0.2 + 1/(\x - 1)});
\end{tikzpicture}
\end{center}
\caption{Current best runtimes for computing $p$-norm flows on a sparse graph as an exponent of $n$, a linear plot.
Here both $p$ and $\frac{1}{p - 1}$ are constants.
%Red curve is the our iteration count as $p \rightarrow 1$.
}
\label{fig:runtimes}
\end{figure}

Furthermore, all existing algorithms for $p$-norm flows
are non-combinatorial.
They make use of numerical and analytic tools such as
homotopy methods~\cite{BCLL18},
multiplicative weight update~\cite{AKPS19},
higher order acceleration~\cite{B20}, and
recursive preconditioning~\cite{AKPS19}.
These approaches, as well as the clear gaps in
our current understanding of this problem,
lead us to ask:

{\emph{Open question: Are there simple combinatorial algorithms (such as cut- or cycle-toggling algorithms) for computing near-optimal $p$-norm flows?}}

%In particular, a $(1+\epsilon)$-approximation to the optimal $p$-norm flow can be found in $O(p^2 m^{\frac{4p-4}{3p-2} + o(1)}\log^2 \frac{1}{\epsilon})$ arithmetic operations for $2 \leq p < \mbox{poly}(m)$ \cite{AS20}, or time $\tilde{O}(pm^{1 + \abs{\frac{1}{2}-\frac{1}{p}}})$ \cite{BCLL18} for $1 < p < \infty$, time $\tilde{O}_p(m^{1 + \frac{\abs{p - 2}}{2p + \abs{p - 2}}})$ time for $1 < p < \infty$ \cite{AKPS19}.  An algorithm with time $O(p(m^{1 + o(1)} + n^{4/3 + o(1)})\log^2 \frac{1}{\epsilon})$ has recently been achieved for $p=\omega(1)$ \cite{ABKS21:arxiv}.  
%Faster running times are known both for unweighted graphs ($r(i,j)=1$ for all $(i,j) \in E$) and the case of electrical flow ($p=2$).  
%For unweighted graphs, i.e.,~$r(i,j)=1$ for all $(i,j) \in E$, a near-optimal flow can be computed in $pm^{\frac{p}{p-1} + o(1)}$ arithmetic operations \cite{AS20} for $2 \leq p \leq \mbox{poly}(m)$ or time $2^{O(p^{3/2})}m^{1 + \frac{7}{\sqrt{p-1}}+o(1)}\mbox{poly}(\log \frac{1}{\epsilon})$  for $p \geq 2$ \cite{KPSW19}.

%\mh{What are the disadvantages of these algorithms?}

\subsection{Our Contributions}

(1) We study cut-toggling algorithms from both structural
and efficiency perspectives, and show:
\begin{itemize}
    \item Cut-toggling algorithms can solve,
    to high accuracy,
    graph Laplacian linear systems in a nearly-linear number of iterations,
    and $p$-norm flows for $1 < p \leq 2$ with an iteration 
    count of $\widetilde{O} (\frac{1}{p-1}2^{\frac{p+1}{p-1}} m + nm^{p-1}R)$ where $R$ is the max ratio of weights.

    \item The cut toggling algorithm for solving graph 
    Laplacians (the $p = 2$ case) can run in almost linear time by taking advantage of the ``offline nature'' of the cut toggles: that is, the choice of cuts can be chosen independent of each other, which allows us to ``batch'' the processing of the cut toggles.
\end{itemize}
(2) To complement the iteration count bound for cut toggling,
we also show that
cycle toggling algorithms can solve $p$-norm flows for $p > 2$ with
an iteration count of $\widetilde{O}((p2^{2p-1}+ (nR)^{\frac{1}{p-1}}) m)$.

Thus, we demonstrate that these two basic flow update steps go
beyond algorithms for minimum-cost flows and can solve other $p$-norm flow problems for $p > 1$.
From the numerical/linear systems solving perspective,
our results can also be viewed as demonstrating that these
numerical routines can be extended in ways that more
closely resemble their combinatorial analogs
in min-cost flow / network simplex algorithms.

Algorithmically, our path towards an almost-linear time
cut-toggling Laplacian solver required overcoming a much more complex
data structure problem, compared to cycle toggling.
For cycle toggling, a simplified version of dynamic trees
is sufficient~\cite{KOSZ13}.
Our almost-linear time implementation of cut toggling, on the other hand,
involves modifying the outer-loop/data-structure interactions.
We will discuss these issues, as well as the likelihood
of them becoming even more intricate in the generalized $p$-norm settings,
in our discussion of potential future work after providing our technical overview.

\subsection{Technical Overview}

\textbf{Special case $p=2$.} We begin with the case of $p=2$, in which the primal problem is that of finding an electrical flow and the dual problem is equivalent to solving the Laplacian linear system $\Lm \xv = \bv$.  Here we show that there is a very natural randomized cut-toggling algorithm which is dual to the randomized cycle-toggling algorithm of Kelner et al.; we will refer to their algorithm as \OALGNAME, and to our algorithm as \ALGNAME.  \ALGNAME\ also starts by choosing a low-stretch spanning tree $T$. It maintains a set of potentials $\mathbf{x}$ (initially zero), and the corresponding (infeasible) flow $\mathbf{f}$ implied by Ohm's Law.  In each iteration, we sample a fundamental cut $S$ of the tree $T$ and perform a cut-toggling update so that the net flow leaving $S$ is $\sum_{i \in S} b(i)$, as required in every feasible flow.  Following arguments dual to those made in Kelner et al.\ we show that this algorithm also performs a near-linear number of iterations in order to find a near-optimal set of potentials $\xv$  and flow $\fv$. 

\begin{restatable}{theorem}{dualkosz}
\label{thm:dual_kosz}
Let $\tau$ be the total stretch of $T$. After $K = \tau \ln(\frac\tau\epsilon)$ iterations, \ALGNAME\ returns $\xv^K \in \R^V$ and $\fv^K \in \R^{\vec{E}}$ such that $\E\norm{\mathbf{x}^* - \mathbf{x}^K}_\mathbf{L}^2 \leq \frac{\epsilon}{\tau}\norm{\mathbf{x}^*}_\mathbf{L}^2$ and $\E[\sE(\fv^K)] \leq (1+\epsilon) \sE(\fv^*)$, for $\fv^*$ and $\xv^*$ optimal primal and dual solutions respectively.
\end{restatable}

However, unlike Kelner et al., we cannot show that each individual cut-toggling update can be made to run in polylogarithmic time. If we abstract the desired cut-toggling update step as a natural data structure problem,  we show that such a data structure cannot be implemented in $O(n^{1-\epsilon})$ time for any $\epsilon > 0$ given a conjecture about the {\em online matrix-vector multiplication problem (OMv)} made by Henzinger, Krinninger, Nanongkai and Saranurak \cite{HKNS15}. They have conjectured that this problem does not have any algorithm that can carry out an online sequence of $n$ Boolean matrix-vector multiplications in time $O(n^{3-\epsilon})$, and show that if the conjecture is false, then various long-standing dynamic graph problems will have faster algorithms.  We show that a single Boolean matrix-vector multiply can be carried out as a sequence of $O(n)$ operations of our desired data structure.  Given the conjecture, then, we cannot implement the data structure operations in $O(n^{1-\epsilon})$ time. Thus there is not a straightforward  near-linear time version of the \ALGNAME\  algorithm.\footnote{In a personal communication, Sherman \cite{Sherman17} said he also had worked out a dual version of the \OALGNAME\  algorithm, but was unable to solve the data structure problem for the updates to potentials.  Our result explains why this might be difficult to do.}

Nevertheless, we surmount this data structural lower bound  by exploiting the fact that the sequence of cuts to be updated can be sampled in advance and, thus, the updates can be batched, circumventing the ``online'' (or ``sequential'') requirement in OMv. This is possible because both the spanning tree $T$ and the probability distribution over cuts of $T$ are fixed at the beginning of the algorithm. More precisely, denote the number of iterations of \ALGNAME\ by $K$ (which is $\widetilde{O}(m)$). 
Instead of sampling the fundamental cuts one at a time, consider sampling the next $l$ cuts that need to be updated for some $l \ll K$. In each ``block" of size $l \ll K$, we contract all the edges of $T$ that do not correspond to one of the $l$ fundamental cuts to be updated. In this way, we work with a contracted tree of size $O(l)$ in each block (instead of the full tree, which has size $O(n)$). This makes the updates faster. However, the price we pay is that at the end of each block, we need to propagate the updates we made (which were on the contracted tree), back to the entire tree. Overall, we show that each block takes $O(l^2 + m)$ time. Since there are $\widetilde{O}(\frac{m}{l})$ blocks, the total runtime is $\wtd{O}(ml + \frac{m^2}{l})$.
Choosing $l=\sqrt{m}$ thus gives a $\tilde{O}(m^{1.5})$ time algorithm.  Interestingly, in a computational study of \OALGNAME, Boman, Deweese, and Gilbert \cite{boman_deweese_gilbert_2016} explored an heuristic implementation that  batched its cycle-toggling updates by looking for collections of edge-disjoint cycles, and found that in many cases this gave a speedup in their experiments. 

By augmenting the batching idea with sparsification and recursion, one can further improve the running time of \ALGNAME\  to $\widetilde{O}(m^{1+\delta})$ for any $\delta > 0$.  To do this, observe that $l$ cut-toggling updates effectively break the spanning tree into $l+1$ components. After contracting the components to get a graph $H$ with $l+1$ vertices, we can show that solving an appropriate Laplacian system on $H$ gives a single update step that makes at least as much progress as the sequence of $l$ updates performed by the straightforward unbatched algorithm. A natural approach is to solve this Laplacian system by recursively calling the algorithm. However, this by itself does not give an improved running time. Instead, we first spectrally sparsify $H$ and then call the algorithm recursively to solve the \emph{sparsified} Laplacian system. Here we use the original Spielman-Teng spectral sparsification~\cite{ST11:journal} because it does not
require calling Laplacian solvers as a subroutine (e.g. \cite{BatsonSST13}). By carefully analyzing the error incurred by sparsification, we are able to show that the update step using sparsification  makes about as much progress as the update step without sparsification. The total running time of the recursive algorithm is then obtained by bounding the time taken at each layer of the recursion tree.

\begin{restatable}{theorem}{fastest}
\label{thm:time}
For any $\delta \in (0,1)$, \ALGNAME\ with batching, sparsification, and recursion can be implemented to run in $O(A^{\frac{1}{\delta}}m^{1+\delta}(\log n)^{\frac{B}{\delta}}(\log \frac{1}{\epsilon})^{\frac{1}{\delta}})$ time, where $A$ and $B$ are constants. 
\end{restatable}

\textbf{General $p>1$.} For a general value of $p > 1$, we propose a cycle-toggling algorithm for solving the minimum weighted $p$-norm flow problem and a cut-toggling algorithm for solving its dual. These algorithms generalize \OALGNAME\ and \ALGNAME. The cycle-toggling algorithm works by maintaining a feasible flow $\fv$. At each iteration, it chooses a low-stretch spanning tree $T$ with respect to weights that are determined by $\fv$. It then samples a fundamental cycle of $T$, and adds $\Delta$ to the flow value on every edge in the cycle, where $\Delta \in \R$ is chosen to maximize the decrease in the energy of the flow. On the other hand, the cut-toggling algorithm maintains a vector of potentials $\xv$. Just like the cycle-toggling algorithm, it chooses a low-stretch tree $T$ at each iteration with respect to weights that are determined by $\xv$. It then samples a fundamental cut of $T$, and adds $\Delta$ to the potential of every vertex in the cut, where $\Delta \in \R$ is chosen to maximize the increase in the dual objective. 

Our main result is a bound on the {iteration complexity} of the cycle-toggling and cut-toggling algorithms. We show that the cycle-toggling algorithm can be used to solve the minimum $p$-norm flow problem for all $p \geq 2$, and the cut-toggling algorithm can solve its dual for all $1 < p \leq 2$. The iteration bounds depend on $R$, which is the ratio between maximum and minimum values of $r(i,j)$.

\begin{restatable}{theorem}{cycleiter}
\label{thm:cycle_iter}
For all $p \geq 2$, the cycle-toggling algorithm finds a primal solution $\fv^K$ satisfying $\E[\sE(\fv^K)] \leq (1+\epsilon)\sE(\fv^*)$ in 
$K = O\left(
\left( p2^{2p-1} \cdot m\log n \log\log n +  m(nR)^{\frac{1}{p-1}}\right) \ln \left(\frac{1}{\epsilon}\right)\left(p\ln(n) + \ln(R)
\right)
\right)$ iterations, for $\sE()$ the primal objective function and $\fv^*$ an optimal primal solution.
\end{restatable}

\begin{restatable}{theorem}{cutiter}
\label{thm:cut_iter}
For all $1 < p \leq 2$, the cut-toggling algorithm finds a dual solution $\xv^K$ satisfying $\E[\sB(\xv^K)] \geq (1-\epsilon)\sB(\xv^*)$ in 
$K = O\left(\left(q 2^{2q-1} \cdot m\log n \log\log n +  nRm^{\frac{1}{q-1}}\right)\ln \frac{1}{\epsilon}\right)$ iterations, for $\sB()$ the dual objective function, and $\xv^*$ an optimal dual solution. Here, $q = \frac{p}{p-1}$. 
\end{restatable}

Note that the cut-toggling algorithm, as stated, returns a dual solution $\xv$ and not a feasible flow $\fv$. However, we can show that it is possible to convert an approximately optimal dual solution $\xv$  to an approximately optimal feasible flow $\fv$, at the expense of multiplying the overall iteration count of the cut-toggling algorithm by a logarithmic factor. The conversion works by first routing a potential-defined flow with respect to $\xv$ (defined by the KKT equations), and then routing the residual supplies along the edges of a spanning tree. 
% \bj{should we elaborate more on this?}\mh{Yes, after you fixed everything else.}

The iteration complexity analysis of the cycle-toggling and cut-toggling algorithms for general $p$ involves two main challenges. The first is that both the progress made by one iteration involves summing terms of the form $\abs{f(e) + \delta}^p$, where $f(e)$ represents the value of the current flow on an edge, and $\delta$ is the amount by which we updated it. (The analysis of cut-toggling has similar terms, but involving $x(i)$.) Whereas for $p=2$ this can be expanded exactly as a sum of three terms, for  integer $p\neq 2$  the binomial expansion gives $p+1$ terms, and for $p$ not equal to an integer the Taylor expansion gives an infinite number of terms. The analysis, thus, involves carefully bounding these Taylor expansions. 

The second challenge is that we were unable to make a primal-dual analysis in the spirit of \cite{KOSZ13} go through for the general $p$ case. This is because (1) the expressions for the duality gap and the progress made in each iteration are more complicated, and thus harder to relate to each other, and (2) the low-stretch spanning tree now changes at each iteration, and thus it is harder to reason about tree-defined potentials as is done in the $p=2$ case. To circumvent these technical challenges of a primal-dual analysis, we instead adopt a primal-only analysis for the cycle-toggling algorithm, and a dual-only analysis for the cut-toggling algorithm. This style of analysis directly compares the progress made in an iteration to (a scalar multiple of) the difference in objective values between the current solution and the optimal solution. The first-order terms exactly match up, and the main work in the analysis is in comparing the higher-order terms. 
%The main 
%\mh{nice. can you say anything more about why that works?}

Unless otherwise noted, all proofs that are missing from the main body of the paper are in the Appendix.

% For the more general problem of $p$-norm flows, we show that $\frac{p}{p-1}2^{\frac{p+1}{p-1}}\widetilde{O}(m) + p^{p-1}2^{p+1}nm^{p-1}R$ cut-toggling iterations are sufficient to find a near-optimal flow for $1 < p \leq 2$ and for $R$ the ratio between the maximum and minimum values of $r(i,j)$.   When $p \geq 2$, we show that $p2^{2p-1}\widetilde{O}(m)+ m(p2^{2p-1}nR)^{\frac{1}{p-1}})$ cycle-toggle iterations are sufficient. The analysis here becomes far more difficult, and we do not work out the running time of an iteration or perform a careful error analysis; each iteration involves solving a polynomial equation to obtain the step size and computing a new low-stretch tree in each iteration.  

%  For $p=2$, we give a duality-based analysis of the cut-toggling algorithm that parallels the analysis of the cycle-toggling algorithm in \cite{KOSZ13}. However, unlike the case with the cycle-toggling algorithm, we prove that assuming the OMv conjecture, there is no natural data structure that allows us to perform each iteration of cut-toggling in sublinear time. To circumvent this lower bound, we exploit the fact that the sequence of cuts to be toggled can be generated {in advance} instead of one at a time. This additional structure allows us to use batching, sparsification, and recursion to overcome the data structure lower bound. In essence, this strategy is to transform a fully online problem into a sequence of easier offline problems, and may be applicable in other settings. 

\paragraph{Future Work.}
Our work raises various directions for future work.
(a) Currently we circumvent our lower bound by batching updates, effectively turning the ``online'' problem with $\tilde O(m)$ updates into a batched problem with $\tilde O(\sqrt m)$ batches of $\sqrt m$ updates each. However to achieve a fast algorithm  we have to compress the graph and call the algorithm recursively, which results in a $O(m^{o(1)})$ increase in the running time. It would be interesting to explore if instead we could use  a data structure with polylogarithmic time per operation that  returns an approximate answer. Note that our lower bound also works for data structures that return a multiplicative approximation, but not for additive approximation. Our recursive algorithm for $p=2$ indicates that small additive errors can be tolerated by \ALGNAME. Thus it is interesting to explore whether there is  an efficient data structure for each update step with small additive error. Potentially a combination of batched updates and approximation could lead to a $\tilde O(m)$-time algorithm for $p=2$.
%Or can the lower bound be circumvented by using a different algorithm when the tree is the star?
%We propose a cut-based combinatorial algorithm to solve Laplacian systems approximately. This algorithm is dual to the cycle-based algorithm by Kelner et al. \cite{KOSZ13}. We show that our algorithm converges in a near-linear number of iterations. 

%To achieve a near-linear running time for, we would further need each iteration to run in polylogarithmic time. We give evidence against this, by presenting a reduction from the OMv conjecture. This is in contrast to the algorithm in \cite{KOSZ13}, which uses a data structure such that each iteration of the algorithm runs in $O(\log n)$ time.   In order to obtain a better running time, one would need to  show it is possible take advantage of the particular structure of updates in the algorithm to implement the data structure.  Note that our reduction crucially needs that a very specific spanning tree (albeit with very small stretch) is chosen. Is it possible for the algorithm to choose a different small-stretch spanning tree that is amendable to a polylogarithmic time implementation?

(b) For general $p$-norm flows, we give bounds on the iteration complexity of the cut-toggling and cycle-toggling algorithms. A natural question is if these algorithms can be made to run in time comparable or better than the current state of the art algorithms. The main bottleneck for cut-toggling and cycle-toggling is the fact that they have to compute a new low-stretch spanning tree at the beginning of every iteration. We believe that the key to obtaining a fast running time is to have an efficient method for dynamically maintaining these low-stretch spanning trees. 
Note that the currently best known algorithms for this problem take $n^{o(1)}$ time per \emph{single} edge weight change and return a spanning tree with expected stretch within $n^{o(1)}$ of the minimum stretch~\cite{DBLP:conf/soda/ChechikZ20,DBLP:conf/soda/ForsterGH21}.
However, at the end of each update step in \ALGNAME\ up to $m$ edges could change their weight, resulting in a $O(m n^{o(1)})$ time per update step if we use the data structures of~\cite{DBLP:conf/soda/ChechikZ20,DBLP:conf/soda/ForsterGH21}. It would be interesting to explore whether the dynamic small-strech tree data structures can be modified to handle these very structured weight changes of potentially a large number of edges more efficiently.

\textbf{Paper Structure.} In Section \ref{sec:notation}, we introduce notation and relevant definitions.  In Section \ref{sec:alg} we give the \ALGNAME\ algorithm for the $p=2$ case and prove that it runs in a near-linear number of iterations. Then in Section \ref{sec:pnorm}, we give a general cycle-toggling algorithm and cut-toggling algorithm and analyze their iteration complexities for finding minimum $p$-norm flows. The remainder of the paper then focuses on computational aspects of \ALGNAME\ for the $p=2$ setting. In Section \ref{sec:ds} we give evidence to show that each iteration of \ALGNAME\ cannot be implemented in sublinear time if they are to be performed one-by-one in ``online" fashion. This lower bound is via a reduction to the OMv conjecture.  Then in Section \ref{sec:speedup}, we show how to overcome this data structural lower bound to obtain almost-linear running time for \ALGNAME.

%% file: notation.tex
\section{Notation and Problem Statement} \label{sec:notation}
We are given an undirected graph $G = (V, E)$, with positive weights $\mathbf{r} \in \R^E$. Although the graph is undirected, it is standard to fix an arbitrary orientation $\vec{E}$ of $E$. Let $\Am$ be the vertex-arc incidence matrix of $(V, \vec{E})$. 
In addition to the graph $G$ and the resistances $\mathbf{r}$, we are given a \emph{supply vector} $\mathbf{b} \in \R^V$ such that $\sum_{i \in V} b(i)=0$. We call any flow $\fv \in \R^{\vec{E}}$ that satisfies $\Am \fv = \bv$ a \emph{$\bv$-flow}. For all $(i,j) \in \vec{E}$, we define $f(j,i) = -f(i,j)$. Our goal is to solve the minimum weighted $p$-norm flow problem, which is shown with its dual below.
\begin{align*}
(P) \quad &\min  \quad \frac1p\sum_{(i,j)\in\vec{E}} r(i,j)\abs{f(i,j)}^p  
\qquad\qquad (D)  &\max \quad \bv^T\xv - \left(1-\frac1p\right)\sum_{(i,j)\in\vec{E}} \left(\frac{\abs{x(i) - x(j)}^p}{r(i,j)}\right)^{\frac{1}{p-1}} \\
&\text{s.t.} \quad \Am\fv = \bv 
&\text{s.t.} \quad \xv \in \R^{V}~\text{unconstrained} 
\end{align*}

Let $\sE(\fv)$ denote the primal objective and $\sB(\xv)$ denote the dual objective. For clarity, we will let $q = \frac{p}{p-1}$ and $w(i,j) = r(i,j)^{-\frac{1}{p-1}}$, so that the dual objective is $\sB(\xv) = \bv^T\xv - \frac1q\sum_{(i,j) \in \vec{E}} w(i,j) \abs{x(i)-x(j)}^q$.  

Let $\fv^*$ denote the optimal primal solution and let $\xv^*$ denote an optimal dual solution. Note that there are infinitely many dual solutions, because the dual objective is invariant under adding a constant to every component of $\xv$.  
By strong duality (note that Slater's condition holds), $\sE(\fv^*) = \sB(\xv^*)$. Moreover,
the KKT conditions give a primal-dual characterization of optimality.
\begin{fact}[KKT Conditions for $p$-Norm Flow]
\label{fact:kkt}
Consider $\fv \in \R^{\vec{E}}$ and $\xv \in \R^V$. Then $\fv$ is optimal for the primal and $\xv$ is optimal for the dual if and only if the following conditions hold:
\begin{enumerate}
    \item $\fv$ is a feasible $\bv$-flow;
    \item \label{ohm} For all $(i,j) \in \vec{E}$, $r(i,j)f(i,j)\abs{f(i,j)}^{p-2} = x(i) - x(j)$.\\ 
    Or equivalently, $f(i,j) = w(i,j)(x(i)-x(j))\abs{x(i)-x(j)}^{q-2}$. 
    % \item (Ohm's Law) There exist potentials $\mathbf{p} \in \R^V$ such that $f(i, j) = \frac{p(i) - p(j)}{r(i,j)}$ for all $(i, j) \in \vec{E}$. Moreover, $\mathbf{p}$ is the solution to $\mathbf{L}\mathbf{p} = \mathbf{b}$, up to adding a constant to every component of $\mathbf{p}$.\footnote{Note that the all-ones vector is in the null space of $\mathbf{L}$.} 
\end{enumerate}
\end{fact}
Thus if one is looking for an exact solution, then solving (P) is equivalent to solving (D): Given $\xv^*$, we can calculate $\fv^*$ by using $f^*(i,j) = w(i,j)(x^*(i)-x^*(j))\abs{x^*(i)-x^*(j)}^{q-2}$. On the other hand, given $\fv^*$, we can recover corresponding potentials $\mathbf{x}^*$ by setting $x^*(v) = 0$ for some arbitrary vertex $v$, and using the equation $x^*(i) - x^*(j) = r(i,j)f^*(i,j)\abs{f^*(i,j)}^{p-2}$ to solve for the potential on every other vertex.

Our goal is to compute an approximate minimum weighted $p$-norm flow $\fv$. More precisely:

\begin{center}
\framebox{\textbf{Goal:} Given $\epsilon > 0$, compute a $\bv$-flow $\fv \in \R^{\vec{E}}$ that satisfies $\sE(\fv) \leq (1+\epsilon)\sE(\fv^*)$.}
\end{center}

An important special case is when $p=2$. In this case, $\fv^*$ is known as the \emph{electrical flow}, and the dual problem is equivalent to solving the linear system $\Lm \xv = \bv$, where $\Lm = \sum_{ij \in E} \frac{1}{r(i,j)}(\mathbf{e_i} - \mathbf{e_j})(\mathbf{e_i} - \mathbf{e_j})\trans$ is the \emph{Laplacian matrix} of $G$. Here, $\mathbf{e_i}$ is the $i$th standard unit basis vector. 

When $p=2$, \Cref{ohm} in \Cref{fact:kkt} is also known as \emph{Ohm's Law}. For a $\mathbf{b}$-flow $\fv$, there exist potentials $\xv$ such that $(\fv, \xv)$ satisfies Ohm's Law if and only if $\fv$ satisfies \emph{Kirchoff's Potential Law}  (KPL): KPL states that for every directed cycle $C$, $\sum_{(i,j) \in C} f(i,j)r(i,j) = 0$.

Our algorithms will make use of low-stretch spanning trees.  Given weights $\mathbf{w}$, the {\em stretch} of a spanning tree $T$ with respect to $\wv$ is defined as 
$$
    \st_T(G, \wv) = \sum_{(i, j) \in \vec{E}} \st_T((i, j), \wv)
    = \sum_{(i, j) \in \vec{E}} \frac{1}{w(i, j)}\sum_{(k, l) \in P(i,j)} w(k, l),
    $$
    where $P(i, j)$ is the unique path from $i$ to $j$ in $T$. 
We can find a spanning tree $T$ with total stretch $\st_T(G)=O(m\log n \log\log n)$ in $O(m \log n \log\log n)$ time \cite{AN12}.

We use the notation $\one$ to stand for the vector of all 1s, and $\one_X$ to be the characteristic vector of a set $X$ that has 1s in the entries corresponding to the elements of $X$ and 0s elsewhere. 

For a vector $\vv$, $\norm{\vv}_\infty$ denotes $\max_i \abs{\vv(i)}$ and $\norm{\vv}_{-\infty}$ denotes $\min_i \abs{\vv(i)}$.

%% file: laplacian.tex
\section{A Cut-Toggling Algorithm for Solving Laplacian Linear Systems}
\label{sec:alg}
As a first step, we consider the case $p=2$, where the primal problem is that of finding an electrical flow and the dual is equivalent to solving the Laplacian linear system $\Lm \xv = \bv$. We present a cut-toggling algorithm for computing an approximate solution to $\Lm\xv = \bv$, and also an approximate minimum-energy $\bv$-flow. The goal of this section is to show that the cut-toggling algorithm converges in a near-linear number of iterations, and that each iteration runs in linear time. Later in \Cref{sec:speedup}, we will show how to speed up the algorithm to an almost-linear total running time. Since our algorithm is dual to the cycle-toggling algorithm of Kelner et al. \cite{KOSZ13} (which we call \OALGNAME\ in this paper), we will begin by describing the \OALGNAME\ algorithm.

The  \OALGNAME\ algorithm  works by maintaining a feasible $\bv$-flow $\fv$, and iteratively updates $\fv$  along cycles to satisfy Kirchkoff's Potential Law on the cycle.  It starts by choosing a spanning tree $T$ that has low stretch, and computes a $\mathbf{b}$-flow $\mathbf{f}^0$ that uses only edges in the tree $T$.  Then for a number of iterations $K$ that depends on the stretch of the tree, it chooses a non-tree edge $(i,j) \in E-T$ according to a probability distribution, and for the fundamental cycle closed by adding edge $(i,j)$ to $T$, it modifies the flow $\mathbf{f}$ so that Kirchoff's Potential Law is satisfied on the cycle.  The probability $P_{ij}$ that edge $(i,j)$ gets chosen is proportional to the total resistance around the cycle closed by $(i,j)$ divided by $r(i,j)$. Given the tree $T$  with root $r$ and the current flow $\mathbf{f}^t$ in iteration $t$, there is a standard way to define a set of potentials $\mathbf{x}^t$ (called the {\em tree-induced} or {\em tree-defined} potentials):
set $x(r)$ to 0, and $x(k)$ to the sum of $f(i,j)r(i,j)$ on the path in $T$ from $k$ to $r$.
We summarize \OALGNAME\  in Algorithm \ref{KOSZalg}.

\begin{algorithm}[t]
	Compute a  tree $T$ with low stretch with respect to resistances $\mathbf{r}$\;
	Find flow $\mathbf{f}^0$ in $T$ satisfying supplies $\mathbf{b}$\;
	Let $\mathbf{x}^0$ be tree-defined potentials for $\mathbf{f}^{0}$ with respect to tree $T$\;
	\For{$t \gets 1$ \KwTo $K$}{
	 \hangindent=3em
		Pick an $(i,j) \in E-T$ with probability $P_{ij}$;
		Update $\mathbf{f}^{t-1}$ to satisfy KPL on the fundamental cycle corresponding to $(i,j)$\;
		Let $\mathbf{f}^t$ be resulting flow\;
		Let $\mathbf{x}^t$ be tree-defined potentials for $\mathbf{f}^t$\;
	}
	\Return{$\mathbf{f}^K, \mathbf{x}^K$}
	\caption{The \OALGNAME\  algorithm for solving $\mathbf{L}\mathbf{x} = \mathbf{b}$.}
	\label{KOSZalg}
\end{algorithm}

Our algorithm, which we will call \ALGNAME, works by maintaining a set of potentials $\xv$. It iteratively samples cuts in the graph, updating potentials on one side of the cut to satisfy flow conservation across that cut.  Following \OALGNAME, we choose a spanning tree $T$ of low stretch.  Then for a number of iterations $K$ that depends on the stretch of tree $T$, we repeatedly sample a fundamental cut from the spanning tree (i.e.\ a cut induced by removing one of the tree edges). We update all of the potentials on one side of the cut by an amount $\Delta$ so that the amount of flow crossing the cut via Ohm's Law is what is required by the supply vector. We summarize \ALGNAME\ in Algorithm \ref{ouralg}. The main result of this section is a bound on the iteration complexity of \ALGNAME.

\dualkosz*

Next we give the algorithm in somewhat more detail.  Let $R(C) = (\sum_{(k,l) \in \delta(C)} \frac{1}{r(k,l)})^{-1}$. Note that $R(C)$ has units of resistance. For every tree edge $(i,j)$, let  $C(i, j)$ be the set of vertices on one side of the fundamental cut defined by $(i, j)$, such that $i \in C(i, j)$ and $j \not\in C(i, j)$. We set up a probability distribution $P_{ij}$ on edges $(i,j)$ in the spanning tree $T$, where $P_{ij} \propto \frac{r(i,j)}{R(C(i,j))}$.  We initialize potentials $x^0(i)$ to 0 for all nodes $i \in V$. In each iteration, we sample edge $(i,j) \in T$ according to the probabilities $P_{ij}$.    Let $b(C) =\mathbf{b}\trans \one_C$ be the total supply of the nodes in $C$. Note that $b(C)$ is also the amount of flow that should be flowing out of $C$ in any feasible $\mathbf{b}$-flow. 

Let $f^{t}(C)$ be the total amount of flow going out of $C$ in the flow induced by $\mathbf{x}^{t}$. That is,
        $$f^{t}(C) = \sum_{\substack{ij \in E \\ i \in C, \, j \not\in C}} \frac{{x}^t(i) - {x}^t(j)}{r(i,j)}.$$
Note that $f^t(C)$ can be positive or negative. In any feasible $\mathbf{b}$-flow, the amount of flow leaving $C$ should be equal to $\mathbf{b}\trans \one_C = b(C)$. Hence, we define 
$\Delta^t = (b(C) - f^{t}(C))\cdot R(C).$
Observe that $\Delta^t$ is precisely the quantity by which we need to increase the potentials of every node in $C$ so that flow conservation is satisfied on $\delta(C)$. We then update the potentials, so that 
        $$
        p^{t+1}(v) = 
        \begin{cases}
        p^{t}(v) +\Delta^t, &\text{if $v \in C$,} \\
        p^{t}(v), &\text{if $v \not\in C$.}
        \end{cases}
        $$
Once we have completed $K$ iterations, we return the final potentials $\mathbf{x}^{K}$. The last step is to convert $\xv^K$ to a feasible flow by taking a \emph{tree-defined flow} with respect to $T$: $f^K(i,j) = \frac{x^K(i)-x^K(j)}{r(i,j)}$ on all non-tree edges, and $\fv^K$ routes the unique flow on $T$ to make $\fv^K$ a feasible $\bv$-flow.

\begin{algorithm}[t]
	Compute a spanning tree $T$ with low stretch with respect to resistances $\rv$\;
	Set $x^0(i) = 0$ for all $i \in V$\;
	\For{$t \gets 1$ \KwTo $K$}{
	 \hangindent=3em
		Pick an edge $(i,j) \in T$ with probability $P_{ij} \propto \frac{r(i,j)}{R(C(i,j)}$ and let $C = C(i,j)$\;
		$\Delta^t \gets (b(C)-f^t(C))\cdot R(C)$\;
        $
        x^{t+1}(v) \gets 
        \begin{cases}
        x^{t}(v) +\Delta^t, &\text{if $v \in C$,} \\
        x^{t}(v), &\text{if $v \not\in C$.}
        \end{cases}
        $
	}
	Let $\fv^K$ be the tree-defined flow with respect to $\xv^K$ and $T$\;
	\Return{$\mathbf{x}^K$, $\fv^K$}
	\caption{Algorithm \ALGNAME\ for solving $\mathbf{L}\mathbf{x} = \mathbf{b}$.}
	\label{ouralg}
\end{algorithm}

% Our main theorem is the following. We define $$\tau := \sum_{(i, j) \in T} \frac{r(i, j)}{R(C(i, j))}.$$
% \begin{theorem}
% After $K = \tau \ln(\frac1\epsilon)$ iterations, the algorithm returns $\mathbf{x} \in \R^V$ such that $\E\norm{\mathbf{x}^* - \mathbf{x}}_\mathbf{L}^2 \leq \epsilon\norm{\mathbf{x}^*}_\mathbf{L}^2$. 
% \end{theorem}

\subsection{Analysis of \ALGNAME}
\label{sec:analysis}
For $p=2$, recall that $\sE(\fv) = \frac12\sum_{e \in E}\abs{f(e)}^2$ and $\sB(\xv) = \bv^T\xv - \frac12\xv^T\Lm\xv$. 
By convex duality, we have $\sB(\xv) \leq \sE(\fv)$ for any $\xv\in \R^V$ and $\bv$-flow $\fv$. Moreover, $\xv$ maximizes $\sB(\xv)$ if and only if $\Lm\xv = \bv$. (See e.g. \cite[Lemma 8.9]{Williamson19}).
% \begin{fact}[{\cite[Lemma 8.9]{Williamson19}}]
% For a vector $\mathbf{x} \in \R^V$, $\mathbf{x}$ satisfies $\mathbf{L}\mathbf{x} = \mathbf{b}$ if and only if $\mathbf{x}$ maximizes $\sB(\mathbf{x})$ over all $\mathbf{x} \in \R^V$, and at its maximum $\sB(\mathbf{x})=\sE(\mathbf{f})$ for the electrical $\mathbf{b}$-flow $\mathbf{f}$.
% \end{fact}
Thus solving the Laplacian system $\mathbf{L}\mathbf{x} = \mathbf{b}$ is equivalent to finding a vector of potentials that maximizes the dual objective. 
% In the next few subsections, we prove the lemmas that form the bulk of the analysis. 
In what follows, we present the lemmas that form the bulk of the analysis. Their proofs are in the Appendix.
These lemmas (and their proofs) are similar to their counterparts in \cite{KOSZ13}, because everything that appears here is dual to what appears there.

First, we show that each iteration of the algorithm increases $\sB(\xv)$. 

\begin{restatable}{lemma}{enerincr}
\label{lem:energy_increase}
Let $\mathbf{x} \in \R^V$ be a vector of potentials and let $C \subset V$. Let $\mathbf{x}'$ be the potentials obtained from $\mathbf{x}$ as in the algorithm (that is, by adding $\Delta$ to the potential of every vertex in $C$ so that flow conservation is satisfied across $\delta(C)$). Then 
$$\sB(\mathbf{x}') - \sB(\mathbf{x}) = \frac{\Delta^2}{2R(C)}.$$
\end{restatable}

 The second ingredient in the analysis is to introduce an upper bound on how large the potential bound $\sB(\xv)$ can become. This will allow us to bound the number of iterations the algorithm takes.  

\begin{defn}[Gap]
Let $\mathbf{f}$ be a feasible $\mathbf{b}$-flow and let $\mathbf{x}$ be any vertex potentials. Define 
$$\gap(\mathbf{f}, \mathbf{x}) := \sE(\mathbf{f}) - \sB(\mathbf{x}) = \frac12\sum_{e \in E} r(e)f(e)^2 - (\mathbf{b}\trans \mathbf{x} -\frac12 \mathbf{x}\trans \mathbf{L}\mathbf{x}).$$
\end{defn}

This same notion of a gap was introduced in the analysis of the Kelner et al.\ algorithm, and was also used to bound the number of iterations of the algorithm.

The electrical flow $\mathbf{f}^*$ minimizes $\sE(\mathbf{f})$ over all $\mathbf{b}$-flows $\mathbf{f}$, and the corresponding vertex potentials $\mathbf{x}^*$ maximize $\sB(\mathbf{x}^*)$ over all vertex potentials $\mathbf{x}$. Moreover, $\sE(\mathbf{f}^*) = \sB(\mathbf{x}^*)$. Therefore, for any feasible flow $\mathbf{f}$, $\gap(\mathbf{f}, \mathbf{x})$ is an upper bound on optimality:
$$\gap(\mathbf{f}, \mathbf{x}) \geq \sB(\mathbf{x}^*) - \sB(\mathbf{x}).$$
The lemma below gives us another way to write $\gap(\mathbf{f}, \mathbf{x})$, and will be useful to us later. This relation is shown in Kelner et al.\  \cite[Lemma 4.4]{KOSZ13}, but we restate it here and reprove it in the Appendix for completeness.

\begin{restatable}{lemma}{gapp}
Another way to write $\gap(\mathbf{f}, \mathbf{x})$ is 
$\gap(\mathbf{f}, \mathbf{x}) = \frac12\sum_{(i, j) \in \vec{E}} r(i, j) \left(f(i, j) - \frac{x(i)-x(j)}{r(i, j)}\right)^2.$
\end{restatable}

The analysis of Kelner et al. \cite{KOSZ13} relies on measuring progress in terms of the above-defined duality gap between primal flow energy and dual potential bound. The high-level idea of the analysis is that one can show that the duality gap decreases by a constant factor each iteration, which implies a linear convergence rate. In the analysis of their algorithm, they maintain a feasible $\mathbf{b}$-flow $\mathbf{f}$ at each iteration, and measure $\gap(\mathbf{f}, \mathbf{x})$ against corresponding tree-defined potentials $\mathbf{x}$.

One difference between their algorithm and ours is that we do not maintain a feasible $\mathbf{b}$-flow at each iteration. However, for $\gap(\mathbf{f}, \mathbf{x})$ to be a valid bound on distance to optimality, we need $\mathbf{f}$ to be a feasible $\mathbf{b}$-flow. To this end, we introduce the definition of ``tree-defined flow'' below.

\begin{defn}[Tree-defined flow]
\label{def:tree_def_flow}
Let $T$ be a spanning tree, $\mathbf{x} \in \R^V$ vertex potentials, and $\mathbf{b} \in \R^V$ satisfying $\one\trans \mathbf{b} = 0$ be a supply vector. The \textbf{tree-defined flow} with respect to $T$, $\mathbf{x}$ and $\mathbf{b}$ is the flow $\mathbf{f}_{T, \xv}$ defined by 
$$\mathbf{f}_{T, \xv}(i, j) = 
\frac{x(i) - x(j)}{r(i, j)} \quad \text{if $(i, j) \not\in T$},
$$
and for $(i, j) \in T$, $f_{T,\xv}(i, j)$ is the unique value such that the resulting $\mathbf{f}_{T, \xv}$ is a feasible $\mathbf{b}$-flow. That is, for $(i, j) \in T$, if $C = C(i, j)$ is the fundamental cut defined by $(i, j)$ and $b(C) = \mathbf{b}\trans \one_C$ is the amount of flow that should be flowing out of $C$ in a feasible $\mathbf{b}$-flow, then  
\begin{align*}
f_{T,\xv}(i, j)  = b(C) - \sum_{\substack{k \in C,\, l \not\in C \\ kl \in E - ij}} f_{T,\xv}(k, l) 
= b(C) - \sum_{\substack{k \in C,\, l \not\in C \\ kl \in E - ij}} \frac{x(k) - x(l)}{r(k, l)}.
\end{align*}
% [Notational note: The reason I'm summing over $\{kl \in E - ij: k \in C, l \not\in C\}$ instead summing over something like $\{(k, l) \in \delta(C) - (i, j)\}$ is because for each directed edge $(u, v) \in \delta(C)$, I really want to take the potential on the vertex in $C$ minus the potential on the vertex not in $C$, regardless of the direction of the edge. In other words, if $u \in C$ and $v \not\in C$, I want to do $p_u - p_v$, and if $v \in C$ and $u \not\in C$, then I want to do $p_v - p_u$. Is there a better way to write this?]
In other words, $\mathbf{f}_{T,\xv}$ is a potential-defined flow outside of the tree $T$, and routes the unique flow on $T$ to make it a feasible $\mathbf{b}$-flow. 
\end{defn}
The below lemma expresses $\gap(\fv_{T, \xv}, \mathbf{x})$ in a nice way. 
\begin{restatable}{lemma}{gaplem}
\label{lem:gap}
Let $T$ be a spanning tree, $\mathbf{x}$ vertex potentials, and $\mathbf{b}$ a supply vector. Let $\fv_{T, \xv}$ be the associated tree-defined flow. Then 
$$\gap(\fv_{T, \xv}, \mathbf{x}) = \frac12 \sum_{(i, j) \in T} r(i, j) \cdot \frac{\Delta(C(i, j))^2}{R(C(i, j))^2}.$$
\end{restatable}

% \subsection{Sampling Edges from a Probability Distribution}
% \label{sec:prob_dist}
Suppose we have a probability distribution $(P_{ij}: (i, j) \in T)$ on the edges in $T$. If the algorithm samples an edge $(i, j) \in T$ from this distribution, then by \Cref{lem:energy_increase} the expected increase in the dual objective is
$$
\E[\sB(\mathbf{x}')] - \sB(\mathbf{x}) =\frac12 \sum_{(i,j) \in T} P_{ij}\cdot\Delta(C(i, j))^2/R(C(i, j)).
$$
We want to set the $P_{ij}$ to cancel terms appropriately so that the right-hand side is a multiple of the gap. Looking at Lemma \ref{lem:gap}, we see that an appropriate choice is to set
$$P_{ij} := \frac{1}{\tau} \cdot \frac{r(i, j)}{R(C(i, j))},$$
where $\tau := \sum_{(i, j) \in T} \frac{r(i, j)}{R(C(i, j))}$ is the normalizing constant. For this choice of probabilities, we have
$$\E[\sB(\mathbf{x}')] - \sB(\mathbf{x}) 
= \frac{1}{2\tau} \sum_{(i, j) \in T} r(i,j)\cdot \frac{\Delta(C(i, j))^2}{R(C(i, j))^2} 
= \frac{1}{\tau}\gap(\fv_{T, \xv}, \mathbf{x}),$$
where $\fv_{T, \xv}$ is the tree-defined flow associated with potentials $\mathbf{x}$. As a consequence of this, we have the following.

\begin{restatable}{lemma}{gapdecr}
\label{lem:gap_decreases}
If each iteration of the algorithm samples an edge $(i, j) \in T$ according to the probabilities $P_{ij} = \frac{1}{\tau} \cdot \frac{r(i, j)}{R(C(i, j))}$, then we have
$$\sB(\mathbf{x}^*) - \E[\sB(\mathbf{x}^{t+1})]
\leq
\left(1 - \frac1\tau\right)
\left(\sB(\mathbf{x}^*) - \sB(\mathbf{x}^t)\right).$$
% In other words, in each iteration decreases the distance to optimum by a multiplicative factor of $(1- \frac{1}{\tau})$. 
\end{restatable}

\begin{restatable}{corollary}{finalgap}
\label{cor:final_gap}
After $K = \tau\ln(\frac{1}{\epsilon})$ iterations, we have $\sB(\mathbf{x}^*) - \E[\sB(\mathbf{x}^K)] \leq \epsilon \cdot \sB(\mathbf{x}^*)$.
\end{restatable}

% \subsection{Overall Running Time}
% \label{sec:convergence}
We now use the previous lemmas to bound the number of iterations \ALGNAME\ takes. Lemma \ref{lem:gap_decreases} shows that the quantity $\sB(\mathbf{x}^*) - \sB(\mathbf{x}^t)$ decreases multiplicatively by $(1 - \frac1\tau)$ each iteration. Thus, a smaller value of $\tau$ gives faster progress. We prove in the Appendix that $\tau = \st_T(G, \rv)$, which is why the algorithm chooses $T$ to be a low-stretch spanning tree.

    We also need to argue that rounding $\xv^K$ to $\fv^K$ via a tree-defined flow preserves approximate optimality. \Cref{lem:rounding_error} in the Appendix shows this: For any distribution over $\xv$ such that $\E_{\xv}[\sB(\xv)] \geq (1-\frac{\epsilon}{\tau})\sB(\xv^*)$, we have $\E_{\xv}[\sE(\fv_{T, \xv})] \leq (1+\epsilon)\sE(\fv^*)$.
    Combining everything together, we conclude: 
    \dualkosz*
    
    % \begin{proof}[Proof of \Cref{thm:dual_kosz}]
    %  By Corollary \ref{cor:final_gap}, after $K=\tau\ln(\frac{\tau}{\epsilon})$ iterations, the algorithm returns potentials  $\xv^K$ such that $\sB(\mathbf{x}^*) - \E[\sB(\mathbf{x}^K)] \leq \frac{\epsilon}{\tau} \cdot \sB(\mathbf{x}^*)$. Combining with Lemma \ref{lem:energy_to_potential}, we get that $\E \norm{\mathbf{x}^* - \mathbf{x}^K}_\mathbf{L}^2 \leq \frac{\epsilon}{\tau}\norm{\mathbf{x}^*}_\mathbf{L}^2$. Finally, \Cref{lem:rounding_error} gives $\E\left[\sE(\fv^K)\right] \leq (1+\epsilon)\sE(\fv^*)$.  
     
    % \end{proof}
    
    We end this section with a na\"{i}ve bound on the total running time of \ALGNAME. 
    
    \begin{restatable}{lemma}{runtime}
    \label{lem:runtime}
    \ALGNAME\ can be implemented to run in $\wtd{O}(mn\log\frac{1}{\epsilon})$ time. 
    \end{restatable}
    
     In Section \ref{sec:ds}, we argue that given a natural abstraction of the data structure problem we use in computing $f(C)$ and updating potentials, it appears unlikely that we can implement each iteration in $o(n^{1-\epsilon})$ time, if each iteration is to be processed one-by-one in an online fashion.  In Section \ref{sec:speedup}, we show how to overcome this data structure lower bound by taking advantage of the fact that the sequence of updates that we perform can be generated in advance.

%% file: pnorm.tex
\section{Iteration Complexity of Cycle-toggling and Cut-toggling for Minimum $p$-norm Flows}
\label{sec:pnorm}

Before we show how to speed up the cut-toggling algorithm for $p=2$, we first demonstrate how cut-toggling and cycle-toggling algorithms can be applied to solve a minimum $p$-norm flow problem for a general $1 < p < \infty$.  The goal of this section is to derive bounds on the \emph{iteration complexity} of the cut-toggling and cycle-toggling algorithms for solving the minimum $p$-norm flow problem. We will show that for all $p \geq 2$, the cycle-toggling algorithm can be used  to find an approximately optimal flow. On the other hand, the cut-toggling algorithm can be used to solve the \emph{dual} of $p$-norm flow problem for all $1 < p \leq 2$, and  one can convert an approximately optimal dual solution to an approximately optimal primal flow.  These iteration bounds are summarized in Theorems \ref{thm:cycle_iter} and \ref{thm:cut_iter}.

% Recall that in the $p=2$ case, the cycle-toggling algorithm of \cite{KOSZ13} works by maintaining a feasible flow, and gradually enforces Kirchkoff's Potential Law (KPL) by sampling a fundamental cycle in each iteration and adding some scalar to the flow on every edge in the cycle. On the other hand, the cut-toggling algorithm from Section \ref{sec:alg} works by 
The cycle-toggling and cut-toggling algorithms for $p$-norms are described in Algorithms \ref{alg:cycle_general} and \ref{alg:cut_general}, respectively. They are generalizations of the algorithm of \cite{KOSZ13} and the \ALGNAME\ algorithm  described in Section \ref{sec:alg}  to the $p$-norm setting. Observe that the main difference is that the low-stretch tree $T$ now changes dynamically in each iteration. Also, the update step is more involved. Note that when $p=2$, Algorithms  \ref{alg:cycle_general} and \ref{alg:cut_general} reduce to \OALGNAME\ and \ALGNAME\  respectively. (In particular, one can use the same spanning tree for all iterations.)

\begin{algorithm}[ht!]
\begin{algorithmic}[1]
% 	\STATE Initialize $\fv^0$ to be a $\frac{1}{p}$-approximate electrical flow with respect to resistances $\rv$. (That is, $\sE(\fv^0) $
    \STATE Pick any spanning tree $T^0$, and initialize $\fv^0$ to be the $\bv$-flow using only edges of $T^0$. 
	\FOR{$t \gets 1$ \KwTo $K$}
	    \STATE Compute a low-stretch spanning tree $T^t$ with respect to weights $r(e){\abs{f^{t-1}(e)}^{p-2}}$.
	    \STATE Sample an edge $(i,j) \in E - T^t$  with probability  $P^t_{ij}$, where
	    $$P^t_{ij}
	    \propto 
	   \max\left\{ {p2^{2p-1}}\frac{\sum_{e \in C(i,j)}r(e) \abs{f^{t-1}(e)}^{p-2}}{r(i,j)\abs{f^{t-1}(i,j)}^{p-2}}, \,\left({p2^{2p-1}} \cdot \frac{\sum_{e \in C(i,j)} r(e)}{r(i,j)}\right)^{\frac{1}{p-1}}\right\}.
	    $$
	    Here, $C(i,j)$ is the fundamental cycle of $T^t$ defined by $(i,j)$.
	    \STATE Let $C^t = C(i,j)$, and let $\vec{C}^t$ be an arbitrary orientation of $C^t$.
	    \STATE Solve the equation $\sum_{(k,l) \in \vec{C}^t} r(i,j)(f^{t-1}(i,j) + \Delta^t)\abs{f^{t-1}(i,j) + \Delta^t}^{p-2}= 0$ for $\Delta^t$.
	    \STATE For all $(k,l) \in \vec{E}$, let
	    $$f^t(k,l)
	    =
	    \begin{cases}
	    f^{t-1}(k,l), &\text{if $kl \notin C^t$} \\
	    f^{t-1}(k,l) + \Delta^t, &\text{if $(k,l) \in \vec{C}^t$} \\
	    f^{t-1}(k,l) - \Delta^t, &\text{if $(l,k) \in \vec{C}^t$}
	    \end{cases}
	    $$ 
	    be the new flow.
	\ENDFOR
	\STATE \Return{$\fv^K$}
	\caption{Cycle-toggling for $p$-norm flow.}
	\label{alg:cycle_general}
\end{algorithmic}
\end{algorithm}

\begin{algorithm}[ht!]
\begin{algorithmic}[1]
    \STATE Let $q := \frac{p}{p-1}$, and let $w(e) = r(e)^{-\frac{1}{p-1}}$.
% 	\STATE Compute a low-stretch spanning tree $T^0$ with respect to weights $\frac{1}{w(i,j)}$.
	\STATE Set $x^0(i) = 0$ for all $i \in V$.
	\FOR{$t \gets 1$ \KwTo $K$}
	    \STATE Compute a low-stretch spanning tree $T^t$ with respect to weights $\frac{\abs{x^{t-1}(i) - x^{t-1}(j)}^{2-q}}{w(i,j)}$ for $ij \in E$.
	    \STATE Sample an edge $(i,j) \in T^t$  with probability  $P^t_{ij}$, where
	    $$
	    P^t_{ij}
	    \propto
	    \max\left\{q2^{2q-1} \frac{\sum_{ij \in \delta(C(i,j))} w(i,j)\abs{x(i) - x(j)}^{q-2}}{w(u,v)\abs{x(u) - x(v)}^{q-2}}, \;\left(q2^{2q-1}\frac{\sum_{ij \in \delta(C(i,j))}w(i,j)}{w(u,v)}\right)^{\frac{1}{q-1}} \right\}.
	    $$
	    Here, $C(i,j)$ is the fundamental cut of $T^t$ defined by $(i,j)$.
	    \STATE Let $C^t = C(i,j)$.
	    \STATE Solve $\sum_{k\in C^t, l\not\in C^t, kl \in E } w(k,l)(x^{t-1}(k) - x^{t-1}(l) +\Delta^t)\abs{x^{t-1}(k)-x^{t-1}(l)+\Delta^t}^{q-2} = b(C)$ for $\Delta^t$.
	    \STATE For all $i \in V$, let
	    $$
	    x^t(i) =
	    \begin{cases}
	    x^{t-1}(i), &\text{if $i \not\in C^t$} \\
	    x^{t-1}(i) + \Delta^t, &\text{if $i \in C^t$}
	    \end{cases}
	    $$
	   % \STATE Let $\fv^t$ be tree-defined flow for $\xv^t$ with respect to tree $T^t$.
	\ENDFOR
% 	\STATE \Return{$\fv^K$, $\xv^K$}
	\STATE \Return{$\xv^K$}
	\caption{Cut-toggling for the dual of $p$-norm flow.}
	\label{alg:cut_general}
\end{algorithmic}
\end{algorithm}

%The minimum $p$-norm flow problem and its dual are defined as follows (see Appendix \ref{app:dual_deriv} for the details of how the dual is derived): 
%\begin{align*}
%(P) \quad &\min  \quad \frac1p\sum_{(i,j)\in\vec{E}} r(i,j)\abs{f(i,j)}^p  
%\qquad\qquad (D)  &\max \quad \bv^T\xv - \left(1-\frac1p\right)\sum_{(i,j)\in\vec{E}} \left(\frac{\abs{x(i) - x(j)}^p}{r(i,j)}\right)^{\frac{1}{p-1}} \\
%&\text{s.t.} \quad \Am\fv = \bv 
%&\text{s.t.} \quad \xv \in \R^{V}~\text{unconstrained} 
%\end{align*}
%For convenience, let $q = \frac{p}{p-1}$ and $w(i,j) = r(i,j)^{-\frac{1}{p-1}}$, so that the dual objective can equivalently be written as $\bv^T\xv - \frac1q \sum_{(i,j)\in \vec{E}} w(i,j)\abs{x(i)-x(j)}^q$. Throughout this section, we will let $\sE(\fv)$ denote the primal objective and $\sB(\xv)$ denote the dual objective. Also, define $R = \frac{\max_e r(e)}{\min_e r(e)}$.
We show next a bound on the number of iterations for these two algorithms to find a near-optimal solution.
\cycleiter*

\cutiter*

Although the cut-toggling algorithm returns a dual solution $\xv^K$, and \emph{not} a flow the lemma below shows that it is possible to convert an approximately optimal dual solution to an approximately optimal primal flow. To do this, we define an analogous notion of ``tree-defined flow" for general $p$-norms. This is a generalization of Definition \ref{def:tree_def_flow} and follows immediately from Fact~\ref{fact:kkt}. It then follows that the cut-toggling algorithm can be used to compute a near-optimal $p$-norm flow.

\begin{defn}[Tree-defined flow for $p$-norms]
Let $T$ be a spanning tree, $\xv \in \R^V$, and $\bv \in \R^V$ satisfying $\one^T \bv = 0$. The \textbf{tree-defined flow} with respect to $T$, $\mathbf{x}$ and $\mathbf{b}$ is the flow $\mathbf{f}_{T, \xv}$ defined by 
$${f}_{T, \xv}(i,j) = w(i,j)(x(i)-x(j))\abs{x(i)-x(j)}^{q-2} \quad \text{if $(i, j) \not\in T$},
$$
and for $(i, j) \in T$, $f_{T,\xv}(i, j)$ is the unique value such that the resulting flow is a feasible $\mathbf{b}$-flow.
\end{defn}

The lemma below shows that converting $\xv^K$ to a feasible flow via a tree-defined flow preserves approximate optimality up to polynomial factors. 
\begin{restatable}[Converting dual solution to flow]{lemma}{flowconvert}
\label{lem:flow_convert}
Suppose $\xv$ satisfies $\sB(\xv) \geq (1-\epsilon')\sB(\xv^*)$, where
% $$\ln\frac{1}{\epsilon'} = O\left(q
%     \ln\left(mnRq2^q\max\left\{\frac3\epsilon, 1\right\}\right)\right).$$
$$
\epsilon'
\leq
\left(
\frac{\min\{\epsilon/3, 1\}}{2n^4(mR)^{\frac{1}{p}} (q2^q)^{1+\frac1q}(nR)^{\frac{1}{p}}}
\right)^q.
$$
Then for any tree $T$, the tree-defined flow $\fv_{T, \xv}$ satisfies $ \sE(\fv_{T, \xv}) \leq (1+\epsilon)\sE(\fv^*)$. 
\end{restatable}
% \mh{ Can you write this a $\epsilon' = ..$?}

%% file: lowerbound.tex
\section{Lower Bound on the Per-Iteration Complexity of the Algorithm}
\label{sec:ds}
In the remainder of the paper, we focus on the case where $p=2$. Recall that each single iteration of \OALGNAME\  can be implemented in logarithmic time.
In this section we show that assuming the OMv conjecture (see below) each single iteration of \ALGNAME\  cannot be implemented in linear time. This implies that in order to speed up our algorithm we need to ``batch-up'' iterations, which is the approach we use in the next section. We first present a natural data structure, called \textit{TreeFlow} data structure, such that each iteration of the algorithm requires only two operations of the  \textit{TreeFlow} data structure and then prove that assuming the OMv conjecture~\cite{HKNS15} it is impossible to implement the \textit{TreeFlow} data structure such that each operation of the data structure takes $O(n^{1-\epsilon})$ time.
To simplify the reduction we reduce from a closely related problem called the \textit{Online Vector-Matrix-Vector Multiplication Problem (OuMv)}. %It was shown in~\cite{HKNS15} that the OMv problem and the OuMv problem have the same asymptotic time complexity. Thus, we state our presentation below in terms of OuMv instead of OMv.

\begin{defn}[Online Vector-Matrix-Vector Multiplication Problem]
We are given a positive integer $n$, and a Boolean $n\times n$ matrix $\mathbf{M}$. At each time step $t =1, \ldots, n$, we are shown a pair of Boolean vectors $(\mathbf{u_t}, \mathbf{v_t})$, each of length $n$. Our task is to output $\mathbf{u_t^\top} \mathbf{M} \mathbf{v_t}$ using Boolean matrix-vector operations. Specifically, ``addition" is replaced by the \textsf{OR} operation, so that $0+0 = 0$, and $0+1 = 1+0 = 1+1 = 1$. Hence, $\mathbf{u_t^\top} \mathbf{M} \mathbf{v_t}$ is always either 0 or 1. 
\end{defn}

The OMv conjecture implies that no algorithm for the OuMv problem can do substantially better than naively multiplying $\mathbf{u_t^\top} \mathbf{M} \mathbf{v_t}$ at time step $t$. Specifically, it say the following:

\begin{lemma}[\cite{HKNS15}]
\label{conj:omv}
Let $\epsilon > 0$ be any constant.
Assuming the OMv conjecture,
there is no algorithm for the online vector-matrix-vector multiplication problem that uses preprocessing time $O(n^{3-\epsilon})$ and takes total time $O(n^{3-\epsilon})$  with error probability at most $1/3$ in the word-RAM model with $O(\log n)$ bit words. 

% In particular, this implies that the time for a single vector-matrix-vector multiply, amortized over the entire sequence of operations, can be no less than $\Theta(n^2)$. 
\end{lemma}

Thus we will reduce the OuMv problem to the \textit{TreeFlow} data structure such that
computing $\mathbf{u_t^\top} \mathbf{M} \mathbf{v_t}$ requires two operations in the \textit{TreeFlow} data structure. The lower bound then follows from Lemma~\ref{conj:omv}.
%It works as follows.
%\begin{enumerate}
 %   \item
%    (1) We define a data structure, such that any implementation of the data structure with $O(\Delta)$ time per operation gives an implementation of \ALGNAME\ with $O(\Delta)$ time per iteration. 
%    (2) Assuming the OMv conjecture, we prove that no implementation of the \textit{TreeFlow} data structure exists where each operation takes $O(n^{1-\epsilon})$ time. 
%\end{enumerate}

\subsection{The \textit{TreeFlow}  Data Structure}
\label{subsec:ds}
The \textit{TreeFlow} data structure is given as input (1) an undirected graph $G = (V, E)$ with $n = |V|$, (2) a spanning tree $T$ of $G$ that is rooted at an arbitrary but fixed vertex, (3) a value $r(u,v)$ for each edge $(u,v) \in E$ (representing the resistance of $(u, v)$), and (4) a value $b(v)$ for each vertex $v \in V$ (representing the supply at $v$). The quantities $r(u, v)$ and $b(v)$ are given at the beginning and will remain unchanged throughout the operations. For any set $C \subset V$, let $S(C) := \sum_{v \in C} b(v)$.

Furthermore,
each vertex $v$ has a non-negative {value}, denoted $\mathsf{value}(v)$, which  can be seen as the ``potential'' of $v$. It is initially 0 and can be modified.
%One should think of $\mathsf{value}(v)$ as the potential of $v$.
For any set $C \subset V$  we define the \emph{flow out of} $C$ to be the quantity
$$f(C) := \sum_{(u,v) \in E, u \in C, v \not\in C} \left(\mathsf{value}(u) - \mathsf{value}(v)\right)/r(u,v).$$

The \textit{TreeFlow} data structure supports the following operations.
\begin{itemize}
    \item $\textsf{addvalue(\textbf{vertex} $v$, \textbf{real} $x$)}$: Add $x$ to the value of every vertex in the subtree of $T$ rooted at $v$. 
    % \item $\textsf{findvalue(\textbf{vertex} $v$)}$: Return the value of $v$. 
    \item $\textsf{findflow(\textbf{vertex} $v$)}$: Return $S(C) - f(C)$, where $C$ is the set of vertices in the subtree of $T$ rooted at $v$. 
\end{itemize}
The \emph{TreeFlow} data structure implements exactly the operations we require for each iteration of \ALGNAME: The \textsf{addvalue} operation allows us to update the potentials on a fundamental cut, and \textsf{findflow} computes $S(C) - f(C)$, thereby allowing us to compute $\Delta$ at each iteration.
Note that if all $b(v)$-values are zero, the \emph{TreeFlow} data structure simply returns $-f(C)$, which gives it its name.

We even show the lower bound  for a ``relaxed'' version defined as follows: In an \emph{$\alpha$-approximate \textit{TreeFlow} data structure} the
operation \textsf{addvalue} remains as above and the operation  \textsf{findflow($v$)} returns a value that is within a multiplicative
 factor $\alpha \ge 1$ (that can be a function of $n$) of the correct answer, i.e., a value between $(S(C) - f(C))/\alpha$ and $(S(C) - f(C))\cdot \alpha$.

% Initializing the data structure must take no more than $\widetilde{O}(m)$ time. 

% Furthermore, the operations \textsf{addvalue}, \textsf{findvalue}, and \textsf{findflow} must all run in $O(\log n)$ time. 

%\begin{claim}
%Any implementation of the algorithm such that each iteration of the algorithm takes $O(K)$ time %gives an implementation of the data structure such that each operation takes $O(K)$ time. 
%\end{claim}

%\begin{proof} (TODO)
%This is unclear. The confusing thing is that our algorithm always interleaves calls to %\textsf{addvalue} and \textsf{findflow}, whereas the data structure has no such restriction. Also, %the reduction in the next section has the data structure call \textsf{addvalue} many times %followed by calling \textsf{findflow} many times, which is not what the algorithm does. 
%\end{proof}

\subsection{The Reduction}
In this section we  show that any exact or approximate \textit{TreeFlow} data structure takes near-linear time per operation assuming the OMv conjecture. The hardness of approximation is interesting, because it turns out that even an approximation of this quantity is sufficient to obtain an algorithm for $\mathbf{L}\mathbf{p} = \mathbf{b}$. (Albeit with a convergence rate that deteriorates with the approximation factor.)
%In the $\alpha$-approximate \textit{TreeFlow} data structure, the \textsf{addvalue} operations are executed as described above and each \textsf{findflow($v$)} operation returns a value that is within a multiplicative  factor $\alpha \ge 1$ of the correct answer, i.e., a value that is at between $(S(C) - f(C))/\alpha$ and $(S(C) - f(C))\cdot \alpha$. Note that in the following lemma, $\alpha$ does not need to a be a constant and can depend on $n$.

\begin{lemma}
Let $\epsilon > 0$ be any constant and let $\alpha \ge 1$ be any value.
Assuming the OMv conjecture, no implementation of the $\alpha$-approximate \textit{TreeFlow} data structure exists that uses preprocessing time $O(n^{3-\epsilon})$ and where the two operations \textsf{addvalue} and \textsf{findflow} both take $O(n^{1-\epsilon})$ time,
 such that over a polynomial number of operations the error probability is at most $1/3$ in the word-RAM model with $O(\log n)$ bit words.
\end{lemma}

\begin{proof}
Given an $n \times n$ Boolean matrix $\mathbf{M}$, we create the following \textit{TreeFlow} data structure.
The graph contains $2n+1$ nodes, namely a special node $x$,  one node $c_j$ for each column $j$ with
$1 \le j \le n$ and one node $d_i$ for each row $i$ with $1 \le i \le n$. There is  an edge $(d_i,c_j)$ if entry $M_{ij}$ = 1. Additionally, every node $c_j$ and every node 
$d_i$ has an edge to $x$. These edges are added to guarantee that the graph is connected. 
We set $r(c,d) = 1$ for every edge $(c,d)$ and
denote this graph by $G$. Let $T$ be the spanning tree of $G$ that is rooted at $x$ and consists of all the edges incident to $x$. Note that the subtree of $T$ root at any node $y \ne x$ consists of a single node $y$.

Now consider the sequence of $n$ vector pairs $(\mathbf{u}_t, \mathbf{v}_t)$ of the OuMv problem. Let $(\mathbf{u},\mathbf{v})$ be any such pair. We show below how to compute $\mathbf{u^\top} \mathbf{M} \mathbf{v}$ with $O(n)$ operations in the \textit{TreeFlow} data structure. 
Thus the sequence of $n$ vector pairs leads to $O(n^2)$ operations.
It then follows from the OMv conjecture and Lemma~\ref{conj:omv} that this sequence of $O(n^2)$ operations in the
\textit{TreeFlow} data structure cannot take time $O(n^{3-\epsilon})$, i.e., that it is not
possible that the complexity of both the $\textsf{addvalue}$ operation
and the $\textsf{findflow}$ operation are
$O(n^{1-\epsilon})$.

It remains to show how to compute $\mathbf{u^\top} \mathbf{M} \mathbf{v}$ with $O(n)$ operations in the \textit{TreeFlow} data structure. 
Initially the value $\mathsf{value}(v)$ of all nodes $v$ is 0. Let $Z$ be a large enough constant that we will specify later.

First, increase the value of $x$ to $Z$ by calling $\mathsf{addvalue}(x, Z)$. 
When given $(\mathbf{u},\mathbf{v})$ we increase the value of each row node $d_i$ with $u_i = 0$ by $Z$ by 
calling $\textsf{addvalue}(d_i, Z).$
Then, we perform the following 2 operations for each column node $c_j$ with $v_j = 1$:
$\textsf{addvalue}(c_j, Z)$
and  $\textsf{findflow}(c_j)$.
%on every row node $d_j$ with $u_j = 1$ and 
Afterwards we decrease the value again  for all nodes with value $Z$, so that  every node has value 0 again. (Alternatively, we could also increase the value of every node to $Z$, in which case we never execute an
$\textsf{addvalue}$ operation with negative second parameter.)

Note that $\mathbf{u^\top} \mathbf{M} \mathbf{v} = 1$ iff there exists an edge between a column node $c_j$ with $v_j = 1$ (i.e.  $\mathsf{value}(c_j)=Z$) and a row node $d_i$ with $u_i = 1$ (i.e.  $\mathsf{value}(d_i)=0$). 

We now show that $\mathbf{u^\top} \mathbf{M} \mathbf{v} = 1$ iff $f(c_j)>0$ for some column node $c_j$ with $v_j = 1$.  
(a) Assume first that  $\mathbf{u^\top} \mathbf{M} \mathbf{v} = 1$ and let $c^*$ denote a node $c_j$ and $d^*$ denote a node $d_i$ such that $v_j = 1$, $u_i = 1$ and $M_{ij} = 1.$ We will show that $f(c^*) > 0$.
Recall that the subtree of $c^*$ consists only of $c^*$. The edge $(c^*, d^*)$ leaves the subtree of $c^*$, contributing a positive amount to $f(c^*)$ because $\mathsf{value}(c^*)=Z$ and $\mathsf{value}(d^*)=0$.
All other edges leaving the subtree of $c^*$ contribute a non-negative amount to $f(c^*)$, since $\mathsf{value}(c^*)=Z$ and $\mathsf{value}(d_k)$ for other $k \ne i$ is either $Z$ or 0. Thus $f(c^*) > 0$.  
(b)
Assume next that  $\mathbf{u^\top} \mathbf{M} \mathbf{v} = 0$. In this case every node $c_j$ with $u_j = 1$ (and value $Z$) only has edges to nodes $d_i$ with $v_i = 0$ (and value $Z$). As before the subtree of every node $c_j$ only consists of $c_j$ and, thus, all edges leaving the subtree of $c_j$ contribute 0 to the flow out of the subtree. Thus, 
for every node $c_j$ with $u_j = 1$ we have $f(c_j) = 0$. 

To summarize  we have shown above that $\mathbf{u}\trans \mathbf{M} \mathbf{v} = 1$ iff $f(c_j) > 0$ for some column node $c_j$ with $\mathsf{value}(c_j) = Z$. We will now show how to use the results of the \textsf{findflow} queries returned by an $\alpha$-approximate \emph{TreeFlow} data structure to determine if $f(c_j)$ is positive or zero.  

Here is where we will choose the value of $Z$. The idea is to make $Z$ large enough so that if $f(c_j) > 0$, then $f(c_j)$ is very large. The idea is that this will allows us to distinguish between $f(c_j) = 0$ versus $f(c_j) > 0$, even if we only have access to an $\alpha$-approximation of $S(c_j) - f(c_j) = b(c_j) - f(c_j)$.

It will suffice to choose $Z$ large enough so that if $f(c_j) > 0$, then $f(c_j) > \max\{b(c_j), b(c_j)(1 - \alpha^2)\}$ (As $(1-\alpha^2) < 0$, the second term makes sense if $b(c_j) < 0$.) The value of $Z$ depends on $\alpha$, the supplies $\mathbf{b}$, and the resistances $\mathbf{r}$. For instance, it suffices to choose $Z > \norm{r}_\infty \norm{b}_\infty \alpha^2$. For this choice of $Z$, we have that if $f(c_j) > 0$ then (since it must have an edge to some $d_i$ with $\mathsf{value}(d_i) = 0$), 
$$f(c_j) \geq \frac{\mathsf{value}(c_j) - \mathsf{value}(d_i)}{r(c_j, d_i)} = \frac{Z-0}{r(c_j, d_i)} > \abs{b(c_j)}\alpha^2 > \max\{b(c_j),\, b(c_j)(1 - \alpha^2)\}.$$
Having chosen $Z$ this way, we have the following:
\begin{itemize}
    \item If $b(c_j) \geq 0$, then  $b(c_j) - f(c_j)$ is non-negative if $f(c_j) = 0$, and negative otherwise (because $f(c_j) > b(c_j)$ when $f(c_j) > 0$.) Any $\alpha$-approximation of $b(c_j) - f(c_j)$ allows us to correctly deduce the sign of $b(c_j) - f(c_j)$, hence also whether $b(c_j) -f(c_j) \ge 0$ or
    whether $b(c_j) - f(c_j) < 0$. From this we can deduce wheter $f(c_j) = 0$ or $f(c_j) > 0$.
    \item Suppose $b(c_j) < 0$. If $f(c_j) = 0$, the approximate data structure returns an answer in the interval $[b(c_j)\cdot \alpha,\,\frac{ b(c_j)}{\alpha}]$. If $f(c_j) > 0$, it returns an answer in the interval $[(b(c_j) - f(c_j))\cdot \alpha,\, \frac{b(c_j) - f(c_j)}{\alpha}]$. Note that 
    the left endpoint of the first interval is to the right of the right endpoint of the second interval as 
    $f(c_j) > b(c_j)\left(1 - \alpha^2\right)$ implies that $$ \implies b(c_j)\cdot \alpha > \frac{b(c_j) - f(c_j)}{\alpha}.$$
    Since the two intervals for $f(c_j) = 0$ and $f(c_j) > 0$ do not overlap, we can correctly distinguish the two cases using the approximate data structure. 
\end{itemize}

To summarize, each \textsf{findflow} query on $c_j$ allows us to determine if $f(c_j) > 0$ or $f(c_j) = 0$. If the flow is positive for some $c_j$, then the answer is $\mathbf{u^\top} \mathbf{M} \mathbf{v} = 1$, otherwise it is 0.
Note that it requires $O(n)$ operations in the \textit{TreeFlow} data structure to determine one $\mathbf{u^\top} \mathbf{M} \mathbf{v}$ value, which completes the proof.
\end{proof}

\begin{remark}
%We believe this reduction to OMv provides strong evidence against the possibility of implementing an iteration of the algorithm in nearly-linear time. 
Note that the proof can be modified to be more similar to the update sequence generated by \ALGNAME\ which alternates between \textsf{addvalue} and \textsf{findflow} operations by inserting
 after each \textsf{addvalue} operation  a \textsf{findflow} operation (whose answer might be ignored for the logic of the proof).
 %, leading to an alternating sequence of \textsf{addvalue} and \textsf{findflow} operations. This is interesting as \ALGNAME\ also alternates between \textsf{addvalue} and \textsf{findflow} operations.
 Note that the proof can also be adapted so that the values stored at the nodes are only increased, but this is not necessary for our application.
%
%
%\end{enumerate}
\end{remark}

%% file: batching.tex
\section{Speeding Up \ALGNAME}
\label{sec:speedup}
We now show how to surmount the OMv lower bound by taking advantage of the fact that the sequence of updates that \ALGNAME\  performs can be generated in advance. In \Cref{sec:batch}, we show that batching the updates yields a modification of the algorithm that runs in $\wtd{O}(m^{1.5})$ time. Then in \Cref{sec:recurse}, we use sparsification and recursion to further improve the runtime to $\wtd{O}(m^{1+\alpha})$ for any $\alpha > 0$. 

\subsection{A Faster Algorithm using Batching}
\label{sec:batch}
First, we show that it is possible to speed up the running time to $\wtd{O}(m^{1.5})$ time by batching the updates performed by \ALGNAME. In Lemma \ref{lem:runtime}, we showed that the algorithm can be implemented to run in time $\widetilde{O}(mn)$. (Here the tilde hides a factor of  $\log n \log\log n \log\frac{1}{\epsilon}$.) This running time essentially comes from $\wtd{O}(m)$ iterations, $O(n)$ time per iteration, and  $O(n^2)$ preprocessing time to compute the $H(C_1, C_2)$ table.
Recall that each iteration of \ALGNAME\ involves sampling a fundamental cut $C$ of the low-stretch spanning tree $T$ from a fixed probability distribution $P$, and then adding a constant to the potential of every vertex in $C$ so that the resulting potential-defined flow satisfies flow conservation across $C$. 

The main idea of batching is as follows. Denote the number of iterations by $K$ (which is $\widetilde{O}(m)$). 
Instead of sampling the fundamental cuts one at a time, consider sampling the next $l$ cuts that need to be updated for some $l \ll K$.  We can perform this sampling in advance because both the tree $T$ and the probability distribution over cuts of $T$ are fixed over the entire course of the algorithm. In each ``block" of size $l \ll K$, we contract all the edges of $T$ that do not correspond to one of the $l$ fundamental cuts to be updated. In this way, we work with a contracted tree of size $O(l)$ in each block (instead of the full tree, which has size $O(n)$). This makes the updates faster. However, the price we pay is that at the end of each block, we need to propagate the updates we made (which were on the contracted tree), back to the entire tree. We will show that by choosing $l = \sqrt{m}$, we can balance this tradeoff and get an improved running time of $\wtd{O}(m^{1.5})$. Pseudocode for \ALGNAME\ with batching is given in Algorithm \ref{alg:faster}. Note that the correctness of this algorithm follows directly from the correctness of \ALGNAME: Algorithm \ref{alg:faster} samples cuts from exactly the same distribution as \ALGNAME, and if we fix the same sequence of cuts to be used by both algorithms, then the output of the two algorithms is identical.

\begin{algorithm}[ht!]
\caption{\ALGNAME\ with batching}%$(T,\q_1,\q_2)$}
\begin{algorithmic}[1]
\label{alg:faster}
	\STATE Compute a  tree $T$ with low stretch with respect to resistances $\rv$.
	\STATE Compute $S(C)$ and $R(C)$ for all fundamental cuts $C$ of $T$.
	\STATE Set $\xv^0(i) = 0$ for all $i \in V$. Set $f(C) = 0$ for all fundamental cuts $C$ of $T$.
	\FOR{$t \gets 1$ \KwTo $\ceil{\frac{K}{l}}$}
	    \STATE Sample $l$ edges $(i_1,j_1), \ldots, (i_l, j_l)$ with replacement from $T$, according to the distribution $P$.
	    \STATE Contract all edges in $T$ that were not sampled in step 5. \\ Let $\wtd{G}$ be the resulting graph and $\wtd{T}$ be the resulting tree.
	    \STATE For each $1 \leq k\leq l$, let ${C}_k$ denote the fundamental cut in $T$ determined by edge $(i_k, j_k)$. Let $\wtd{C}_k$ denote the fundamental cut in $\wtd{T}$ determined by $(i_k, j_k)$.
	    \STATE $\yv(\tilde{v}) \gets 0$ for all $\tilde{v} \in V(\wtd{G})$.
	    \FOR{$k \gets 1$ \KwTo $t$}
	        \STATE Compute $\Delta_k = (S(C_k) - f(C_k))\cdot R(C_k)$. \COMMENT{Requires $f(C_k)$ to be already computed}
	        \STATE $\yv(\tilde{v})  \gets \yv(\tilde{v}) + \Delta_k$ for all $\tilde{v} \in \wtd{C}_k$.
	        \STATE Update values of $f(C_j)$ for all $j\in\{k+1,\ldots, t\}$.
	    \ENDFOR
	    \FORALL{$i \in V$}
	        \STATE Let $\tilde{v}(i)$ be the vertex in $\wtd{G}$ that $i$ was contracted to.
	         \STATE $\xv^t(i) \gets \xv^{t-1}(i) + \yv(\tilde{v}(i))$.
	    \ENDFOR
        \STATE Recompute $f(C)$ for all fundamental cuts $C$ of $T$.
	\ENDFOR
	\STATE Let $\fv^{\ceil{K/l}}$ be the tree-defined flow with respect to $\mathbf{x}^{\ceil{K/l}}$ and $T$.
	\STATE \Return{$\mathbf{x}^{\ceil{K/l}}$, $\fv^{\ceil{K/l}}$}
	\caption{\ALGNAME\ with batching.}
\end{algorithmic}
\end{algorithm}

\begin{theorem}
    The overall running time of \ALGNAME\ with batching is $O(m^{1.5}\log n \log\log n \log\frac{1}{\epsilon})$. This is achieved by choosing $l = \sqrt{m}$. 
\end{theorem}
\begin{proof}

% Suppose we have just sampled the next $l$ fundamental cuts to be updated. Call the cuts $C_1, \ldots, C_l$. Recall our convention that each cut $C_i$ is associated with a unique (directed) spanning tree edge $(u_i, v_i)$, such that $C_i$ is the vertex set of the subtree of $T$ rooted at $u_i$. (Recall also our convention that the spanning tree edges are directed toward the root.)

%See the below figure for an illustration. 

% \begin{figure}[H]
%     \centering
%     \includegraphics[scale=0.2]{contract.jpg}
%     \caption{Contracting the tree.}
%     \label{fig:my_label}
% \end{figure}

Consider a batch of $l$ updates. Note that the contracted tree $\wtd{T}$ has at most $l+1$ vertices. After contracting, we need to perform $l$ updates. This involves, for each $k \in \{1,2,\ldots, l\}$:
\begin{itemize}
    \item Computing $\Delta_k := (S(C_k) - f(C_k))\cdot R(C_k)$,
    \begin{itemize}
        \item This takes $O(1)$ time assuming $f(C_k)$ has already been computed.  (Recall that the values $S(C_k)$ and $R(C_k)$ are computed at the very beginning of the algorithm, which takes $O(m\log n)$ time.)
    \end{itemize}
    \item Adding $\Delta_k$ to $\yv(\tilde{v})$ for every $\tilde{v} \in \wtd{C}_k$.
    \begin{itemize}
        \item This takes $O(l)$ time, because the contracted tree has size $O(l)$. 
    \end{itemize}
    \item Updating the values $f(C_{k+1}), f(C_{k+2}), \ldots, f(C_{l})$ so they can be used in the later iterations of the inner loop. 
    \begin{itemize}
        \item If each $f(C_j)$ can be updated in $O(1)$ time, this takes $O(l)$ time. 
        \item To update each $f(C_j)$ in $O(1)$ time, we can precompute at the beginning of the block the $H(C_i, C_j)$ table for $i,j \in \{1,2,\ldots, l\}$, like we did before. The difference now is that we only need to compute the table for the cuts that will be updated in the block. There are $l$ such cuts, so the total time to compute the table is $O(l^2)$, again using Karger's method.
    \end{itemize}
\end{itemize}

At the end of each block, we propagate the updates we made on the contracted graph back to the original graph. This involves
\begin{itemize}
    \item Determining the new potential of each node in $G$.
    \begin{itemize}
        \item This takes $O(n)$ time, because one can simply iterate over all the nodes of $G$.
    \end{itemize}
    \item Determining the value of $f(C)$ for each fundamental cut determined by $T$. (Recall that our convention is that the edges of $T$ are directed towards the root, and that the fundamental cuts we consider are the vertex sets of the subtrees of $T$.)
    \begin{itemize}
        \item This can be done in $O(m)$ time using a dynamic program that works from the leaves to the root. First, we compute $f(C)$ at each leaf of the tree. Next, suppose we have are at a non-leaf node $v$, and let $C$ be the set of vertices in the subtree rooted at $v$. Suppose we have already computed $f(C_1), f(C_2), \ldots, f(C_k)$, where $C_1, \ldots, C_k$ are the proper subtrees of $v$. Then we can compute $f(C)$ as follows:
        $$f(C) = \sum_{i=1}^k f(C_i) + \sum_{w: vw \in E} \frac{p(v) - p(w)}{r(v, w)}.$$
        This sum correctly counts the flow leaving $C$. This is because any edge leaving $C$ is counted once. On the other hand, if an edge is between $C_i$ and $C_j$, then it is counted once in the $f(C_i)$ term, and once with the \emph{opposite sign} in the $f(C_j)$ term, so it zeros out. Similarly, if an edge is between $C_i$ and $v$, it also zeros out. 
        
        The running time of this dynamic program is $O(m)$, because the time taken at each node is proportional to its degree, and the sum of all the node degrees is equal to $2m$. 
    \end{itemize}
\end{itemize}

To summarize, there are $K$ iterations, divided into blocks of size $l$. In each block, we pay the following.
\begin{itemize}
    \item \textbf{Start of block:} $O(m)$ time to contract the tree, and $O(l^2)$ time to compute the $H(C, C')$ table for the cuts that will be updated in the block.
    \item \textbf{During the block:} $O(l)$ time per iteration. Since each block consists of $l$ iterations, this is $O(l^2)$ in total.
    \item \textbf{End of block:} $O(m)$ time to propagate the changes from the contracted tree to the original tree.
\end{itemize}
Hence, each block takes $O(m + l^2)$ time. Multiplying by the number of blocks, which is $K/l$, this gives a running time of $O(K(\frac{m}{l} + l))$. Choosing $l = \sqrt{m}$ to minimize this quantity, we get $O(K\sqrt{m})$. 

The final running time is therefore $O(K\sqrt{m})$ plus the preprocessing time. Note that we no longer need to spend $O(n^2)$ time to compute the $H(C, C')$ table at the start of the algorithm; this was replaced by $O(l^2)$ to compute a smaller table at the start of each block. Hence, preprocessing now just consists of finding a low-stretch spanning tree ($O(m\log n \log\log n$)), plus computing the values of $R(C)$ ($O(m\log n)$), and $f(C)$ ($O(n)$).

Thus, the preprocessing time is dominated by the time it takes to run the iterations. So, the total running time is now: $O(K\sqrt{m}) = O(m^{1.5}\log n \log\log n \log \frac{1}{\epsilon})$.

\end{proof}

%% file: sparsifyrecurse.tex
\subsection{A Still Faster Algorithm via Batching, Sparsification, and Recursion}
\label{sec:recurse}
We now show that we can further speed up the algorithm using sparsification and recursion. The goal is to show that we can we can obtain a running time of the form $O(A^{\frac{1}{\delta}}m^{1+\delta}(\log n)^{\frac{B}{\delta}}(\log \frac{1}{\epsilon})^{\frac{1}{\delta}})$ for any $\delta > 0$, where $A$ and $B$ are constants. 

Consider batching the iterations of the algorithm as follows. Pick a positive integer $d$, and repeat $K$ times:
\begin{itemize}
    \item Sample the next $d$ updates to be performed by the algorithm. These correspond to $d$ edges of the spanning tree $T$.
    \item Let $V_0, V_1, \ldots, V_d$ be the vertex sets that $T$ is partitioned into by the $d$ tree edges.
    \item Add $\Delta(i)$ to every vertex in $V_i$. We will choose the values $\Delta(0), \Delta(1), \ldots, \Delta(d)$ to greedily maximize the increase in the dual bound. 
\end{itemize}
Note that our original algorithm corresponds to the case when $d = 1$.  The lemma below quantifies the increase of the dual objective after one step of the above update.
\begin{lemma}
\label{lem:dual_incr}
Let $(V_0, \ldots, V_d)$ be a partition of $V$. 
Let $\xv \in \R^V$ be a vector of potentials, and let $\Delta = (\Delta(0), \ldots, \Delta(d))$ be any vector in $\R^{d+1}$. Let $\tilde{\xv}$ be obtained from $\xv$ by adding $\Delta(i)$ to the potential of every node in $V_i$. Then, the increase in the dual bound is given by the formula
$$\mathcal{B}(\wtd{\xv}) - \mathcal{B}(\xv) = {\bv}_H^T\Delta - \frac12\Delta^T{\Lm}_H\Delta,$$
where
\begin{itemize}
    \item $H$ is the contracted graph with vertices $V_0, V_1, \ldots, V_d$ and resistances $r(V_k, V_l) = \left(\sum_{ij \in \delta(V_k, V_l)} \frac{1}{r(i,j)}\right)^{-1}$,
    \item ${\Lm}_H$ is the Laplacian matrix of ${H}$, and
    \item ${b}_H(k) = b(V_k) - f(V_k)$ for $k = 0,1,\ldots, d$.
\end{itemize}
In particular, the choice of $\Delta$ that maximizes $\sB(\wtd{\xv}) - \sB(\xv)$ is given by the solution to ${\Lm}_H\Delta = \bv_H$.
\end{lemma}

\begin{proof}
We write the increase in the dual potential bound. Recall that $f(V_k)$ is the amount of flow leaving $V_k$ in the flow $f(i,j) = \frac{x(i)-x(j)}{r(i,j)}$. We let $f(V_k, V_l)$ be the amount of flow going from $V_k$ to $V_l$.
\begin{align*}
    &2\left(\sB(\wtd{\xv}) - \sB(\xv)\right) \\
    =\; &(2\bv^T\wtd{\xv} - \wtd{\xv}^T\Lm\wtd{\xv}) - (2\bv^T\xv- \xv^T\Lm\xv) \\
    =\; &2\sum_k b(V_k)\Delta(k) + \sum_{(i,j) \in \vec{E}} \frac{1}{r(i,j)}\left[(x(i)-x(j))^2 - (\tilde{x}(i) - \tilde{x}(j))^2\right] \\
    =\; & 2\sum_k b(V_k)\Delta(k) + \sum_{(i,j) \in \vec{E}} \frac{1}{r(i,j)}\left[(x(i)-x(j) + \tilde{x}(i) - \tilde{x}(j))(x(i)-x(j) - \tilde{x}(i) + \tilde{x}(j))\right] \\
    =\; & 2\sum_k b(V_k)\Delta(k) + \sum_{k<l}\sum_{(i,j)\in \delta(V_k,V_l)} \frac{1}{r(i,j)}\left[(2x(i)-2x(j) +\Delta(k)  -\Delta(l))(\Delta(l) - \Delta(k))\right] \\
    =\; & 2\sum_k b(V_k)\Delta(k) + 2\sum_{k<l}(\Delta(l) - \Delta(k))\sum_{(i,j)\in \delta(V_k,V_l)} \frac{1}{r(i,j)}(x(i)-x(j)) -\sum_{k<l}(\Delta(k) - \Delta(l))^2\sum_{(i,j) \in \delta(V_k, V_l)} \frac{1}{r(i,j)}  \\
    =\; & 2\sum_k b(V_k)\Delta(k) + 2\sum_{k<l}(\Delta(l) - \Delta(k))f(V_k, V_l) -\Delta^T\Lm_{H}\Delta  \\
    =\; & 2\sum_k b(V_k)\Delta(k) - 2\sum_{k}\Delta(k) f(V_k) -\Delta^T\Lm_{H}\Delta  \\
    =\; & 2\sum_k (b(V_k)-f(V_k))\Delta(k) -\Delta^T\Lm_{H}\Delta \\
    =\; & 2{\bv}_H^T\Delta -\Delta^T\Lm_{H}\Delta
\end{align*}
Note that this is a concave function of $\Delta$, because $\Lm_{H}$ is positive semidefinite. Therefore, maximizing this expression is equivalent to setting its gradient to 0. Taking its gradient and setting to 0 yields $\Lm_{H}\Delta = \bv_H$, as claimed. 

\end{proof}

\begin{remark}
Another interpretation of the $\Delta$ that maximizes $\sB(\wtd{\xv}) - \sB(\xv)$ in the Lemma above is as follows: $(\Delta(0), \ldots, \Delta(d))$ are the values such that if one adds $\Delta(i)$ to the potential of every vertex in $V_i$, the resulting potential-induced flow satisfies the flow constraints $f(V_k) = b(V_k)$ for all $k = 0, \ldots, d$. 
\end{remark}

% \subsection{Sparsify}
% So far, we have reduced the problem of solving $Lx = b$ to the problem of solving $K/d$ linear systems of the form $\wtd{\Lm}\Delta = \tilde{b}$, where $\wtd{\Lm}$ is the Laplacian matrix of a contracted graph $\tilde{G}$, which has $d+1$ vertices.

% Instead of solving the linear system $\wtd{\Lm}\Delta = \tilde{b}$, we will do the following:
% \begin{itemize}
%     \item Compute a sparsified version $L_1$ of $\wtd{\Lm}$,
%     \item Solve the equation $L_1\Delta_1 = \tilde{b}$,
%     \item Update $x \gets x + \Delta_1$. 
% \end{itemize}
% The following Lemma shows that as long as $L_1$ is a spectral approximation of $\wtd{\Lm}$, the increase in the dual bound when we add $\Delta_1$ to $x$ will approximate the increase when we add $\Delta$ to $x$.
% \begin{lemma}
% Suppose $x$
% \end{lemma}
% The main disadvantage of this approach is that we have to compute a sparsifier, which is a big black box that detracts from the simplicity of our original algorithm.

\subsection{The Sparsify and Recurse Algorithm}
Next we give the algorithm with sparsification and recursion in more detail. Observe that $d$ cut-toggling updates effectively break the spanning tree into $d+1$ components.  After contracting the components to get a graph $H$ with $d+1$ vertices, \Cref{lem:dual_incr} shows that solving the Laplacian system $\Lm_H \Delta = \bv_H$ gives the update that maximizes the increase in $\sB(\wtd{\xv}) - \sB(\xv)$  among all updates that increment the potential of all vertices in $V_i$ by the same amount.  In particular, the progress made by this update step is is at least as large as the progress made by the sequence of $d$ updates performed by the straightforward unbatched algorithm. 

A natural approach is to solve $\Lm_H \Delta = \bv$ recursively.
However, this by itself does not give an improved running time.
Instead, we will first spectrally sparsify $H$ to get a sparsified approximation $\wtd{\Lm}_H$ of $\Lm_H$, satisfying $(1-\gamma)\Lm_H \preceq \wtd{\Lm}_H \preceq (1+\gamma) \Lm_H$ for an appropriate constant $\gamma \in (0,1)$.
Such a matrix $\wtd{\Lm}_H$ is known as a $\gamma$-spectral sparsifier of $\Lm_H$.
We then call the algorithm recursively {on $H$} to solve $\wtd{\Lm}_H \wtd{\Delta} = \bv_H$.
Thus, a main task of the analysis is to bound the error incurred by solving the sparsified system instead of the exact one.
For the spectral sparsification, we use the original Spielman-Teng algorithm~\cite{ST11:journal} because it does not
require calling Laplacian solvers as a subroutine (e.g. \cite{BatsonSST13}). 
A variant of it with better failure probabilities is given in
Theorem 6.1~\cite{PS13}: one can find a sparsifier with $O(n\log^c n / \gamma^2)$ nonzero entries in $O(m\log^{c_1}n)$ time, with probability at least $1 - 1/n^2$.
Here, $c$ and $c_1$ are constants, and $m,n$ are the number of edges and vertices, respectively, in the graph before sparsifying. 
% For any $p \in (0, \frac12)$, the Spielman-Teng algorithm outputs a $\gamma$-spectral sparsifier $\wtd{\Lm}_H$ with $O(n\log^c(\frac{n}{p}) / \gamma^2)$ nonzero entries with probability at least $1 - p$, and its expected running time is  $O(m\log^{c_1}(\frac{n}{p}))$. Here, $c$ and $c_1$ are constants, and $m,n$ are the number of edges and vertices, respectively, in the graph before sparsifying. 
Pseudocode for \ALGNAME\ with batching, sparsification, and recursion is given in Algorithm \ref{alg:fastest}. 

\begin{algorithm}[H]
\caption{\ALGNAME\ with batching, sparsification, and recursion}%$(T,\q_1,\q_2)$}
% {\bf Input:} $G = (V, E)$ and resistances $\rv$, supply vector $\bv$, $n_0 \in \Z_+$, $\gamma \in (0,1)$, $\epsilon' \in (0,1)$.
\begin{algorithmic}[1]
\label{alg:fastest}
%\REQUIRE Time horizon $T$, user type distributions 
\STATE If $\abs{V} \leq n_0$, solve $\Lm_G\xv = \bv$ using Gaussian elimination and \textbf{return} $\xv$.
\\\COMMENT{$\Lm_G$ is the Laplacian matrix of $G$.}
\STATE Compute a low-stretch spanning tree $T$ of $G$. 
\STATE Initialize $\xv^0 = 0$.
\FORALL{$t=0$ to $K$}
\STATE Generate the next $d$ updates. These correspond to $d$ tree edges $e_1^t, \ldots, e_d^t \in T$.
\STATE Let $V^t_0, \ldots, V^t_d$ be the vertex sets of the connected components of $T - \{e_1^t, \ldots, e_d^t\}$. 
\STATE Contract $G$ to $H^t$ with $d+1$ vertices: Each $V_i^t$ is a vertex in $H^t$, and resistances in $H^t$ are 
$$r_{H^t}(V_k^t, V_l^t) = \left(\sum_{ij \in \delta(V_k, V_l)} \frac{1}{r(i,j)}\right)^{-1}.$$
        % \begin{itemize}
        %     \item The vertices of $H^t$ are $V^t_0, \ldots, V^t_d$.
        %     \item In $H^t$, the resistances are $r_{H^t}(V_k, V_l) = \left(\sum_{ij \in \delta(V_k, V_l)} \frac{1}{r(i,j)}\right)^{-1}$.
        %     \item Contracting $H^t$ and computing the new resistances takes $O(m)$ time.
        % \end{itemize}
\STATE Compute $\wtd{\Lm}_{H^t}$, a $\gamma$-spectral sparsifier of $\Lm_{H^t}$, where $\gamma \in (0,1)$ will be a parameter that we will determine later.
        % \begin{itemize}
        %     \item In other words, $\wtd{\Lm}_{H^t}$ satisfies $\Lm_{H^t}(1-\gamma) \preceq \wtd{\Lm}_{H^t} \preceq \Lm_{H^t}(1+\gamma)$.
        %     \item We will use the algorithm in Theorem 6.1 of this paper by Peng and Spielman: \cite{PS13}. It states that we find a sparsifier with $O(n\log^c n / \gamma^2)$ nonzero entries in $O(m\log^{c_1}n)$ time, with probability at least $1 - 1/n^2$. Here, $c$ and $c_1$ are constants, and $m,n$ are the number of edges and vertices, respectively, in the graph before sparsifying. 
        %     \item We use this algorithm because it works for weighted graphs and because it does not require calling a Laplacian solver as a subroutine. 
        % \end{itemize}
\STATE  Let $\bv^t_G = \bv - \Lm_G\xv^t$.
\STATE Let $\bv^t_{H^t} \in \R^{V(H^t)}$ be defined as follows:
    $b^t_{H^t}(V_i^t) = \sum_{u \in V_i^t} b^t_G(u) \quad \text{for all $i=0,1,\ldots,d$}.$
\STATE Call the algorithm recursively to solve the Laplacian system $\wtd{\Lm}_{H^t}\wtd{\Delta}^t = \bv^t_{H^t}$ for $\wtd{\Delta}^t$. This will \\ 
return an approximate solution $\wtd\Delta^t_{\epsilon'}$ that satisfies $\norm{\wtd\Delta^t_{\epsilon'} - \wtd\Delta^t}_{\wtd{\Lm}_{H^t}}^2 \leq \epsilon' \norm{\wtd{\Delta}^t}^2_{\wtd{\Lm}_{H^t}}$. Here, $\epsilon'$ is the \\
error parameter that we will input to the recursive call, to be determined later.
\STATE Update $\xv^t$ using $\tilde{\Delta}^t_{\epsilon'}$ to get the next iterate $\xv^{t+1}$. For every vertex $u \in V$, 
$$x^{t+1}(u) \gets x^t(u) + \tilde{\Delta}^t_{\epsilon'}(V^t_i),$$ where $V^t_i$ is the set in $V_0^t, \ldots, V_d^t$ such that $u \in V^t_i$. In other words, we update $\xv^t$ to $\xv^{t+1}$ by \\ 
adding $\tilde{\Delta}^t_{\epsilon'}(V^t_i)$ to the potential of every vertex in $V^t_i$. 
\STATE \textbf{If} $\sB(\xv^{t+1}) \leq \sB(\xv^t)$, revert $\xv^{t+1} \gets \xv^t$. 
\ENDFOR
\STATE Return $\xv^K$ and the corresponding tree-defined flow $\fv^K$.
\end{algorithmic}
\end{algorithm}

The base case of the recursive algorithm is when $\abs{V} \leq n_0$, where $n_0$ is a constant that the algorithm can choose. For the base case, we simply use Gaussian elimination to solve $\Lm_G\xv =\bv$, which takes $O(1)$ time since $n_0$ is a constant. For every $t$, contracting $G$ down to $H^t$ and computing the new resistances takes $O(m)$ time.

\subsection{Analysis of the Sparsify and Recurse Algorithm}
We now analyze Algorithm \ref{alg:fastest}. We first do a convergence analysis, then analyze the running time.
\subsubsection{Error Analysis}
The lemma below bounds the expected rate of convergence of $\xv^t$ to $\xv^*$.
\begin{lemma}
\label{lem:err}
For all $t \geq 0$, we have $\E\norm{\xv^* - \xv^t}_{\Lm_G}^2 \leq \left(1 - \beta + \beta e^{-\frac{d}{\tau}}\right)^t \norm{\xv^*}_{\Lm_G}^2$. Here,
\begin{itemize}
    \item $\tau = O(m\log n \log\log n)$ is the stretch of the spanning tree,
    \item $d$ is the number of updates in each batch,
    \item $\beta = \left(1-\frac{1}{n_0^2}\right)\left(1-\left(4\epsilon'\cdot\frac{1+\gamma}{1-\gamma}\cdot\left(1+\frac{\gamma^2}{(1-\gamma)^2}\right) + \frac{2\gamma^2}{(1-\gamma)^2}\right)\right).$
\end{itemize}
In particular, if we choose $n_0 = 10$, $\gamma = \frac{1}{100}$, and $\epsilon' = \frac{1}{100}$, then $\beta \geq \frac{4}{5}$, so that
$$\E\norm{\xv^* - \xv^t}_{\Lm_G}^2 \leq \left(\frac{1}{5} + \frac{4}{5} e^{-\frac{d}{\tau}}\right)^t \norm{\xv^*}_{\Lm_G}^2.$$
\end{lemma}
\begin{proof}
Define the random variable $D^t := \sB(\xv^*) - \sB(\xv^t)$.
We will show in Lemma \ref{lem:gap_reduce} that for every possible realization $\xv^t$, we have 
$$\E\left[D^{t+1} \mid \mathbf{x}^t\right] \leq \left(1 -\beta+\beta e^{-d/\tau}\right)\E\left[D^t \mid \mathbf{x}^t\right].$$
This implies that $\E[D^{t+1}] \leq \left(1 -\beta+\beta e^{-d/\tau}\right)\E\left[D^t\right]$ unconditionally. 

It then follows that
$$\E\left[D^t\right] \leq \left(1 -\beta+\beta e^{-d/\tau}\right)^t\E\left[D^0\right] 
=\left(1 -\beta+\beta e^{-d/\tau}\right)^t\left(\sB(\mathbf{x}^*) - \sB(\mathbf{x}^0)\right) 
= \left(1 -\beta+\beta e^{-d/\tau}\right)^t\sB(\mathbf{x}^*).$$
Thus, 
$$\sB(\mathbf{x}^*) - \E[\sB(\mathbf{x}^t)] \leq  \left(1 -\beta+\beta e^{-d/\tau}\right)^t\sB(\mathbf{x}^*).$$
\end{proof}

As in the original analysis of \ALGNAME, we will study the duality gap and analyze its decrease at each step of the algorithm. Consider some iteration $t$ of the algorithm. Recall that $\xv^t$ is the iterate at the start of iteration $t$. For every possible sequence of $e_1^t, \ldots, e_d^t$ (the trees edges chosen in iteration $t$), define the following:
\begin{itemize}
    \item Let $\hat{\xv}^{t+1}$ be the vector obtained from $\xv^t$ by adding ${\Delta}^t(V^t_i)$ to every vertex in $V^t_i$, where  ${\Delta}^t := \Lm_{H^t}^\dag \bv^t_{H^t}$. 
    \item Let $\bar{\xv}^{t+1}$ be obtained from $\xv^t$ by applying the updates for the sequence of tree edges $e_1^t, \ldots, e_d^t$, \emph{one by one}. (i.e. Exactly as in the original, unbatched version of \ALGNAME\ described in \Cref{sec:alg}.)
\end{itemize}
\begin{lemma}
\label{lem:errs}
Fix any choice of $e^t_1, \ldots, e^t_d$, and assume that $\wtd{\Lm}_{H^t}$ is a $\gamma$-approximate sparsifier of $\Lm_{H^t}$. Then
$$\sB(\xv^{t+1}) - \sB(\xv^t) \geq  (1-\alpha)\left(\sB(\hat{\xv}^{t+1}) - \sB(\xv^t)\right)$$
where $\alpha = 4\epsilon'\cdot\frac{1+\gamma}{1-\gamma}\cdot\left(1+\frac{\gamma^2}{(1-\gamma)^2}\right) + \frac{2\gamma^2}{(1-\gamma)^2}$. 

If we further assume that $\gamma \in (0,\frac12)$, we can simplify to get
$$\sB(\xv^{t+1}) - \sB(\xv^t) \geq  (1-12\epsilon'-8\gamma^2-48\epsilon'\gamma^2)\left(\sB(\hat{\xv}^{t+1}) - \sB(\xv^t)\right).$$
\end{lemma}
\begin{proof}
To simplify notation, in this proof we will use 
\begin{itemize}
    \item $\bv_H$ to denote ${\bv^t_{H^t}}$, 
    \item $\tilde{\Delta}_{\epsilon'}$ to denote $\tilde{\Delta}_{\epsilon'}^t$,
    \item $\tilde{\Delta}$ to denote $\tilde{\Delta}^t$
    \item $\Delta$ to denote $\Delta^t$,
    \item $H$ to denote $H^t$,
\end{itemize}

Later in this proof, we will show that 
\begin{equation}
\label{eq:delta_approx}
\norm{\tilde{\Delta}_{\epsilon'} - \Delta}_{\Lm_H}^2 \leq \alpha \norm{\Delta}_{\Lm_H}^2,
\end{equation}
for a constant $\alpha$ that depends on $\epsilon'$ and $\gamma$. Assuming (\ref{eq:delta_approx}) holds, by the definition of the matrix norm it follows that
$$\left(\tilde{\Delta}_{\epsilon'} - \Delta\right)^T\Lm_H\left(\tilde{\Delta}_{\epsilon'} - \Delta\right) \leq \alpha \Delta^T\Lm_H\Delta.$$
Expanding the left-hand side and rearranging, we get
$$2\tilde{\Delta}_{\epsilon'}^T\Lm_H \Delta - \tilde{\Delta}_{\epsilon'}^T\Lm_H \tilde{\Delta}_{\epsilon'} \geq (1 - \alpha) \Delta^T\Lm_H\Delta.$$
Using $\Lm_H \Delta = \bv_H$, this becomes
$$2\tilde{\Delta}_{\epsilon'}^T\bv_H - \tilde{\Delta}_{\epsilon'}^T \Lm_H\tilde{\Delta}_{\epsilon'} \geq (1 - \alpha) \Delta^T\Lm_H\Delta.$$
Recall that $\xv^{t+1}$ is obtained from $\xv^t$ by adding $\tilde{\Delta}_{\epsilon'}(V_i^t)$ to every vertex in $V_i^t$. 
Using Lemma \ref{lem:dual_incr}
with $\Delta(i) = \wtd{\Delta}_{\epsilon'}(V_i^t)$, $\xv = \xv^t$,  $\wtd{\xv} = \xv^{t+1}$, $\wtd{\Lm} = \Lm_H$, and $\wtd{\bv} = \bv_H$ it follows that the left-hand side is equal to $\sB(\xv^{t+1}) - \sB(\xv^t)$. 
On the other hand, $\sB(\hat{\xv}^{t+1}) - \sB(\xv^t) = 2\bv_H^T\Delta - \Delta^T\Lm_H\Delta = \Delta^T\Lm_H\Delta$. (Since $\Lm_H \Delta = \bv_H$.)
Thus, the right-hand side is equal to $(1-\alpha)\left(\sB(\hat{\xv}^{t+1}) - \sB(\xv^{t})\right)$. Thus we have
$$\sB(\xv^{t+1}) - \sB(\xv^t) \geq  (1 - \alpha)\left(\sB(\hat{\xv}^{t+1}) - \sB(\xv^t)\right),$$
as claimed. 

It remains to prove (\ref{eq:delta_approx}). To prove (\ref{eq:delta_approx}), note that we have
\begin{enumerate}
    \item $\norm{\tilde{\Delta}_{\epsilon'} - \tilde{\Delta}}_{\wtd{\Lm}_H}^2 \leq \epsilon' \norm{\tilde{\Delta}}_{\wtd{\Lm}_H}^2$ (This is the error from the recursive solve).
    \item $\norm{\tilde{\Delta} - \Delta}^2_{\Lm_H} \leq h(\gamma) \norm{\Delta}_{\Lm_H}^2$ (Follows by part 2 of Proposition \ref{prop:spec_approx}. This is the error from sparsification). 
\end{enumerate}
The first inequality, together with $(1-\gamma)\Lm_H \preceq \wtd{\Lm}_H \preceq (1+\gamma)\Lm_H$ and part 1 of Proposition \ref{prop:spec_approx}, implies that 
$$\norm{\tilde{\Delta}_{\epsilon'} - \tilde{\Delta}}_{\Lm_H}^2 \leq
\frac{1}{1-\gamma}
\norm{\tilde{\Delta}_{\epsilon'} - \tilde{\Delta}}_{\wtd{\Lm}_H}^2 \leq
\frac{\epsilon'}{1-\gamma} \norm{\tilde{\Delta}}_{\wtd{\Lm}_H}^2\leq \epsilon' \cdot \frac{1+\gamma}{1-\gamma} \cdot \norm{\tilde{\Delta}}^2_{\Lm_H}.$$

Now, using the inequality $\norm{a+b}^2 \leq 2\norm{a}^2 + 2\norm{b}^2$ (which holds for any norm), we note that
$$\norm{\tilde{\Delta}}^2_{\Lm_H} \leq 2\norm{\Delta}_{\Lm_H}^2 + 2\norm{\tilde{\Delta} -\Delta}_{\Lm_H}^2 \leq 2\norm{\Delta}_{\Lm_H}^2 +2h(\gamma)\norm{\Delta}_{\Lm_H}^2.$$
Hence,
$$\norm{\tilde{\Delta}_{\epsilon'} - \tilde{\Delta}}_{\Lm_H}^2 \leq 2\epsilon' \cdot \frac{1+\gamma}{1-\gamma} \cdot (1+h(\gamma)) \norm{\Delta}_{\Lm_H}^2.$$
Again using $\norm{a+b}^2 \leq 2\norm{a}^2 + 2\norm{b}^2$, we have
\begin{align*}
    \norm{\tilde{\Delta}_{\epsilon'} - {\Delta}}_{\Lm_H}^2
    &\leq 2\left(\norm{\tilde{\Delta}_{\epsilon'} - \tilde{\Delta}}_{\Lm_H}^2 + \norm{\tilde{\Delta} - {\Delta}}_{\Lm_H}^2\right) \\
    &\leq 2\left(2\epsilon' \cdot \frac{1+\gamma}{1-\gamma} \cdot (1+h(\gamma)) \norm{\Delta}_{\Lm_H}^2 + h(\gamma)\norm{\Delta}_{\Lm_H}^2\right) \\
    &= \left(4\epsilon' \cdot \frac{1+\gamma}{1-\gamma} \cdot (1+h(\gamma)) + 2h(\gamma) \right)\norm{\Delta}_{\Lm_H}^2.
\end{align*}
Therefore, (\ref{eq:delta_approx}) holds with $\alpha =4\epsilon' \cdot \frac{1+\gamma}{1-\gamma} \cdot (1+h(\gamma)) + 2h(\gamma)$.
\end{proof}
\begin{lemma}
\label{lem:gap_reduce}
For any vector $\xv^t$, we have
$$\sB(\xv^*) - \E[\sB({\xv}^{t+1})] \leq \left(1 -\beta+\beta e^{-d/\tau}\right) \left(\sB(\xv^*) - \sB(\xv^t)\right),$$
where $\beta = (1-\frac{1}{n_0^2})(1-\alpha)$. 

Here, the expectation is taken over the random choices of $e_1^t, \ldots, e^t_d$, and also over the randomness of the sparsification step. (Recall that the sparsify algorithm is randomized, and in particular it successfully returns a $\gamma$-approximate sparsifier with probability $\geq 1 - \frac{1}{\abs{V(H^t)}^2}$.)

% \richard{also over randomness of the sparsify step,
% but the recursive solve works unconditionally?}
\end{lemma}
\begin{proof}
By Lemma \ref{lem:gap_decreases}, we know that
$$\sB(\xv^*) - \E[\sB(\bar{\xv}^{t+1})] \leq \left(1 - \frac1\tau\right)^d\left(\sB(\xv^*) - \sB(\xv^t)\right) \leq e^{-\frac{d}{\tau}}\left(\sB(\xv^*) - \sB(\xv^t)\right).$$
Rearranging, this is equivalent to
$$\E[\sB(\bar{\xv}^{t+1})] - \sB(\xv^t) \geq \left(1 - e^{-\frac{d}{\tau}} \right)\left(\sB(\xv^*) - \sB(\xv^t)\right).$$
Observe that for every realization of $e^t_1, \ldots, e^t_d$, we have $\sB(\hat{\xv}^{t+1}) \geq \sB(\bar{\xv}^{t+1})$. This is because $\hat{\xv}^{t+1} - \xv^t = \Delta^t$, where $\Delta^t$ by definition is the vector that \emph{maximizes} the increase $\sB(\hat{\xv}^{t+1}) - \sB(\xv^t)$ while subject to being incremented by the same amount on each of the components $V^t_0, \ldots, V^t_d$. On the other hand, the vector $\bar{\xv}^{t+1} - \xv^t$ is also  incremented by the same amount on each of the components $V^t_0, \ldots, V^t_d$ by the way our original algorithm works.

Since $\sB(\hat{\xv}^{t+1}) \geq \sB(\bar{\xv}^{t+1})$ holds for every realization of $e^t_1, \ldots, e^t_d$, it follows that $\E[\sB(\hat{\xv}^{t+1})] \geq \E[\sB(\bar{\xv}^{t+1})]$, where the expectation is taken over the random choices of $e^t_1, \ldots, e^t_d$ made by the algorithm. Hence,
\begin{equation}
\label{eq:opt_incr}
\E[\sB(\hat{\xv}^{t+1})] - \sB(\xv^t) \geq \left(1 - e^{-\frac{d}{\tau}}\right) \left(\sB(\xv^*) - \sB(\xv^t)\right).
\end{equation}
To conclude, we will use Lemma \ref{lem:errs} to translate the above inequality (which is in terms of $\hat{\xv}^{t+1}$), to an inequality in terms of $\xv^{t+1}$. We have
\begin{itemize}
    \item With probability $\geq 1 - \frac{1}{n_0^2}$, the sparsifier is successful and by Lemma \ref{lem:errs}, $$\E[\sB(\xv^{t+1})] - \sB(\xv^t) \geq (1-\alpha)\left(\E[\sB(\hat{\xv}^{t+1})] - \sB(\xv^t) \right).$$
    \item With probability $\leq \frac{1}{n_0^2}$, the sparsifier is unsuccessful and $\E[\sB(\xv^{t+1})] - \sB(\xv^t) \geq 0$. This is because in the algorithm, we evaluate $\sB(\xv^{t+1})$ and only update $\xv^t$ to $\xv^{t+1}$ if $\sB(\xv^{t+1}) \geq \sB(\xv^t)$. Otherwise, we make $\xv^{t+1} = \xv^t$. 
\end{itemize}
Note that the above expectations are with respect to the random choices of $e^t_1, \ldots, e^t_d$, conditioned on the sparsifier being successful/unsuccessful. Now, taking another expectation with respect to the randomness of the sparsifier, we get
\begin{align*}
\E[\sB(\xv^{t+1})] - \sB(\xv^t) &\geq \left(1-\frac{1}{n_0^2}\right)(1-\alpha)\left(\E[\sB(\hat{\xv}^{t+1})] - \sB(\xv^t) \right) \\
&\geq \left(1-\frac{1}{n_0^2}\right)(1-\alpha)\left(1 - e^{-\frac{d}{\tau}}\right) \left(\sB(\xv^*) - \sB(\xv^t)\right) &\text{(by (\ref{eq:opt_incr}))}
\end{align*}
Rearranging the above inequality gives
$$\sB(\xv^*) - \E[\sB({\xv}^{t+1})] \leq \left(1 - \left(1-\frac{1}{n_0^2}\right)(1-\alpha)\left(1 - e^{-\frac{d}{\tau}}\right)\right) \left(\sB(\xv^*) - \sB(\xv^t)\right),$$
as claimed.
\end{proof}

\subsubsection{Running Time Analysis}
The following theorem bounds the running time of the algorithm.
\fastest*

% \todo{Can probably avoid the exponent on $\log 1/\epsilon$, because only the topmost layer needs $\epsilon$ accuracy.}

\begin{proof}
By \Cref{lem:err}, it suffices to run the algorithm for $K$ iterations, for any $K \geq \frac{\ln \epsilon}{\ln\left(\frac{1}{5} + \frac{4}{5}e^{-d/\tau}\right)}$.

Using the inequalities $e^{-x} \leq 1 - \frac{x}{2}$ and $\ln(1-x) \leq -\frac{x}{2}$ which hold for $x\in(0,1)$, we see that it suffices to choose $K = \frac{5\tau}{d} \ln(1/\epsilon)$. 

Recall that $\tau \leq c_3m\log^2n$, for some constant $c_3$. Thus, we choose 
%\begin{itemize}
%\item 
$d = c_3\bar{m}(\log^2n)\cdot{m}^{-\delta}.$ Here, $\bar{m}$ is the number of edges in the \emph{current iteration}, while ${m}$ is the number of edges in the \emph{topmost} iteration (i.e. in the original graph $G$).
%\end{itemize}
With this choice of $d$, we have $K = 5c_3m^{\delta}\ln(1/\epsilon).$ (Note that $K$ is the same at every level of the recursion tree.)

The work at one iteration consists of 
\begin{itemize}
    \item Computing $T$,
    \item Doing $K$ times
    \begin{itemize}
    \item Contracting and sparsifying to a graph with $d$ vertices and $a_1 d \log^cn / \gamma^2$ edges, for some constant $a_1$.
    \item Doing a recursive call.
    \end{itemize}
\end{itemize}
If $m$ is the number of edges in the graph at one level of the recursion tree, then at the next level, the number of edges is
$$a_1 d \log^cn / \gamma^2 = a_1c_3\bar{m}(\log^2n)\cdot{m}^{-\delta}\log^cn / \gamma^2 = a_1c_3\bar{m}\cdot{m}^{-\delta}(\log n)^{2+c}/\gamma^2$$
Therefore, the number of edges in a graph at level $l$ of the recursion tree is
$$\textrm{edges}_l = {m}^{1-l\delta}(a_1c_3)^l(\log n)^{(2+c)l}/\gamma^{2l}.$$
The total work required by a node at level $l$ of the recursion tree is dominated by the sparsifier, which takes time
$$\textrm{work}_l = K\cdot a_2\cdot \textrm{edges}_l \cdot \log^{c_1}n$$
for some constant $a_2$.

Finally, the total number of recursion tree nodes at level $l$ is equal to $K^l$. This implies that the total work required by all the nodes at level $l$ of the recursion tree is equal to 
\begin{align*}
\textrm{total work}_l &= \textrm{work}_l \cdot K^l \\
&= K\cdot a_2\cdot \textrm{edges}_l \cdot \log^{c_1}n \cdot K^l\\
&= K\cdot a_2\cdot {m}^{1-l\delta}(a_1c_3)^l(\log n)^{(2+c)l}/\gamma^{2l} \cdot \log^{c_1}n \cdot K^l\\
&= 5^{l+1}c_3^{l+1}{m}^{1+\delta}(\ln 1/\epsilon)^{l+1}a_2(a_1c_3)^l(\log n)^{(2+c)l}/\gamma^{2l} \cdot \log^{c_1} n \\
&\leq A^l {m}^{1+\delta}(\log n)^{Bl}(\ln 1/\epsilon)^{l+1}
\end{align*}
for some constants $A,B$.

To conclude, we note that the total work summed across all the levels is at most a constant factor times the total work at the maximum level, which is $l = \frac{1}{\delta}$. 

\end{proof}

%% file: appendix.tex
\section{Deriving the dual of $p$-norm flow}
\label{app:dual_deriv}
Here, we show the steps of deriving the dual of the $p$-norm flow problem. Although in the paper we focus on the case where $1 < p \leq 2$, the derivation of the dual problem works for all $p > 1$. We follow the process described in Chapter 5 of \cite{boyd_vandenberghe}. The primal problem is
\begin{align*}
    \min \quad &\frac{1}{p}\sum_{e} r(e)\abs{f(e)}^p \\
    \text{s.t.} \quad &Af = b
\end{align*}
We will show that its dual is
\begin{align*}
    \max_x \quad b^T x - \frac{p-1}{p} \sum_{ij \in E} \left(\frac{\abs{x(i) - x(j)}^p}{r_{ij}}\right)^{\frac{1}{p-1}}
\end{align*}
First, we form the Lagrangian by bringing the constraints into the objective as a penalty term:
\begin{align*}
    L(f, x) = \frac{1}{p} \sum_e r(e) \abs{f(e)}^p + x^T(Af - b).
\end{align*}
The Lagrangian dual function is
\begin{align*}
    g(x) = \min_f L(f, x) = \min_f \left\{\frac{1}{p} \sum_e r(e) \abs{f(e)}^p + x^T(Af - b)\right\}.
\end{align*}
Note that for any $x$, $g(x)$ is a lower bound on the optimal value of the primal. This is because for any $f$ with $Af = b$, 
\begin{align*}
    g(x) \leq \frac{1}{p} \sum_e r(e) \abs{f(e)}^p + x^T(Af - b) = \frac{1}{p} \sum_e \abs{f(e)}^p.
\end{align*}
The dual problem is the optimization problem whose goal is to maximize the lower bound:
\begin{align*}
    \max_x g(x) = \max_x \min_f \left\{\frac{1}{p} \sum_e r(e) \abs{f(e)}^p + x^T(Af - b)\right\}
\end{align*}
For a given $x$, the function $\frac{1}{p} \sum_e r(e) \abs{f(e)}^p + x^T(Af - b)$ is convex in $f$. Hence, it is minimized when its gradient with respect to $f$ is 0. Taking the gradient component by component, and setting to 0, we get
\begin{align*}
    \forall e: \qquad r(e)f(e)\abs{f(e)}^{p-2} +\left(A^Tx\right)_e = 0
\end{align*}
which implies
\begin{align}
\label{eq:kkt}
    \forall e = (i,j): \qquad r(e)f(e)\abs{f(e)}^{p-2} = (x(j) - x(i)).
\end{align}
Taking absolute values of both sides, then rearranging, gives 
\begin{align*}
\abs{f(e)}^p = \abs{\frac{x(j) - x(i)}{r(e)}}^{\frac{p}{p-1}} \quad \forall \; e = (i,j) \in E.
\end{align*}
Also, multiplying both sides of \cref{eq:kkt} by $f(e)$ gives $r(e)\abs{f(e)}^p = -f(e)(x(i)-x(j))$. This gives
\begin{align*}
    x^T Af = f^T Ax = \sum_{e = ij \in E} f(e)(x(i) - x(j)) = -\sum_{ij \in E} r(e)\abs{f(e)}^p.
\end{align*}
Plugging these back, we get that the dual problem is
\begin{align*}
    \max_x \left\{ \frac{1}{p} \sum_{e = ij} r(e)\abs{\frac{x(j) - x(i)}{r(e)}}^{\frac{p}{p-1}} - \sum_{e=ij}r(e)\abs{\frac{x(j) - x(i)}{r(e)}}^{\frac{p}{p-1}} - b^Tx\right\}
\end{align*}
Collecting terms and replacing $x$ with $-x$, the above is equal to
\begin{align*}
    \max_x \left\{ b^Tx - \left(1 - \frac{1}{p}\right) \sum_{e = ij} r(e)\abs{\frac{x(i) - x(j)}{r(e)}}^{\frac{p}{p-1}} \right\},
\end{align*}
as desired.

\section{Omitted Proofs from Section \ref{sec:alg}}

\enerincr*

\begin{proof}
The way we update $\mathbf{x}$ is by adding a constant $\Delta$ to the potentials of every vertex in $C$, where 
$$
\Delta = (b(C) - f(C))\cdot R(C)
$$
Recall that $f(C)$ is the net amount of flow going out of $C$ in the flow induced by $\mathbf{x}$. That is,
$$f(C) = \sum_{\substack{ij \in E \\ i \in C,\, j \not\in C}} \frac{x(i) - x(j)}{r(i, j)}$$

Note that the new potentials $\mathbf{x}'$ can be expressed as $\mathbf{x}' = \mathbf{x} + \Delta \one_C$. 

We have
\begin{align*}
    2(\sB(\mathbf{x}') - \sB(\mathbf{x}))
    &= 2\mathbf{b}\trans \mathbf{x}' - (\mathbf{x}')\trans \mathbf{L}\mathbf{x}' - (2\mathbf{b}\trans \mathbf{x} - \mathbf{x}\trans \mathbf{L}\mathbf{x}) \\
    &= 2\mathbf{b}\trans(\mathbf{x} + \Delta\cdot\one_C) - 2\mathbf{b}\trans \mathbf{x} - (\mathbf{x}')\trans \mathbf{L}\mathbf{x}' + \mathbf{x}\trans \mathbf{L}\mathbf{x} \\
    &= 2\Delta\cdot \mathbf{b}\trans \one_C - \sum_{ij \in E}\frac{1}{r(i, j)}\left[(x'(i) - x'(j))^2 - (x(i) - x(j))^2 \right] \\
    &= 2\Delta\cdot \mathbf{b}\trans \one_C - \sum_{ij \in \delta(C)}\frac{1}{r(i, j)}\left[(x'(i) - x'(j))^2 - (x(i) - x(j))^2 \right] \\
    &= 2\Delta\cdot \mathbf{b}\trans \one_C - \sum_{\substack{i \in C, \, j \not\in C \\ ij \in \delta(C)}}\frac{1}{r(i, j)}\left[(x(i)+\Delta - x(j))^2 - (x(i) - x(j))^2 \right] \\
    &= 2\Delta \cdot \mathbf{b}\trans \one_C - \sum_{\substack{i \in C, \, j \not\in C \\ ij \in \delta(C)}}\frac{1}{r(i,j)}\left[2\Delta\cdot(x(i)-x(j)) +\Delta^2\right] \\
    &= 2\Delta \cdot \mathbf{b}\trans \one_C - 2\Delta\cdot f(C) - \Delta^2 \sum_{(i,j) \in \delta(C)} \frac{1}{r(i,j)} \\
    &= 2\Delta \cdot \mathbf{b}\trans \one_C - 2\Delta\cdot f(C) - \Delta^2\cdot R(C)^{-1} \\
    &= 2\Delta \cdot b(C) - 2\Delta\cdot f(C) - \Delta^2 R(C)^{-1} \\
    &= 2\Delta ^2 R(C)^{-1} - \Delta^2 R(C)^{-1} \\
    &= \Delta^2/R(C).
\end{align*}
\end{proof}

\gapp*

\begin{proof}
By definition, we have
\begin{align*}
   2\gap(\mathbf{f}, \mathbf{x})
    &= \sum_{e \in E} r(e)f(e)^2 - (2\mathbf{b}\trans \mathbf{x} - \mathbf{x}\trans \mathbf{L}\mathbf{x}).
\end{align*}
Note that 
\begin{align*}
    \mathbf{b}\trans \mathbf{x} = \sum_{i \in V}b(i)x(i) 
    = \sum_{i \in V} x(i) \left(\sum_{j: (i,j) \in \vec{E}} f(i,j) - \sum_{j:(j,i) \in \vec{E}} f(j,i)\right)
    = \sum_{(i,j) \in \vec{E}}f(i,j)(x(i) - x(j))
    \end{align*}
    and 
    \begin{align*}
    \mathbf{x}\trans \mathbf{L}\mathbf{x} &= 
    \sum_{(i,j) \in \vec{E}} \frac{(x(i) - x(j))^2}{r(i,j)}.
    \end{align*}
    Plugging these into our expression for $\gap(\mathbf{f}, \mathbf{x})$, we obtain 
    \begin{align*}
        2\gap(\mathbf{f}, \mathbf{x})
        &= \sum_{(i,j) \in \vec{E}}\left[r(i,j)f(i,j)^2 - 2f(i,j)(x(i) - x(j)) + \frac{(x(i) - x(j))^2}{r(i,j)}\right] \\
        &= \sum_{(i,j) \in \vec{E}} r(i,j)\left(f(i,j) - \frac{x(i)-x(j)}{r(i,j)}\right)^2
    \end{align*}
    which is what we wanted to show.

\end{proof}

\gaplem*

\begin{proof}
Recall that $C(i,j)$, $\Delta(C(i,j))$ and $R(C(i,j))$ were defined as follows:
\begin{itemize}
    \item $C(i,j)$ is the set of vertices on the side of the fundamental cut of $T$ determined by $(i,j)$ containing $i$. In other words, $C(i,j)$ consists of the vertices in the component of $T - ij$ with $i \in C(i,j)$ and $j \not\in C(i,j)$.
    \item $R(C(i,j)) = \left(\sum_{ij \in \delta(C)} \frac{1}{r(i,j)}\right)^{-1}$.
    \item $\Delta(C(i,j)) = (b(C(i,j)) - f(C(i,j)))R(C(i,j))$, where
    \begin{itemize}
        \item $b(C(i,j)) = \mathbf{b}\trans \one_{C(i,j)}$, and 
        \item $f(C(i,j)) = \displaystyle\sum_{\substack{k \in C(i,j), \, l \not\in C(i,j) \\ kl \in E}} \frac{x(k) - x(l)}{r(k, l)}$
    \end{itemize}
\end{itemize}
We have 
\begin{align*}
    2\gap(\fv_{T, \xv}, \mathbf{x})
    &= \sum_{(i, j) \in E} r(i, j) \left(f_{T,\xv}(i, j) - \frac{x(i)-x(j)}{r(i, j)}\right)^2 \\
    &= \sum_{(i, j) \in T} r(i, j) \left(f_{T,\xv}(i, j) - \frac{x(i)-x(j)}{r(i, j)}\right)^2 \\
    &= \sum_{(i, j) \in T} r(i, j) \left[\left(b(C(i, j)) - \sum_{\substack{ k \in C(i, j), l \not\in C(i, j)\\kl \in E - ij }} \frac{x(k) - x(l)}{r(k, l)}\right) - \frac{x(i) - x(j)}{r(i, j)} \right]^2 \\
    &= \sum_{(i, j) \in T} r(i, j) \left[b(C(i,j)) - \sum_{\substack{k \in C(i, j), l \not\in C(i, j)\\kl \in E}} \frac{x(k) - x(l)}{r(k, l)}\right]^2 \\
    &= \sum_{(i, j) \in T} r(i, j) \left[b(C(i,j)) - f(C(i, j))\right]^2 \\
    &= \sum_{(i, j) \in T} r(i, j)\cdot \frac{\Delta(C(i, j))^2}{R(C(i, j))^2}
\end{align*}
\end{proof}

\gapdecr*

\begin{proof}
We know from the discussion above that
$$\E[\sB(\mathbf{x}^{t+1})] - \sB(\mathbf{x}^t) = \frac{1}{\tau}\gap(\fv_{T, \xv}, \mathbf{x}^t),$$
where $\fv_{T, \xv}$ is the tree-defined flow associated with potentials $\mathbf{x}^t$. Since $\gap(\fv_{T, \xv}, \mathbf{x}^t) \geq \sB(\mathbf{x}^*) - \sB(\mathbf{x}^t)$, we get
$$\E[\sB(\mathbf{x}^{t+1})] - \sB(\mathbf{x}^t) \geq \frac{1}{\tau}\left(\sB(\mathbf{x}^*) - \sB(\mathbf{x}^t)\right).$$
Rearranging gives
$$\sB(\mathbf{x}^*) - \E[\sB(\mathbf{x}^{t+1})] \leq \left(1 - \frac1\tau\right)\left(\sB(\mathbf{x}^*) - \sB(\mathbf{x}^t)\right),$$
as desired. 
\end{proof}

\finalgap*

\begin{proof}
Define the random variable $D_t := \sB(\mathbf{x}^*) - \sB(\mathbf{x}^t)$. By Lemma \ref{lem:gap_decreases}, we know that
$$\E\left[D^{t+1} \mid \mathbf{x}^t\right] \leq \left(1 - \frac{1}{\tau}\right)\E\left[D^t \mid \mathbf{x}^t\right]$$
for all possible vectors of potentials $\mathbf{x}^t$. This implies that $\E\left[D^{t+1}\right] \leq \left(1 - \frac{1}{\tau}\right)\E\left[D^t \right]$ unconditionally. 

By induction on $t$, it then follows that
$$\E\left[D^K\right] \leq \left(1 - \frac1\tau\right)^K\E\left[D^0\right] =\left(1 - \frac1\tau\right)^K\left(\sB(\mathbf{x}^*) - \sB(\mathbf{x}^0)\right) = \left(1 - \frac1\tau\right)^K\sB(\mathbf{x}^*).$$
Thus, 
$$\sB(\mathbf{x}^*) - \E[\sB(\mathbf{x}^K)] \leq  \left(1 - \frac1\tau\right)^K\sB(\mathbf{x}^*).$$

Using the inequality $1-x \leq e^{-x}$, we obtain
$$\sB(\mathbf{x}^*) - \E[\sB(\mathbf{x}^K)] \leq e^{-K/\tau} \sB(\mathbf{x}^*).$$
Hence, if $K \geq \tau\ln(\frac{1}{\epsilon})$, then we will have $\sB(\mathbf{x}^*) - \E[\sB(\mathbf{x}^K)] \leq \epsilon\cdot\sB(\mathbf{x}^*)$, as desired. 
\end{proof}

\dualkosz*
    \begin{proof}[Proof of \Cref{thm:dual_kosz}]
     By Corollary \ref{cor:final_gap}, after $K=\tau\ln(\frac{\tau}{\epsilon})$ iterations, the algorithm returns potentials  $\xv^K$ such that $\sB(\mathbf{x}^*) - \E[\sB(\mathbf{x}^K)] \leq \frac{\epsilon}{\tau} \cdot \sB(\mathbf{x}^*)$. Combining with Lemma \ref{lem:energy_to_potential}, we get that $\E \norm{\mathbf{x}^* - \mathbf{x}^K}_\mathbf{L}^2 \leq \frac{\epsilon}{\tau}\norm{\mathbf{x}^*}_\mathbf{L}^2$. Finally, \Cref{lem:rounding_error} gives $\E\left[\sE(\fv^K)\right] \leq (1+\epsilon)\sE(\fv^*)$.  
     
    \end{proof}

   \begin{restatable}{lemma}{stretch}
    \label{lem:stretch}
    We have $\tau = \st_T(G, \rv)$. 
    \end{restatable}

\begin{proof}
    We write out the definitions of $\tau$ and $\st_T(G, \rv)$:
    $$\tau = \sum_{(i, j) \in T} \frac{r(i, j)}{R(C(i, j))} = \sum_{(i, j) \in T} r(i, j)\sum_{(k,l) \in \delta(C(i, j))} \frac{1}{r(k, l)}$$
    and
    $$
    \st_T(G, \rv) = \sum_{(i, j) \in \vec{E}} \st_T((i, j), \rv)
    = \sum_{(i, j) \in \vec{E}} \frac{1}{r(i, j)}\sum_{(k, l) \in P(i,j)} r(k, l),
    $$
    where $P(i, j)$ is the unique path from $i$ to $j$ in $T$. 
    
    It turns out that the expressions for $\tau$ and $\st_T(G)$ are summing exactly the same terms, just in different ways. Indeed, we have
    \begin{align*}
        \tau &= \sum_{(i, j) \in T} \sum_{(k,l) \in \delta(C(i, j))}\frac{r(i, j)}{r(k, l)} \\
        &= \sum_{(k, l) \in \vec{E}} \sum_{(i, j) \in P(k, l)} \frac{r(i, j)}{r(k, l)} \\
        &= \st_T(G, \rv). 
    \end{align*}
    To switch the order of summation from the first line to the second line, we used the fact that for an edge $(k, l) \in \vec{E}$, we have $(k,l) \in \delta(C(i, j))$ if and only if $(i, j) \in P(k, l)$. This is because $T$ is a spanning tree. 
    
    \end{proof}

    % \subsection{Convergence of Energy Implies Convergence of Potentials}
    By Corollary \ref{cor:final_gap}, we know that the potentials $\mathbf{x}^t$ found by the algorithm satisfy the property that $\sB(\mathbf{x}^t)$ converges to $\sB(\mathbf{x}^*)$ at a linear rate, in expectation. The following lemma shows that if $\mathbf{x}$ is a set of potentials such that $\sB(\mathbf{x})$ is close to $\sB(\mathbf{x}^*)$, then $\mathbf{x}$ is close to $\mathbf{x}^*$ as a vector (measured in the matrix norm defined by the Laplacian $\mathbf{L}$). 
    
    \begin{restatable}{lemma}{enertopot}
    \label{lem:energy_to_potential}
    Let $\mathbf{x}$ be any vector of potentials. Then 
    $\frac12\norm{\mathbf{x}^* - \mathbf{x}}_\mathbf{L}^2 = \sB(\mathbf{x}^*) - \sB(\mathbf{x}).$
    In particular, if $\sB(\mathbf{x}^*) - \sB(\mathbf{x}) \leq \epsilon \cdot \sB(\mathbf{x}^*)$, then $\norm{\mathbf{x}^* - \mathbf{x}}_\mathbf{L}^2 \leq  \epsilon \norm{\mathbf{x}^*}_\mathbf{L}^2.$
    \end{restatable}

\begin{proof}
    We have 
    \begin{align*}
        \norm{\mathbf{x}^* - \mathbf{x}}_L^2
        &= (\mathbf{x}^* - \mathbf{x})\trans \mathbf{L} (\mathbf{x}^* - \mathbf{x}) \\
        &= (\mathbf{x}^*)\trans \mathbf{L}\mathbf{x}^* - 2\mathbf{x}\trans \mathbf{L}\mathbf{x}^* + \mathbf{x}\trans \mathbf{L}\mathbf{x} \\
        &= 2\sB(\mathbf{x}^*) - 2\mathbf{x}\trans \mathbf{b} + \mathbf{x}\trans \mathbf{L}\mathbf{x} \\
        &= 2\sB(\mathbf{x}^*) - 2\sB(\mathbf{x}).
    \end{align*}
    In particular, if $\sB(\mathbf{x}^*) - \sB(\mathbf{x}) \leq \epsilon \cdot \sB(\mathbf{x}^*)$, then $\norm{\mathbf{x}^* - \mathbf{x}}_\mathbf{L}^2 \leq 2\epsilon \cdot \sB(\mathbf{x}^*) = \epsilon \norm{\mathbf{x}^*}_\mathbf{L}^2$.
    This is because
    $$2\sB(\mathbf{x}^*) = 2\mathbf{b}\trans \mathbf{x}^* - (\mathbf{x}^*)\trans \mathbf{L}\mathbf{x}^* = (\mathbf{x}^*)\trans \mathbf{L}\mathbf{x}^* =  \norm{\mathbf{x}^*}_\mathbf{L}^2.$$
    \end{proof}
    
    Next, we show that if $\sB(\xv)$ is sufficiently close to $\sB(\xv^*)$, then the associated tree-defined flow $\fv_{T, \xv}$ has energy sufficiently close to $\sE(\fv^*)$.  
    
    \begin{restatable}{lemma}{rounderror}
    \label{lem:rounding_error}
    For any distribution over $\xv$ such that $\E_{\xv}[\sB(\xv)] \geq (1-\frac{\epsilon}{\tau})\sB(\xv^*)$, we have $\E_{\xv}[\sE(\fv_{T, \xv})] \leq (1+\epsilon)\sE(\fv^*)$.
    \end{restatable}

\begin{proof}
For ease of notation, in this proof let $\fv = \fv_{T, \xv}$. (Note that $\fv$ is a random vector that is a function of $\xv$.)
We have $\E_{\xv}[\sE(\fv) - \sE(\fv^*)] = \E_\xv[\gap(\fv, \xv^*)]$. 

For a fixed choice of $\xv$, consider running the algorithm for one more iteration starting from $\xv$ to obtain a vector $\xv'$. Then we have $\E[\sB(\xv')] - \sB(\xv) = \frac1\tau \gap(\fv, \xv)$. This implies $\sB(\xv^*) - \sB(\xv) \geq \frac1\tau \gap(\fv, \xv)$. Taking expectations with respect to $\xv$, we get $\E_{\xv}[\sB(\xv^*) - \sB(\xv)] \geq \frac1\tau\E_{\xv}[ \gap(\fv, \xv)]$.
% $\sB(\xv^*) - \sB(\xv) \geq \frac1\tau \gap(\fv, \xv)$. 
Thus,
\begin{align*}
    \E_{\xv}[\sE(\fv) - \sE(\fv^*)] &=  \E_{\xv}[\gap(\fv, \xv^*)] \\
     &=   \E_{\xv}[\gap(\fv, \xv) - (\sB(\xv^*) - \sB(\xv))]  \\
    &\leq  (\tau-1)\E_\xv[\sB(\xv^*) - \sB(\xv)] \\
    &\leq   \tau\E_\xv[\sB(\xv^*) - \sB(\xv)] \\
    &\leq  \epsilon\sB(\xv^*) \\
    &= \epsilon \sE(\fv^*).
\end{align*}
\end{proof}

% \begin{proof}
% We have $\sE(\fv) - \sE(\fv^*) = \gap(\fv, \xv^*)$. 

% Consider running the algorithm for one more iteration starting from $\xv$ to obtain a vector $\xv'$. Then we have 
% $\E[\sB(\xv')] - \sB(\xv) = \frac1\tau \gap(\fv, \xv)$, which implies that $\sB(\xv^*) - \sB(\xv) \geq \frac1\tau \gap(\fv, \xv)$. Thus,
% \begin{align*}
%     \sE(\fv) &= \sE(\fv^*) + \gap(\fv, \xv^*) \\
%     &=  \sE(\fv^*) + \gap(\fv, \xv) - (\sB(\xv^*) - \sB(\xv)) \\
%     &\leq \sE(\fv^*) + (\tau-1)(\sB(\xv^*) - \sB(\xv)) \\
%     &\leq  \sE(\fv^*) + \tau(\sB(\xv^*) - \sB(\xv)) \\
%     &\leq  (1+\epsilon)\sE(\fv^*).
% \end{align*}
% \end{proof}

\runtime*
\begin{proof}
    
    We can find a spanning tree $T$ with total stretch $\tau=O(m\log n \log\log n)$ in $O(m \log n \log\log n)$ time \cite{AN12}.
    
    For concreteness, fix an arbitrary vertex to be the root of $T$, and direct all edges in $T$ towards the root. The set of fundamental cuts we consider will be the vertex sets of subtrees of $T$.  
    
    To compute $b(C)$ for these $n-1$ fundamental cuts $C$, we can work our way from the leaves up to the root. If $C = \{v\}$ is a leaf of $T$, then $b(C) = b(v)$. Otherwise, $C$ is a subtree rooted at $v$, and $b(C) = b(v) + \sum_{C'} b(C')$, where the sum is over the subtrees of $C$. Hence we can compute $b(C)$ for all fundamental cuts $C$ in $O(n)$ time.
    
    To compute $R(C)$ for the fundamental cuts $C$, we can maintain $n-1$ variables, one for each fundamental cut. The variable corresponding to cut $C$ will represent $R(C)^{-1} = \sum_{e \in \delta(C)} \frac{1}{r(e)}$, and the variables are initialized to 0. We then iterate through all the edges in the graph, and for each such edge $e$, add $\frac{1}{r(e)}$ to the value of each variable that represents a cut $C$ such that $e \in \delta(C)$. Although this naive implementation takes $O(mn)$ time in the worst-case, it is possible to improve this running time to $O(m\log n)$ using link-cut trees \cite{ST83}. One can also achieve this running time using the same data structure as the one used in \cite{KOSZ13}. 
    
    The last part of the running time is the time it takes to run a single iteration of the algorithm.
    In each iteration of the algorithm, we need to compute $\Delta = (b(C) - f(C))\cdot R(C)$, where $C$ is the fundamental cut selected at that iteration. In the above two paragraphs, we described how to precompute the values of $b(C)$ and $R(C)$ for every fundamental cut $C$; note that these values are fixed at the beginning and do not change during the course of the algorithm. Hence, it remains to compute $f(C)$. One way to compute $f(C)$ is to simply iterate over all the edges in $\delta(C)$, summing each edge's contribution to $f(C)$. This takes time proportional to $\abs{\delta(C)}$, which could be $O(m)$ in the worst case. We can get this down to $O(n)$ per iteration by maintaining the values of $f(C)$ for every fundamental cut $C$, and updating these values each time the algorithm updates potentials. Since there are $n-1$ cuts, to do this in $O(n)$ time requires us to be able to update $f(C)$ for a single cut in $O(1)$ time. To do this, we can precompute an $(n-1) \times (n-1)$ table with a row/column for each fundamental cut, where the $(C_1, C_2)$ entry is the amount by which the flow out of $C_2$ increases if we add 1 to the potential of every node in $C_1$. Let $H(C_1, C_2)$ denote this value. With this table, updating the value of $f(C)$ after a potential update step essentially reduces to a single table lookup, which takes $O(1)$ time.
    
    Finally, note that one can construct the $H(C_1, C_2)$ table in $O(n^2)$ time using results from \cite{Karger00}. In the language of Definitions 5.3 and 5.5 in that paper, we are trying to compute $C(v^\downarrow, w^\downarrow)$ for all vertices $v,w$ in the tree, where the edge weights are the reciprocals of the resistances. At the bottom of page 11, it states that the $n^2$ values $C(v, w^\downarrow)$ can be computed in $O(n^2)$ time. At the top of Page 12, it then says that we get the values of $C(v^\downarrow, w^\downarrow)$ using $n$ treefix sums. (Each treefix sum is the procedure described in Lemma 5.8, and takes $O(n)$ time.))
    
%     This is because each row of the table can be filled in $O(m)$ time, by starting from the leaves of the tree and working up to the root (in the same way we computed $b(C)$ for all the cuts $C$). In more detail, suppose we are computing the row corresponding to some cut $C$. That is, we want to compute how much the flow out of every cut $C'$ changes, if we increase the potentials of every node in $C$ by 1. Let $H(C, C')$ denote this value. Then we can compute $H(C, C')$ recursively as follows:
%     $$
%     H(C, C') = 
%     \begin{cases}
%     \sum_{w: vw \in E} \frac{\one[v \in C] - \one[w \in C]}{r(v,w)}, &\text{if $C' = \{v\}$ is a leaf} \\[5pt]
%   \sum_{w: vw \in E} \frac{\one[v \in C] - \one[w \in C]}{r(v,w)} + \sum_{\bar{C}} H(C, \bar{C}), &\text{if $C'$ is subtree rooted at $v$}
%     \end{cases}
%     $$
%     In the second part of the above expression, the sum is over all subtrees $\bar{C}$ of $C'$. The above recursive formula allows us to compute $H(C, C')$ for all cuts $C'$ in $O(m)$ time (for a fixed $C$). Hence computing $H(C, C')$ for all pairs $(C, C')$ (which corresponds to filling out the table), takes $O(mn)$ preprocessing time.  
    
    To summarize, we can run each iteration of the algorithm in $O(m)$ time, which can be reduced to $O(n)$ time if we precompute the $H(C_1, C_2)$ table, which incurs an overhead of $O(n^2)$ storage and $O(n^2)$ preprocessing time. 
    
    Suppose each iteration of the algorithm takes $O(I)$ time, and the algorithm uses $O(L)$ preprocessing time (not including the time needed to compute the low-stretch spanning tree). Then the total running time of the algorithm is $O(L + I \tau\ln(\frac{\tau}{\epsilon}) + m \log n \log\log n) = O(L + mI\log n \log\log n\log\frac{\tau}{\epsilon})$. 
    
    If we use the version which uses $O(n^2)$ preprocessing time and $O(n)$ time per iteration, then $L = O(n^2)$ and $I = O(n)$. This gives the running time of \ALGNAME\ to be $O(mn\log n \log\log n\log\frac{\tau}{\epsilon})$. 
    
    \end{proof}

\section{Missing Proofs from Section \ref{sec:pnorm}}
    \cycleiter*
    
    \begin{proof}[Proof of \Cref{thm:cycle_iter}]
 Let $\sE(\fv) := \frac{1}{p} \sum_{e \in E} r(e) \abs{f(e)}^p$ be the objective that we are trying to minimize. Suppose $\fv = \fv^{t-1}$ is the flow at the beginning of iteration $t$ of the algorithm, let $\fv' = \fv^t$ be the flow at the end of iteration $t$, and let $\fv^*$ be the optimal flow. Our goal will be to show that
 $$\E\left[ \sE(\fv) - \sE(\fv')  \right] \geq \frac{1}{\tau} \left(\sE(\fv) - \sE(\fv^*)\right).$$
 If we can show this, then using the same arguments as Corollary \ref{cor:final_gap}, this will immediately imply that in 
 $K = O
\left(
\tau\ln(\frac1\epsilon)
\ln\left(
    \frac{\sE(\fv^0) - \sE(\fv^*)}{\sE(\fv^*)}
\right)
\right)$ 
iterations, we will have a flow $\fv^K$ that satisfies $\E[\sE(\fv^K)] \leq (1 + \epsilon) \sE(\fv^*)$. To bound the initial energy $\sE(\fv^0)$, note that by \Cref{prop:flow_upperbound} and \Cref{prop:ener_lower_bound} we have the rough bounds
$$
\sE(\fv^0) \leq \frac1p n \norm{\rv}_\infty \norm{\bv}_1^p 
\quad \text{and} \quad 
\sE(\fv^*) \geq \frac1p \norm{\rv}_{-\infty}\left( \frac{\norm{\bv}_\infty}{n}\right)^p
$$
which together imply that $\sE(\fv^0) \leq n^{2p+1}R \sE(\fv^*)$. Thus  $\ln\left(
    \frac{\sE(\fv^0) - \sE(\fv^*)}{\sE(\fv^*)}
\right) = O\left(p\ln(n) + \ln (R) \right)$, matching the expression in the statement of the theorem. 

The remainder of this proof will be devoted to showing that $\E\left[ \sE(\fv) - \sE(\fv')  \right] \geq \frac{1}{\tau} \left(\sE(\fv) - \sE(\fv^*)\right)$, where 
$\tau = O(p2^{2p-1} \cdot m\log n \log\log n +  m(nR)^{\frac{1}{p-1}})$.
  
 Let $\Delta = \fv^* - \fv$. Then
\begin{align*}
    \sE(\fv + \Delta) - \sE(\fv)
    &= \frac{1}{p} \sum_e r(e) \left(\abs{f(e)+\Delta(e)}^p - \abs{f(e)}^p\right) \\
    &\geq \frac1p \sum_e \left( r(e) \left(p f(e)\abs{f(e)}^{p-2}\Delta(e) + 2^{-p}\left(\abs{f(e)}^{p-2}\Delta(e)^2 +\abs{\Delta(e)}^p \right) \right)\right) \\
    &= \sum_e r(e)f(e)\abs{f(e)}^{p-2}\Delta(e) + \frac{1}{p2^p} \sum_e r(e)\left(\abs{f(e)}^{p-2}\Delta(e)^2 + \abs{\Delta(e)}^p\right)
\end{align*}
Here, the inequality is using Lemma B.3 in \cite{KPSW19}. Also, note that $\sE(\fv+\Delta) - \sE(\fv) \leq 0$, because we are considering a minimization problem.

Write $\Delta = \sum_C \Delta_C\one_C$, where $\Delta_C \in \R$ and the sum is over all the fundamental cycles of the spanning tree $T = T^t$. Here, $\one_C$ is the indicator vector of $C$.  This decomposition is possible because the set of cycles of a spanning tree is a basis for the space of circulations. In particular, since every non-tree edge is contained in exactly one fundamental cycle, this implies that $\Delta_C = \Delta(u,v)$, where $(u,v)$ is the non-tree edge that uniquely determines $C$. 

Consider the cycle-toggling update $\fv \gets \fv + \alpha_C\Delta_C\one_C$, where $\alpha_C > 0$. This update preserves the feasibility of $\fv$, because pushing a constant amount of flow around a cycle does not change the amount of flow entering or leaving any vertex. The change in the objective value after the update is
\begin{align*}
    \sE(\fv+\alpha_C\Delta_C\one_C) - \sE(\fv)
    &= \frac1p \sum_{e \in C} r(e)\left(\abs{f(e)+\alpha_C\Delta_C}^p - \abs{f(e)}^p \right) \\
    &\leq \frac1p \sum_{e \in E} r(e)\left(pf(e)\abs{f(e)}^{p-2}\alpha_C\Delta_C + p2^{p-1}\left(\abs{f(e)}^{p-2}\alpha_C^2\Delta_C^2 + \alpha_C^p\abs{\Delta_C}^p \right) \right) \\
    &= \sum_{e \in C} r(e)f(e)\abs{f(e)}^{p-2}\alpha_C\Delta_C + 2^{p-1} \sum_{e \in C}r(e) \left( \abs{f(e)}^{p-2}\alpha_C^2\Delta_C^2 + \alpha_C^p\abs{\Delta_C}^p \right)
\end{align*}
Here, the inequality is using Lemma B.2 in \cite{KPSW19}.
Note that the more negative this quantity is, the more progress the update step makes. Suppose that with probability $p(C)$, we choose cycle $C$ and push $\alpha_C\Delta_C$ units of flow along it. We set $p(C) = \frac{1/\alpha_C}{\tau}$, where $\tau = \sum_C \frac{1}{\alpha_C}$. Then the expected change in the objective value is
\begin{align*}
    &\E\left[\sE(\fv + \alpha_C\Delta_C\one_C) - \sE(\fv) \right] \\
    &\leq\sum_C p(C)\alpha_C\Delta_C \sum_{e \in C} r(e)f(e)\abs{f(e)}^{p-2} + 2^{p-1} \sum_C p(C) \sum_{e \in C}r(e) \left( \abs{f(e)}^{p-2}\alpha_C^2\Delta_C^2 + \alpha_C^p\abs{\Delta_C}^p \right) \\
    &= \frac{1}{\tau}\sum_C \Delta_C \sum_{e \in C} r(e)f(e)\abs{f(e)}^{p-2} + \frac{2^{p-1}}{\tau} \sum_C  \sum_{e \in C}r(e) \left( \abs{f(e)}^{p-2}\alpha_C\Delta_C^2 + \alpha_C^{p-1}\abs{\Delta_C}^p \right) 
\end{align*}
Let us compare $\E\left[\sE(\fv + \alpha_C\Delta_C\one_C) - \sE(\fv) \right]$ to $\sE(\fv^*) - \sE(\fv)$.

\textbf{First-order terms}:
\begin{itemize}
    \item $\sE(\fv^*) -\sE(\fv)$:
    \begin{align*}
        \sum_e r(e)f(e)\abs{f(e)}^{p-2} \Delta(e)
    \end{align*}
    \item $\E\left[\sE(\fv + \alpha_C\Delta_C\one_C) - \sE(\fv) \right]$:
    \begin{align*}
        \frac{1}{\tau} \sum_C \Delta_C \sum_{e \in C} r(e)f(e)\abs{f(e)}^{p-2} 
        &= \frac{1}{\tau} \sum_e r(e)f(e)\abs{f(e)}^{p-2}  \sum_{C: e \in C} \Delta_C\\
        &= \frac{1}{\tau} \sum_e r(e)f(e)\abs{f(e)}^{p-2}\Delta(e)
    \end{align*}
    Hence the first-order terms match exactly, up to the factor of $\frac{1}{\tau}$. 
\end{itemize}

\textbf{Higher-order terms}:
\begin{itemize}
    \item $\sE(\fv^*) -\sE(\fv)$:
    \begin{align*}
        \frac{1}{p2^p} \sum_e r(e)\left(\abs{f(e)}^{p-2}\Delta(e)^2 + \abs{\Delta(e)}^p\right)
    \end{align*}
    \item $\E\left[\sE(\fv + \alpha_C\Delta_C\one_C) - \sE(\fv) \right]$:
    \begin{align*}
       & \frac{2^{p-1}}{\tau} \sum_C  \sum_{e \in C}r(e) \left( \abs{f(e)}^{p-2}\alpha_C\Delta_C^2 + \alpha_C^{p-1}\abs{\Delta_C}^p \right)  
    \end{align*}
    For a fundamental cycle $C = C(u,v)$, where $(u,v)$ is a non-tree edge, let us choose
    \begin{equation}
        \alpha_C \leq \min \left\{ \frac{1}{p2^{2p-1}}\frac{r(u,v)\abs{f(u,v)}^{p-2}}{\sum_{e \in C}r(e) \abs{f(e)}^{p-2}}, \,\left(\frac{1}{p2^{2p-1}} \cdot \frac{r(u,v)}{\sum_{e \in C} r(e)}\right)^{\frac{1}{p-1}}\right\}
    \end{equation}
    Then 
    \begin{align*}
        2^{p-1} \sum_C \sum_{e \in C}r(e)\abs{f(e)}^{p-2}\alpha_C\Delta_C^2
        &\leq \frac{1}{p2^p}\sum_{C = C(u,v)}r(u,v) \abs{f(u,v)}^{p-2}\Delta_{C(u,v)}^2 \\
        &= \frac{1}{p2^p}\sum_{e\not\in T}r(e) \abs{f(e)}^{p-2}\Delta(e)^2
    \end{align*}
    and
    \begin{align*}
        2^{p-1}\sum_C\sum_{e \in C}r(e)\alpha_C^{p-1}\abs{\Delta_C}^p
        &\leq \frac{1}{p2^p} \sum_{C = C(u,v)} r(u,v)\abs{\Delta_C}^p \\
        &= \frac{1}{p2^p} \sum_{e\not\in T} r(e)\abs{\Delta(e)}^p.
    \end{align*}
\end{itemize}
Combining, we get that 
\begin{align*}
    \E\left[ \sE(\fv) - \sE(\fv + \alpha_C\Delta_C\one_C)\right]
    &\geq \frac{1}{\tau} \left(\sE(\fv) - \sE(\fv^*)\right).
\end{align*}
Finally, we need to relate this to the progress made by the algorithm. The algorithm samples a fundamental cycle $C$ with probability $P(C)$ and updates $\fv' \gets \fv + \Delta^t\one_C$, where $\Delta^t$ is the solution to 
$$\sum_{(k,l) \in \vec{C}} r(i,j)(f(i,j) + \Delta^t)\abs{f(i,j) + \Delta^t}^{p-2}= 0.$$ Observe that
$$\Delta^t = \arg\max_\delta \left\{\sE(\fv)-\sE(\fv+\delta\one_{C})\right\},$$
which can be seen by taking the derivative of $\sE(\fv)-\sE(\fv+\delta\one_{C})$ (note this is a concave function of $\delta$) with respect to $\delta$, and setting it to 0. Thus the next iterate $\fv'$ of the algorithm satisfies
$$\E\left[\sE(\fv) - \sE(\fv')\right] 
\geq  \E\left[ \sE(\fv) - \sE(\fv + \alpha_C\Delta_C\one_C)\right]
\geq \frac{1}{\tau} \left(\sE(\fv) - \sE(\fv^*)\right),$$
which is what we wanted to show.

It remains to analyze the value of $\tau$. 
We have
\begin{align*}
    \tau
    &= \sum_C \frac{1}{\alpha_C}\\
    &= \sum_{C=C(u,v)} \max\left\{ {p2^{2p-1}}\frac{\sum_{e \in C}r(e) \abs{f(e)}^{p-2}}{r(u,v)\abs{f(u,v)}^{p-2}}, \,\left({p2^{2p-1}} \cdot \frac{\sum_{e \in C} r(e)}{r(u,v)}\right)^{\frac{1}{p-1}}\right\} \\
    &\leq {p2^{2p-1}} \underbrace{\sum_{C=C(u,v)} \frac{\sum_{e \in C}r(e) \abs{f(e)}^{p-2}}{r(u,v)\abs{f(u,v)}^{p-2}}}_{A} + \left(p2^{2p-1}\right)^{\frac{1}{p-1}} \underbrace{\sum_{C=C(u,v)} \left( \frac{\sum_{e \in C} r(e)}{r(u,v)}\right)^{\frac{1}{p-1}}}_{B}
\end{align*}
 Note that $A$ can be rearranged as
\begin{align*}
    A &= \sum_{(u,v)\not\in E(T)}\left(1 +  \frac{1}{r(u,v)\abs{f(u,v)}^{p-2}}\cdot \sum_{e \in P_T(u,v)} r(e)\abs{f(e)}^{p-2}\right) \\
    &= (m-n+1) + \st_T\left(G, {r(e)\abs{f(e)}^{p-2}}\right) - (n-1) \\
    &= m - 2n + 2 + \st_T\left(G, {r(e)\abs{f(e)}^{p-2}}\right)
\end{align*}
Here, $\st_T(G, w)$ represents the stretch of the spanning tree $T$ in the graph $G$ with respect to the weights $w$. On the other hand, $B$ can be trivially upper-bounded by
\begin{align*}
    B \leq m (nR)^{\frac{1}{p-1}},
\end{align*}
where $R = \frac{\max_e r(e)}{\min_e r(e)}$. Therefore, if we choose $T$ at each iteration to be a low-stretch spanning tree with respect to the weights ${r(e)\abs{f(e)}^{p-2}}$, then we can make $A \leq \wtd{O}(m)$, and hence
\begin{equation*}
    \tau \leq p2^{2p-1} \cdot \wtd{O}(m) + \left(p2^{2p-1}\right)^{\frac{1}{p-1}} \cdot m(nR)^{\frac{1}{p-1}}.
\end{equation*}    
    Observe that since $p \geq 2$, we have $\left(p2^{2p-1}\right)^{\frac{1}{p-1}} = O(1)$.

    \end{proof}

    \cutiter*
    
    \begin{proof}[Proof of \Cref{thm:cut_iter}]
 Suppose $\xv = \xv^{t-1}$ is the iterate at the beginning of iteration $t$ of the algorithm. Let $\xv' = \xv^t$ be the iterate at the end of iteration $t$, and let $\xv^*$ be the optimal solution. Our goal will be to show that
\begin{equation*}
    \E[\sB(\xv') - \sB(\xv)] \geq \frac{1}{\tau} \left( \sB(\xv^*) - \sB(\xv)\right),
\end{equation*}
for $\tau = O\left(q 2^{2q-1} \cdot m\log n \log\log n +  nRm^{\frac{1}{q-1}}\right)$. If we can show this, then the same arguments as \Cref{lem:gap_decreases} and \Cref{cor:final_gap} will show that after $K = O\left(\tau\ln(\frac1\epsilon)\right)$ iterations, we will have a potential vector $\xv^K$ that satisfies $\E[\sB(\xv^K)] \geq (1-\epsilon)\sB(\xv^*)$. The remainder of the proof will be devoted to proving that $\E[\sB(\xv') - \sB(\xv)] \geq \frac{1}{\tau} \left( \sB(\xv^*) - \sB(\xv)\right).$

Let $T = T^t$ be the spanning tree chosen in iteration $t$ of the algorithm. Without loss of generality, assume that $x(r) = x^*(r) = 0$ where $r$ is the root of $T$. Define $\Delta = \xv^*- \xv$. Decompose
\begin{align*}
\Delta = \sum_{C} \Delta_C\one_C,
\end{align*}
where the sum is over all fundamental cuts of the tree.
Recall that by convention, we direct all tree edges toward the root, and for an edge $(u,v) \in E(T)$, the fundamental cut associated with $(u,v)$ is the cut $C(u,v)$ that consists of all the vertices in the component of $T - (u, v)$ on the same side as $u$. Then for an edge $(u, v) \in E(T)$, we have
\begin{align*}
\Delta_{C(u,v)} = \Delta(u) - \Delta(v).
\end{align*}

% Assume that $1 < p \leq 2$. The primal problem is
% \begin{align*}
%     \min \quad &\frac{1}{p}\sum_{e} r(e)\abs{f(e)}^p \\
%     \text{s.t.} \quad &Af = b
% \end{align*}
% and its dual is
% \begin{align*}
%     \max_x \quad b^T x - \frac{p-1}{p} \sum_{ij \in E} \left(\frac{\abs{x(i) - x(j)}^p}{r_{ij}}\right)^{\frac{1}{p-1}}
% \end{align*}
% For clarity, define $w(i,j) = (\frac{1}{r_{ij}})^{\frac{1}{p-1}}$ and let $q := \frac{p}{p-1}$. Then the dual is 
% \begin{align*}
%     \max_x \quad b^T x - \frac{1}{q} \sum_{ij \in E} w(i,j)\abs{x(i) - x(j)}^q.
% \end{align*}
% The details of how the dual is derived are in Appendix \ref{app:dual_deriv}.
% Let $\sB(x)$ be the dual objective function.
% Suppose our current iterate is $x$, and let $\Delta := x^* - x$. Then
Recalling that $\sB(\,\cdot\,)$ is the dual objective function, we have
\begin{align*}
    \sB(\xv + \Delta) - \sB(\xv) 
    &= \bv^T\Delta  - \frac{1}{q}  \sum_{(i,j) \in \vec{E}} w(i,j)\left[ \abs{x(i)-x(j) +\Delta(i)-\Delta(j)}^q - \abs{x(i) - x(j)}^q\right] \\
    &\leq \bv^T \Delta - \sum_{(i,j) \in \vec{E}} w(i,j) (\Delta(i) - \Delta(j))(x(i) - x(j))\abs{x(i) - x(j)}^{q-2} \\
    & \qquad\qquad - \frac{1}{q}\sum_{(i,j) \in \vec{E}} w(i,j) \cdot 2^{-q} \cdot \left[(\Delta(i) - \Delta(j))^2\abs{x(i) - x(j)}^{q-2} + \abs{\Delta(i) - \Delta(j)}^q \right]
    % &= b^T\Delta - \sum_{ij \in E} w(i,j) (\Delta(i) - \Delta(j))(x(i) - x(j))^{q-1} - \frac{1}{q}\sum_{ij \in E} w(i,j) \sum_{k=2}^{q-1} \binom{q}{k} (\Delta(i) - \Delta(j))^k (x(i) - x(j))^{q-k} 
\end{align*}
To go from the first to the second line, we used the inequality $\abs{1 + \delta}^q - 1 - q\delta \geq 2^{-q}\left(\delta^2 + \abs{\delta}^q\right)$, which holds for all $\delta \in \R$ and $q \geq 2$ (Lemma B.3 in \cite{KPSW19}).

On the other hand, consider taking some fundamental cut $C=C(u,v)$ of $T$ and adding $\alpha_C\Delta_C$ to the potential of every vertex of $C$ for some $\alpha_C > 0$. Then
\begin{align*}
    \sB(\xv + \alpha_C\Delta_C\one_C) - \sB(\xv)
    &= \alpha_C\Delta_C\one_C^T\bv - \frac{1}{q} \sum_{ij \in \delta(C)} w(i,j)\left[ \abs{x(i) - x(j) + \alpha_C\Delta_C}^q - \abs{x(i) - x(j)}^q\right] \\
    &\geq \alpha_C\Delta_C\one_C^T\bv - \sum_{ij \in \delta(C)} w(i,j) \alpha_C\Delta_C(x(i)-x(j))\abs{x(i)-x(j)}^{q-2} \\
    &\qquad\qquad - \frac{1}{q} \sum_{ij \in \delta(C)} w(i,j) \cdot q2^{q-1}\left[ \alpha_C^2\Delta_C^2 \abs{x(i) - x(j)}^{q-2} +  \alpha_C^q\abs{\Delta_C}^q\right]
    % & \qquad\qquad- \frac{1}{q} \sum_{ij \in \delta(C)} w(i,j)\left[ \sum_{k=2}^q \binom{q}{k} (\alpha_C\Delta_C)^k(x(i)-x(j))^{q-k}\right] \\
\end{align*}
Here, we used the inequality $\abs{1+\delta}^q - 1 - q\delta \leq q2^{q-1}(\delta^2 + \abs{\delta}^q)$, which holds for all $\delta \in \R$ and $q \geq 1$ (Lemma B.2 in \cite{KPSW19}). Suppose that with probability $p(C)$, we choose cut $C$ and add $\alpha_C\Delta_C$ to the potential of every vertex in $C$. Then the expected increase in the dual objective is
\begin{align*}
   \E[ \sB(\xv + \alpha_C\Delta_C\one_C) - \sB(\xv)]
   &\geq \sum_C p(C)\alpha_C\Delta_C\one_C^Tb - \sum_Cp(C)\sum_{ij \in \delta(C)} w(i,j) \alpha_C\Delta_C(x(i)-x(j))\abs{x(i)-x(j)}^{q-2} \\
   & \qquad - \frac{1}{q} \sum_Cp(C) \sum_{ij \in \delta(C)} w(i,j) \cdot q2^{q-1}\left[ \alpha_C^2\Delta_C^2 \abs{x(i) - x(j)}^{q-2} +  \alpha_C^q\abs{\Delta_C}^q\right]
\end{align*}
Let us choose $p(C) = \frac{1/\alpha_C}{\tau}$, where $\tau = \sum_C \frac{1}{\alpha_C}$. Then
\begin{align*}
   \E[ \sB(\xv + \alpha_C\Delta_C\one_C) - \sB(\xv)]
   &\geq \frac{1}{\tau}\sum_C \Delta_C\one_C^Tb - \frac{1}{\tau}\sum_C \sum_{ij \in \delta(C)} w(i,j) \Delta_C(x(i)-x(j))\abs{x(i)-x(j)}^{q-2} \\
   & \qquad - \frac{1}{q}\cdot\frac{1}{\tau} \sum_C \sum_{ij \in \delta(C)} w(i,j) \cdot q2^{q-1}\left[ \alpha_C\Delta_C^2 \abs{x(i) - x(j)}^{q-2} +  \alpha_C^{q-1}\abs{\Delta_C}^q\right]
\end{align*}
We want to show that $\E[ \sB(\xv + \alpha_C\Delta_C\one_C) - \sB(\xv)]$ is greater than some fraction of $\sB(\xv + \Delta) - \sB(\xv) $. To do this, we compare the two expressions term by term.

\textbf{First-order terms:}
\begin{itemize}
    \item $\sB(\xv+ \Delta) - \sB(\xv)$: 
    \begin{align*}
        \bv^T\Delta - \sum_{ij \in E} w(i,j)(\Delta(i) - \Delta(j))(x(i) - x(j))\abs{x(i) - x(j)}^{q-2}.
    \end{align*}
    \item $\E[ \sB(\xv + \alpha_C\Delta_C\one_C) - \sB(\xv)]$:
    \begin{align*}
       & \frac{1}{\tau}\sum_C \Delta_C\one_C^T\bv - \frac{1}{\tau}\sum_C \sum_{ij \in \delta(C)} w(i,j) \Delta_C(x(i) - x(j))\abs{x(i)-x(j)}^{q-2} \\
       &=\frac{1}{\tau} \Delta^T \bv - \frac{1}{\tau} \sum_{ij \in E}w(i,j)(x(i) - x(j))\abs{x(i) - x(j)}^{q-2} \sum_{uv \in P_T(ij)} \Delta_{C(u,v)} \\
       &=\frac{1}{\tau} \Delta^T \bv - \frac{1}{\tau} \sum_{ij \in E}w(i,j)(x(i) - x(j))\abs{x(i) - x(j)}^{q-2} \sum_{uv \in P_T(ij)} (\Delta(u) - \Delta(v)) \\
       &=\frac{1}{\tau} \Delta^T \bv - \frac{1}{\tau} \sum_{ij \in E}w(i,j)(\Delta(i) - \Delta(j))(x(i) - x(j))\abs{x(i) - x(j)}^{q-2}
    \end{align*}
\end{itemize}
Thus the first-order terms match exactly up to the factor of $\frac{1}{\tau}$.

\textbf{Higher-order terms:} (Loss).
\begin{itemize}
    \item $\sB(\xv+ \Delta) - \sB(\xv)$: 
    \begin{align*}
        \frac{1}{q}\sum_{ij \in E} w(i,j) \cdot 2^{-q} \cdot \left[(\Delta(i) - \Delta(j))^2\abs{x(i) - x(j)}^{q-2} + \abs{\Delta(i) - \Delta(j)}^q \right]
    \end{align*}
    \item  $\E[ \sB(\xv + \alpha_C\Delta_C\one_C) - \sB(\xv)]$:
    \begin{align*}
       & \frac{1}{q}\cdot\frac{1}{\tau} \sum_C \sum_{ij \in \delta(C)} w(i,j) \cdot q2^{q-1}\left[ \alpha_C\Delta_C^2 \abs{x(i) - x(j)}^{q-2} +  \alpha_C^{q-1}\abs{\Delta_C}^q\right] 
    \end{align*}
    For a cut $C = C(u,v)$, let us choose
    \begin{align*}
        \alpha_C \leq \min\left\{ \frac{w(u,v)\abs{x(i) - x(j)}^{q-2}}{\sum_{ij \in \delta(C)} w(i,j)\abs{x(i) - x(j)}^{q-2}} \cdot \frac{1}{q}\cdot2^{1-2q}, \;\left(\frac{w(u,v)}{\sum_{ij \in \delta(C)}w(i,j)}\right)^{\frac{1}{q-1}}\left(\frac{1}{q}\right)^{\frac{1}{q-1}} 2^{\frac{1-2q}{q-1}}\right\}.
    \end{align*}
    Then 
    \begin{align*}
        &\sum_C \sum_{ij \in \delta(C)} w(i,j) \cdot q2^{q-1}\left[ \alpha_C\Delta_C^2 \abs{x(i) - x(j)}^{q-2}\right] \\
        &= \sum_C q2^{q-1}\alpha_C\Delta_C^2\sum_{ij \in \delta(C)} w(i,j)  \abs{x(i) - x(j)}^{q-2} \\
        &\leq \sum_{C = C(u,v)} 2^{-q} \Delta_C^2 w(u,v) \abs{x(u) - x(v)}^{q-2} \\
         &= \sum_{C = C(u,v)} 2^{-q} w(u,v)(\Delta(u) - \Delta(v))^2  \abs{x(u) - x(v)}^{q-2}
    \end{align*}
    and 
    \begin{align*}
        &\sum_C \sum_{ij \in \delta(C)} w(i,j) \cdot q2^{q-1}  \alpha_C^{q-1}\abs{\Delta_C}^q \\
        &= \sum_C q2^{q-1}\alpha_C^{q-1}\abs{\Delta_C}^q \sum_{ij \in \delta(C)} w(i,j) \\
        &\leq \sum_{C=C(u,v)} 2^{-q} \abs{\Delta_C}^qw(u,v) \\
        &= \sum_{C = C(u,v)} 2^{-q} w(u,v) \abs{\Delta(u) - \Delta(v)}^q.
    \end{align*}

\end{itemize}

\textbf{Putting things together.}
    Combining, we get
    \begin{align*}
        \E[ \sB(\xv + \alpha_C\Delta_C\one_C) - \sB(\xv)] &\geq 
        \frac{1}{\tau} \Delta^T\bv - \frac{1}{\tau} \sum_{ij \in E}w(i,j)(\Delta(i) - \Delta(j))(x(i) - x(j))\abs{x(i)-x(j)}^{q-2} \\
        & -\frac{1}{q}\cdot\frac{1}{\tau}\sum_{C = C(u,v)} 2^{-q} w(u,v)\left[(\Delta(u) - \Delta(v))^2  \abs{x(u) - x(v)}^{q-2} + w(u,v) \abs{\Delta(u) - \Delta(v)}^q\right] \\
        &\geq \frac{1}{\tau} \left(\sB(\xv^*) - \sB(\xv)\right).
    \end{align*}
    
    Finally, we need to relate this to the progress made by the algorithm. The algorithm samples a fundamental cut $C$ with probability $P(C)$ and updates $\xv' \gets \xv + \Delta^t\one_C$, where $\Delta^t$ is the solution to 
$$\sum_{k\in C^t, l\not\in C^t, kl \in E } w(k,l)(x^{t-1}(k) - x^{t-1}(l) +\Delta^t)\abs{x^{t-1}(k)-x^{t-1}(l)+\Delta^t}^{q-2} = b(C).$$
Observe that
$$\Delta^t = 
\arg\max_\delta \left\{\sB(\xv+\delta\one_{C}) - \sB(\xv)\right\},
$$
which can be seen by taking the derivative of $\sB(\xv+\delta\one_{C}) - \sB(\xv)$ (note this is a concave function of $\delta$) with respect to $\delta$, and setting it to 0. Thus the next iterate $\xv'$ of the algorithm satisfies
$$\E\left[\sB(\xv') - \sB(\xv)\right] 
\geq  \E[ \sB(\xv + \alpha_C\Delta_C\one_C) - \sB(\xv)]
\geq \frac{1}{\tau} \left(\sB(\xv^*) - \sB(\xv)\right),$$
which is what we wanted to show.
    
\textbf{What is $\tau$?}
We have
\begin{align*}
    \tau &= \sum_C \frac{1}{\alpha_C} \\
    &= \sum_{C = C(u,v)} \max\left\{ \frac{\sum_{ij \in \delta(C)} w(i,j)\abs{x(i) - x(j)}^{q-2}}{w(u,v)\abs{x(u) - x(v)}^{q-2}} \cdot q\cdot2^{2q-1}, \;\left(\frac{\sum_{ij \in \delta(C)}w(i,j)}{w(u,v)}\right)^{\frac{1}{q-1}}q^{\frac{1}{q-1}} \cdot2^{\frac{2q-1}{q-1}}\right\} \\[5pt]
    &\leq q\cdot2^{2q-1}\underbrace{\sum_{C = C(u,v)}  \frac{\sum_{ij \in \delta(C)} w(i,j)\abs{x(i) - x(j)}^{q-2}}{w(u,v)\abs{x(u) - x(v)}^{q-2}}}_{A}  + q^{\frac{1}{q-1}} \cdot2^{\frac{2q-1}{q-1}} \underbrace{\sum_{C = C(u,v)} \left(\frac{\sum_{ij \in \delta(C)}w(i,j)}{w(u,v)}\right)^{\frac{1}{q-1}}}_{B} 
\end{align*}
Note that 
\begin{align*}
    A &= \sum_{ij \in E} w(i,j)\abs{x(i) - x(j)}^{q-2}\sum_{uv \in P_T(ij)} \frac{1}{w(u,v)\abs{x(u) - x(v)}^{q-2}} \\
    &= \mathrm{st}_T\left(G,\; \frac{1}{w(i,j)\abs{x(i) - x(j)}^{q-2}}\right).
\end{align*}
Thus, if we pick $T$ to be a low-stretch spanning tree with respect to the weights $\frac{1}{w(i,j)(x(i) - x(j))^{q-2}}$, we can make $A \leq m\log n \log\log n$. 

On the other hand, letting $W = \frac{\max w(i,j)}{\min w(i,j)}$, we have the trivial upper bound that  $B \leq n (mW)^{\frac{1}{q-1}}$. Combining, we get
\begin{align*}
    \tau \leq q\cdot 2^{2q-1} \cdot m\log n \log\log n+ q^{\frac{1}{q-1}} \cdot2^{\frac{2q-1}{q-1}} \cdot n(mW)^{\frac{1}{q-1}}.
\end{align*}
Finally, recalling that  $w(i,j) = \left(\frac{1}{r_{ij}}\right)^{\frac{1}{p-1}} = \left(\frac{1}{r_{ij}}\right)^{q-1}$ gives $W^{\frac{1}{q-1}}$, we get that
$$
W = \frac{\max (1/r_{ij})^{q-1}}{\min (1/r_{ij})^{q-1}} = \left(\frac{\max r_{ij}}{\min r_{ij}}\right)^{q-1} 
= R^{q-1}.
$$
Combining this with the fact that $q^{\frac{1}{q-1}} \cdot2^{\frac{2q-1}{q-1}} \leq O(1)$ (since $q \geq 2$) gives the theorem statement.
\end{proof}

    \flowconvert*
    
    \begin{proof}
Let $\fv$ be the potential-defined flow with respect to $\xv$. That is, 
$$f(i,j) = w(i,j)(x(i)-x(j))\abs{x(i)-x(j)}^{q-2}$$ for all $(i,j) \in \vec{E}$. Since $\xv$ is not an optimal dual solution, $\fv$ is not a feasible $\bv$-flow. Route the residual supplies $\bv - \Am\fv$ along the edges of any spanning tree $T$; call the resulting flow $\tilde{\fv}$. Then $\fv_{T, \xv} = \fv + \tilde{\fv}$, and this is a feasible flow. Our goal is to show that $\sE(\fv + \tilde{\fv}) \leq (1+\epsilon)\sE(\fv^*)$.

The proof proceeds in 5 steps.
\begin{enumerate}
    \item Show that if $\sB(\xv) \geq (1-\epsilon')\sB(\xv^*)$ then $\xv$ and $\xv^*$ are close (in an appropriate metric).
    \item Show that $\fv$ and $\fv^*$ are close (again in an appropriate metric).
    \item Show that $\sE(\fv) \leq \sE(\fv^*)(1+\delta)^p$ for a suitable choice of $\delta>0$, 
    \item Show that the residuals $\norm{\bv - \Am\fv}$ are small, which implies that $\sE(\tilde{\fv})$ is small.
    \item Show that $\sE(\fv + \tilde{\fv}) \leq \sE(\fv^*)(1+\epsilon)$.
\end{enumerate}
\textbf{Step 1:} We use the following inequality, which is proved in the proof of \Cref{thm:cut_iter}. It holds for all possible values of $\xv$ and $\Delta$.
\begin{align*}
    \sB(\xv + \Delta) - \sB(\xv) 
    &\leq \bv^T \Delta - \sum_{(i,j) \in \vec{E}} w(i,j) (\Delta(i) - \Delta(j))(x(i) - x(j))\abs{x(i) - x(j)}^{q-2} \\
    & \qquad\qquad - \frac{1}{q}\sum_{(i,j) \in \vec{E}} w(i,j) \cdot 2^{-q} \cdot \left[(\Delta(i) - \Delta(j))^2\abs{x(i) - x(j)}^{q-2} + \abs{\Delta(i) - \Delta(j)}^q \right]
\end{align*}
Substituting $\xv \gets \xv^*$ and $\Delta \gets \xv-\xv^*$, we get
\begin{align*}
    \sB(\xv) - \sB(\xv^*) 
    &\leq \bv^T \Delta - \sum_{(i,j) \in \vec{E}} w(i,j) (\Delta(i) - \Delta(j))(x^*(i) - x^*(j))\abs{x^*(i) - x^*(j)}^{q-2} \\
    & \qquad\qquad - \frac{1}{q}\sum_{(i,j) \in \vec{E}} w(i,j) \cdot 2^{-q} \cdot \left[(\Delta(i) - \Delta(j))^2\abs{x^*(i) - x^*(j)}^{q-2} + \abs{\Delta(i) - \Delta(j)}^q \right]
\end{align*}
Since $\xv^*$ is an optimal dual solution, we know $\nabla \sB(\xv^*) = 0$. This implies that 
$$b(i) = \sum_{j: ij \in E} w(i,j)(x^*(i) - x^*(j))\abs{x^*(i) - x^*(j)}^{q-2} \quad \forall \; i \in V,$$
which gives that for any $\Delta \in \R^V$, we have
$$
\bv^T \Delta = \sum_{(i,j) \in \vec{E}} w(i,j) (\Delta(i) - \Delta(j))(x^*(i) - x^*(j))\abs{x^*(i) - x^*(j)}^{q-2}.
$$
Plugging this into the previously displayed inequality, we get
\begin{align*}
\sB(\xv^*) - \sB(\xv) &\geq \frac{1}{q}\sum_{(i,j) \in \vec{E}} w(i,j) \cdot 2^{-q} \cdot \left[(\Delta(i) - \Delta(j))^2\abs{x(i) - x(j)}^{q-2} + \abs{\Delta(i) - \Delta(j)}^q \right] \\
&\geq \frac{1}{q} \cdot 2^{-q}\sum_{(i,j) \in \vec{E}} w(i,j) \cdot  \abs{\Delta(i) - \Delta(j)}^q.
\end{align*}
Thus,
\begin{equation}
\label{eq:xerr}
    \sum_{(i,j) \in \vec{E}} w(i,j) \cdot  \abs{\Delta(i) - \Delta(j)}^q 
    \leq q2^q \left(\sB(\xv^*) - \sB(\xv)\right) 
    \leq q2^q \epsilon' \sB(\xv^*).
\end{equation}
It follows that $\xv$ and $\xv^*$ are ``close'' in the sense of the above formula.

\textbf{Step 2:}  We now use the fact that $\xv$ and $\xv^*$ are ``close" (in the sense of the bound derived in Step 1), to show that $\fv$ and $\fv^*$ are close. Specifically, we will derive an upper bound on $\abs{f(i,j) - f^*(i,j)}$ for all edges $(i,j)$. Recall that $\fv$ is defined via
$$f(i,j) = w(i,j)(x(i)-x(j))\abs{x(i)-x(j)}^{q-2},$$
whereas
$$f^*(i,j) = w(i,j)(x^*(i) - x^*(j))\abs{x^*(i)-x^*(j)}^{q-2}.$$
Inequality \ref{eq:xerr} implies that for every edge $ij \in E$, we have
$$w(i,j)\abs{x(i)-x(j) - (x^*(i) - x^*(j))}^q \leq \epsilon'q2^q\sB(\xv^*).$$
Comparing $f(i,j)$ and $f^*(i,j)$, we have
\begin{align*}
    \abs{f(i,j)-f^*(i,j)}
    &= w(i,j)\abs{(x(i)-x(j))\abs{x(i)-x(j)}^{q-2} - (x^*(i) - x^*(j))\abs{x^*(i)-x^*(j)}^{q-2}} \\
    &\leq w(i,j)(q-1)\max\left\{\abs{x(i)-x(j)}^{q-2}, \abs{x^*(i)-x^*(j)}^{q-2}\right\}\abs{x(i)-x(j) - (x^*(i) - x^*(j))} \\
    &\leq w(i,j)(q-1)\max\left\{\abs{x(i)-x(j)}^{q-2}, \abs{x^*(i)-x^*(j)}^{q-2}\right\} \left(\frac{\epsilon'q2^q\sB(\xv^*)}{w(i,j)}\right)^{\frac{1}{q}}.
\end{align*}
Here, the second line is using the inequality $\abs{a\abs{a}^{k} -  b\abs{b}^{k}} \leq (k+1)\max\left\{\abs{a}^{k}, \abs{b}^k\right\}\cdot\abs{a-b}$, which is proved as Proposition \ref{prop:abs_diff} in Appendix \ref{sec:flow_convert_lems}. Also, by the triangle inequality, we have
$$\abs{x(i) - x(j)} \leq \abs{x^*(i)-x^*(j)} + \abs{x(i)-x(j) - (x^*(i) - x^*(j))} \leq \abs{x^*(i)-x^*(j)} + \left(\frac{\epsilon'q2^q\sB(\xv^*)}{w(i,j)}\right)^{\frac{1}{q}}.$$
Furthermore, we have by Proposition \ref{prop:flow_upperbound} in Appendix \ref{sec:flow_convert_lems} that $\abs{x^*(i) - x^*(j)} \leq r(i,j)(nR)^{\frac{p-1}{p}}\norm{\bv}_1^{p-1}$. Plugging these inequalities back into what we have above, we get
\begin{align*}
    \abs{f(i,j)-f^*(i,j)}
    &\leq w(i,j)(q-1)\left\{ r(i,j)(nR)^{\frac{p-1}{p}}\norm{\bv}_1^{p-1}+ \left(\frac{\epsilon'q2^q\sB(\xv^*)}{w(i,j)}\right)^{\frac{1}{q}}\right\}^{q-2} \left(\frac{\epsilon'q2^q\sB(\xv^*)}{w(i,j)}\right)^{\frac{1}{q}}.
\end{align*}
From now on, let us denote $\alpha := \norm{\fv - \fv^*}_\infty$. In the above calculations, we derived an upper bound on $\alpha$.

\textbf{Step 3:} 
Now we show that because $\fv$ and $\fv^*$ are close,  $\sE(\fv)$ cannot be much larger than $\sE(\fv^*)$. Specifically, we will prove that $\sE(\fv) \leq \sE(\fv^*)(1+\delta)^p$ for a suitable $\delta > 0$. 
\begin{align*}
    \sE(\fv) &= \frac{1}{p} \sum_{e} r(e) \abs{f(e)}^p \\
    &= \frac{1}{p} \sum_e r(e) \abs{f(e) - f^*(e) + f^*(e)}^p \\
    &\leq  \left[\left(\frac{1}{p}\sum_e r(e) \abs{f(e) - f^*(e)}^p\right)^{\frac{1}{p}} + \left(\frac{1}{p}\sum_e r(e)\abs{f^*(e)}^p\right)^{\frac{1}{p}} \right]^p \quad \text{(Proposition \ref{prop:ener_add} in Appendix \ref{sec:flow_convert_lems})}\\
    &\leq \left[\left(\frac{1}{p}m\norm{\rv}_\infty \alpha^p\right)^{\frac{1}{p}} + \sE(\fv^*)^{\frac{1}{p}} \right]^p \\
    &= \sE(\fv^*)\left(1 + \left(\frac{m\norm{\rv}_\infty \alpha^p}{p\sE(\fv^*)}\right)^{\frac{1}{p}}  \right)^p \\
    &\leq \sE(\fv^*)\left(1 + \left(\frac{m\norm{\rv}_\infty \alpha^pn^p}{\norm{\rv}_{-\infty} \norm{\bv}_\infty^p}\right)^{\frac{1}{p}}  \right)^p \quad \text{(Proposition \ref{prop:ener_lower_bound} in Appendix \ref{sec:flow_convert_lems})} \\
    &= \sE(\fv^*)
    \left(1 + 
    \left(mR\right)^{\frac{1}{p}} 
    \cdot \frac{\alpha n}{\norm{\bv}_\infty}
    \right)^p.
\end{align*}

\textbf{Step 4:} We will now show that the residual flow $\sE(\tilde{\fv})$ has small energy. Since $\alpha = \norm{\fv - \fv^*}_\infty$, this implies $\norm{\bv - \Am\fv}_\infty \leq \alpha n$, because each vertex has at most $n$ edges incident to it, and the difference between the value of $\fv$ and $\fv^*$ on an edge is at most $\alpha$. Recall that $\tilde{\fv}$ is a flow that routes the residual supplies $\bv - \Am\fv$ on some spanning tree $T$. Thus $\abs{\tilde{f}(i,j)} = \sum_{k \in C(i,j)} (\bv - \Am\fv)(k) \leq \norm{\bv - \Am\fv}_1 $ (where $C(i,j)$ is the fundamental cut of $T$ determined by $(i,j)$.) This implies
\begin{align*}
    \sE(\tilde{\fv})
    &= \frac{1}{p} \sum_{e \in T} r(e)\abs{\tilde{f}(e)}^p \\
    &\leq \frac{1}{p}(n-1) \norm{\rv}_\infty \norm{\bv - \Am\fv}_1^p \\
    &\leq \frac{1}{p}n \norm{\rv}_\infty \left(n\norm{\bv - \Am\fv}_\infty\right)^p \\
    &\leq \frac{1}{p}n \norm{\rv}_\infty \left(\alpha n^2\right)^p \\
    &\leq \sE(\fv^*) \cdot \frac{n \norm{\rv}_\infty \left(\alpha n^2\right)^pn^p}{\norm{\rv}_{-\infty}\norm{\bv}_\infty^p} \quad \text{(Proposition \ref{prop:ener_lower_bound} in Appendix \ref{sec:flow_convert_lems})}\\
    &= \sE(\fv^*) \cdot nR \cdot \frac{\left(\alpha n^3\right)^p}{\norm{\bv}_\infty^p}.
\end{align*}

\textbf{Step 5:} Finally, we now show that $\sE(\fv + \tilde{\fv})$ is not much larger than $\sE(\fv^*)$. 
\begin{align*}
    \sE(\fv + \tilde{\fv})
    &\leq \left[ \sE(\fv)^{\frac1p} + \sE(\tilde{\fv})^{\frac1p}\right]^p \quad \text{(Proposition \ref{prop:ener_add} in Appendix \ref{sec:flow_convert_lems})}\\
    &\leq 
    \left[
    \sE(\fv^*)^{\frac1p}\left(1 + 
    \left(mR\right)^{\frac{1}{p}} 
    \cdot \frac{\alpha n}{\norm{\bv}_\infty}
    \right)
    + 
    \sE(\fv^*)^{\frac1p}
    (nR)^{\frac1p} \cdot \frac{ \left(\alpha n^3\right)}{\norm{\bv}_\infty}
    \right]^p \\
    &\leq \sE(\fv^*)
    \left[
    1
     + 
    2\left(mR\right)^{\frac{1}{p}} 
    \cdot \frac{\alpha n^3}{\norm{\bv}_\infty}
    \right]^p 
\end{align*}
Suppose $\alpha$ is small enough so that 
$$ 2\left(mR\right)^{\frac{1}{p}} 
    \cdot \frac{\alpha n^3}{\norm{\bv}_\infty}
    \leq \min\{\epsilon/3, 1\}.
    $$
    Then we would have 
    \begin{align*}
    \sE(\fv + \tilde{\fv}) 
   & \leq 
    \sE(\fv^*)
    \left[
    1
     + 
    2\left(mR\right)^{\frac{1}{p}} 
    \cdot \frac{\alpha n^3}{\norm{\bv}_\infty}
    \right]^p  \\
   & \leq  \sE(\fv^*)
    \left[
    1
     + 
    2\left(mR\right)^{\frac{1}{p}} 
    \cdot \frac{\alpha n^3}{\norm{\bv}_\infty}
    \right]^2 \\
    &\leq \sE(\fv^*)\left[
    1
     + 
    6\left(mR\right)^{\frac{1}{p}} 
    \cdot \frac{\alpha n^3}{\norm{\bv}_\infty}
    \right] \\
    &\leq \sE(\fv^*)\left[
    1+ \epsilon
    \right].
    \end{align*}
Thus to obtain $\sE(\fv + \tilde{\fv}) \leq \sE(\fv^*)$, it suffices to make $\alpha$ small enough so that $ 2\left(mR\right)^{\frac{1}{p}} 
    \cdot \frac{\alpha n^3}{\norm{\bv}_\infty}
    \leq \min\{\epsilon/3, 1\}.$
    
    Rearranging for $\alpha$, the previous inequality is equivalent to
    $$\alpha
    \leq 
    \frac{\min\{\epsilon/3, 1\}{\norm{\bv}_\infty}}{2n^3(mR)^{\frac{1}{p}}}.
    $$
    Recalling the upper bound on $\alpha$ obtained at the end of Step 2, it suffices to have
     \begin{equation}
     %\label{eq:flow_convert_want}
     \max_{ij}
     \left\{
     w(i,j)(q-1)\left\{ r(i,j)(nR)^{\frac{p-1}{p}}\norm{\bv}_1^{p-1}+ \left(\frac{\epsilon'q2^q\sB(\xv^*)}{w(i,j)}\right)^{\frac{1}{q}}\right\}^{q-2} \left(\frac{\epsilon'q2^q\sB(\xv^*)}{w(i,j)}\right)^{\frac{1}{q}}
     \right\}
     \leq 
     \frac{\min\{\epsilon/3, 1\}{\norm{\bv}_\infty}}{2n^3(mR)^{\frac{1}{p}}}.
     \end{equation}

     For clarity, define $\delta(i,j) =\left(\frac{\epsilon'q2^q\sB(\xv^*)}{w(i,j)}\right)^{\frac{1}{q}} $. Then the above inequality can be rewritten as 
     \begin{equation}
     \label{eq:flow_convert_want}
     \max_{ij}
     \left\{
     w(i,j)(q-1)\left\{ r(i,j)(nR)^{\frac{p-1}{p}}\norm{\bv}_1^{p-1}+ \delta(i,j)\right\}^{q-2} \delta(i,j)
     \right\}
     \leq 
     \frac{\min\{\epsilon/3, 1\}{\norm{\bv}_\infty}}{2n^3(mR)^{\frac{1}{p}}}.
     \end{equation}
     Next we will try to simplify the left-hand side of (\ref{eq:flow_convert_want}). First we will choose $\epsilon'$ small enough so that $\delta(i,j) \leq r(i,j)(nR)^{\frac{p-1}{p}}\norm{\bv}_1^{p-1}$. This will allow us to replace the sum of two terms on the LHS that is raised to the $(q-2)$th power with $2r(i,j)(nR)^{\frac{p-1}{p}}\norm{\bv}_1^{p-1}$. 
     Note that we have the following upper bounds for $\delta(i,j)$:
     \begin{align*}
         \delta(i,j)
         &= \left(\frac{\epsilon'q2^q\sB(\xv^*)}{w(i,j)}\right)^{\frac{1}{q}} \\
         &\leq \left(\frac{\epsilon'q2^q\frac{1}{p}n\norm{\rv}_\infty\norm{\bv}_1^p}{w(i,j)}\right)^{\frac{1}{q}} \quad \text{(Proposition \ref{prop:flow_upperbound} in Appendix \ref{sec:flow_convert_lems})} \\
         &\leq \left(\frac{\epsilon'q2^qn\norm{\rv}_\infty\norm{\bv}_1^p}{w(i,j)}\right)^{\frac{1}{q}} 
     \end{align*}
     
      For $\delta(i,j) \leq r(i,j)(nR)^{\frac{p-1}{p}}\norm{\bv}_1^{p-1}$ to hold, using the upper bound on $\delta(i,j)$ derived above it suffices to make it so that
      \begin{align*}
          \left(\frac{\epsilon'q2^qn\norm{\rv}_\infty\norm{\bv}_1^p}{w(i,j)}\right)^{\frac{1}{q}} \leq r(i,j)(nR)^{\frac{p-1}{p}}\norm{\bv}_1^{p-1},
      \end{align*}
      or in other words, 
      \begin{align*}
          \epsilon'
          &\leq \frac{r(i,j)^q(nR)^{q\cdot\frac{p-1}{p}}\norm{\bv}_1^{q(p-1)}w(i,j)}{q2^qn\norm{\rv}_\infty\norm{\bv}_1^p} \\
          &= \frac{r(i,j)(nR)\norm{\bv}_1^{p}}{q2^qn\norm{\rv}_\infty\norm{\bv}_1^p} \\
          &=  \frac{r(i,j)R}{q2^q\norm{\rv}_\infty}.
      \end{align*}
      Since $r(i,j)R \geq \norm{\rv}_\infty$, it suffices to choose $\epsilon' \leq \frac{1}{q2^q}$. Assuming we choose $\epsilon' \leq \frac{1}{q2^q}$, and recalling that $w(i,j) = r(i,j)^{-\frac{1}{p-1}} = r(i,j)^{-(q-1)}$, we then have
      \begin{align*}
           &\max_{ij}
     \left\{
     w(i,j)(q-1)\left\{ r(i,j)(nR)^{\frac{p-1}{p}}\norm{\bv}_1^{p-1}+ \delta(i,j)\right\}^{q-2} \delta(i,j)
     \right\} \\
     &\leq \max_{ij}
     \left\{
     w(i,j)q2^{q-2}\left\{ r(i,j)(nR)^{\frac{p-1}{p}}\norm{\bv}_1^{p-1} \right\}^{q-2} \delta(i,j)
     \right\} \\
     &\leq q2^q\max_{ij}
     \left\{
     w(i,j) r(i,j)^{q-2}(nR)^{(q-2)\frac{p-1}{p}}\norm{\bv}_1^{(q-2)(p-1)} \delta(i,j)
     \right\} \\
     &= q2^q(nR)^{(q-2)\frac{p-1}{p}}\norm{\bv}_1^{(q-2)(p-1)}\max_{ij}
     \left\{
     w(i,j) r(i,j)^{q-2} \delta(i,j)
     \right\} \\
     &= q2^q(nR)^{\frac{2-p}{p}}\norm{\bv}_1^{2-p}\max_{ij}
     \left\{
     \frac{ \delta(i,j)}{r(i,j)}
     \right\} \\
     &\leq q2^q(nR)^{\frac{2-p}{p}}\norm{\bv}_1^{2-p}\max_{ij}
     \left\{
     \left(\frac{\epsilon'q2^qn\norm{\rv}_\infty\norm{\bv}_1^p}{w(i,j)r(i,j)^q}\right)^{\frac{1}{q}}
     \right\} \\
     &= q2^q(nR)^{\frac{2-p}{p}}\norm{\bv}_1^{2-p}\max_{ij}
     \left\{
     \left(\frac{\epsilon'q2^qn\norm{\rv}_\infty\norm{\bv}_1^p}{r(i,j)}\right)^{\frac{1}{q}}
     \right\} \\
     &\leq  q2^q(nR)^{\frac{2-p}{p}}\norm{\bv}_1^{2-p}
     \left(
     \epsilon'q2^qnR\norm{\bv}_1^p
     \right)^{\frac{1}{q}}
     \\
     &= (q2^q)^{1+\frac1q}(nR)^{\frac{1}{p}}\norm{\bv}_1
     \left(
     \epsilon'
     \right)^{\frac{1}{q}} \\
     &\leq (q2^q)^{1+\frac{1}{q}}(nR)^{\frac{1}{p}}n\norm{\bv}_\infty
     \left(
     \epsilon'
     \right)^{\frac{1}{q}}
      \end{align*}
    Therefore, for \Cref{eq:flow_convert_want} to hold, it suffices to have
    $$
    (q2^q)^{1+\frac1q}(nR)^{\frac{1}{p}}n\norm{\bv}_\infty
     \left(
     \epsilon'
     \right)^{\frac{1}{q}}
     \leq 
     \frac{\min\{\epsilon/3, 1\}{\norm{\bv}_\infty}}{2n^3(mR)^{\frac{1}{p}}}.
    $$
    Rearranging for $\epsilon'$, we get
    $$
   \left( \epsilon' \right) ^{\frac1q}
    \leq 
    \frac{\min\{\epsilon/3, 1\}}{2n^4(mR)^{\frac{1}{p}} (q2^q)^{1+\frac1q}(nR)^{\frac{1}{p}}}.
    $$
    Therefore, for \Cref{eq:flow_convert_want} to hold it suffices to choose $(\epsilon')^{\frac{1}{q}}$ equal to the right-hand side of the inequality above. 
    %in which case we get
    % \begin{align*}
    % \ln \frac{1}{\epsilon'} 
    % &= 
    % q\left(
    % \ln2+4\ln n+\frac{1}{p}\ln(mnR^2) 
    % +\left(1 + \frac1q\right)\ln(q2^q)
    % + \ln \left(\max\left\{\frac{3}{\epsilon}, 1 \right\}\right)
    % \right) \\
    % &= O\left(q
    % \ln\left(mnRq2^q\max\left\{\frac3\epsilon, 1\right\}\right)\right).
    % \end{align*}
    
    \end{proof}

\input{misclemmas}

\input{previous}

%% file: misclemmas.tex
\section{Miscellaneous Lemmas}
\label{sec:misc}
In this section we collect some lemmas that are used elsewhere in the paper.

\subsection{Inequalities used in proof of Lemma \ref{lem:flow_convert}}
\label{sec:flow_convert_lems}
\begin{prop}
\label{prop:abs_diff}
For all $a, b \in \R$ and $k \geq 0$, 
$$\abs{a\abs{a}^k - b\abs{b}^k} \leq (k+1)
\max\left\{\abs{a}^k, \abs{b}^k\right\} \cdot\abs{a-b}.$$
\end{prop}
\begin{proof}
First, note that without loss of generality we may assume that either $a,b \geq 0$ or $a \geq 0 $ and $b \leq 0$. 

First, suppose that $a, b \geq 0$. Then the inequality we are trying to prove becomes
\begin{align*}
    \abs{a^{k+1} - b^{k+1}}
    \leq (k+1)\max\left\{{a}^k, {b}^k\right\} \cdot\abs{a-b}.
\end{align*}
Without loss of generality, we may assume that $a > b$. Applying the change of variables $t = \frac{a}{b}$, the inequality is equivalent to 
\begin{align*}
    t^{k+1} - 1
    \leq (k+1)t^k(t-1).
\end{align*}
Consider the function $g(t) = t^{k+1}$. By the convexity of $g$, we know that for all $h$, we have
$$g(t+h) - g(t) \geq g'(t)h.$$
Plugging in $h = -(t-1)$ gives
$$1 - t^{k+1} \geq -(k+1)t^k(t-1),$$
which is what we wanted to show.
 
Next, suppose that $a \geq 0$ and $b \leq 0$. Then
\begin{align*}
     \abs{a\abs{a}^k - b\abs{b}^k}
     &=  a^{k+1} + \abs{b}^{k+1} \\
     &\leq \max\left\{a^k, \abs{b}^k\right\}\left(a + \abs{b}\right)  \\
     &= \max\left\{\abs{a}^k, \abs{b}^k\right\} \cdot\abs{a-b} \\
     &\leq (k+1) \max\left\{\abs{a}^k, \abs{b}^k\right\} \cdot\abs{a-b}.
\end{align*}

\end{proof}

\begin{prop}
\label{prop:flow_upperbound}
Let $f^*$ be the optimal solution to the minimum $p$-norm flow problem. Then
$$
\sE(\fv^*) \leq \frac{1}{p}n\norm{\rv}_\infty\norm{\bv}_1^p
$$
In addition, for all $(i,j) \in \vec{E}$,
$$\abs{f^*(i,j)} \leq (nR)^{\frac{1}{p}} \norm{\bv}_1,$$
where $R = \frac{\max_e r(e)}{\min_e r(e)}$. This implies that
$$\abs{x^*(i) - x^*(j)} \leq r(i,j)(nR)^{\frac{p-1}{p}}\norm{\bv}_1^{p-1} \quad \forall \; (i,j) \in \vec{E}.$$

\end{prop}
\begin{proof}
This upper bound comes from upper-bounding the energy of an optimal flow. For any spanning tree $T$, the unique $\bv$-flow $\fv_T$ that is supported on the edges of $T$ is a feasible flow. This flow has $\fv_T(u,v) = \sum_{i \in C(u,v)} b(i)$ for all $(u,v) \in E(T)$, where $C(u,v)$ is the fundamental cut of $T$ containing $u$. Hence $\norm{\fv_T}_\infty \leq \norm{\bv}_1$, which implies that 
$$\sE(\fv^*) \leq \sE(\fv_T) \leq \frac{1}{p}n\norm{\rv}_\infty\norm{\bv}_1^p.$$
On the other hand, for any edge $(i,j)$, 
$$\sE(\fv^*) \geq \frac{1}{p} r(i,j)\abs{f^*(i,j)}^p \geq \frac{1}{p} \norm{\rv}_{-\infty} \abs{f^*(i,j)}^p.$$
Combining these inequalities gives $\abs{f^*(i,j)} \leq (nR)^{\frac{1}{p}} \norm{\bv}_1$, as claimed.

Using this gives us
$$\abs{x^*(i) - x^*(j)} = r(i,j)\abs{f^*(i,j)}^{p-1} \leq r(i,j)(nR)^{\frac{p-1}{p}}\norm{\bv}_1^{p-1}.$$
\end{proof}

\begin{prop}
\label{prop:ener_add}
For any  $\fv, \fv' \in \R^E$, we have 
$$\sE(\fv + \fv') \leq \left( \sE(\fv)^{\frac1p} + \sE(\fv')^{\frac1p} \right)^p.$$
\end{prop}

\begin{proof}
For all $e \in E$, define $g(e) = r(e)^{\frac{1}{p}}f(e)$ and $g'(e) = r(e)^{\frac{1}{p}} f'(e)$. Then 
$$\sE(\fv + \fv') = \frac{1}{p}\norm{\gv + \gv'}_p^p \leq \frac{1}{p}\left(\norm{\gv}_p + \norm{\gv'}_p\right)^p= \left( \sE(\fv)^{\frac1p} + \sE(\fv')^{\frac1p} \right)^p.$$
\end{proof}

\begin{prop}
\label{prop:ener_lower_bound}
For any feasible $\bv$-flow $\fv$, we have $\sE(\fv) \geq \frac{1}{p}\left(\frac{\norm{\bv}_\infty}{n}\right)^p \norm{\rv}_{-\infty}$. 
% Here, $\norm{\rv}_{-\infty} = \min_e\{r(e)\}$. 
\end{prop}
\begin{proof}
Consider $i \in V$ with $b(i) = \norm{\bv}_\infty$. Vertex $i$ has at most $n$ edges leaving it, and the sum of the flow values on these edges is $b(i)$. Thus there must be some edge $(i,j) \in \vec{E}$ with $\abs{f(i,j)} \geq \frac{b(i)}{n}$.  The contribution of this edge to the total energy of $\fv$ is already 
$$\frac{1}{p} r(i,j)\norm{f(i,j)}^p \geq \frac{1}{p} \norm{\rv}_{-\infty}\left(\frac{b(i)}{n}\right)^p. $$
\end{proof}

\subsection{Spectral Approximations}

Let $\Am, \Bm$ be $n \times n$ symmetric, positive semidefinite matrices. We say that $\Bm$ is a $\gamma$-spectral sparsifier of $\Am$ if
\[
\left( 1 - \gamma \right) \Am
\preceq
\Bm
\preceq
\left( 1 + \gamma \right) \Am.
\]

% \begin{claim}
% If $\Am$ and $\Bm$ are full rank,
% $\Am \approx_{\gamma} \Bm$, then for any vector $\xv$, we have
% \[
% \norm{\xv - \Bm^{-1} \Am \xv}_{\Am}
% \leq
% \sqrt{2\gamma}
% \norm{\xv}_{\Am}
% \]
% or equivalently, for any vector $\bv$,
% \[
% \norm{\Am^{-1} \bv - \Bm^{-1} \bv}_{\Am}
% \leq
% \sqrt{2\gamma}
% \norm{\Am^{-1} \bv}_{\Am}
% \]
% TODO: maybe this is true with just $O(\gamma)$ too...
% \end{claim}
% \begin{proof}

% \end{proof}
\begin{prop}[Spectral Approximations]
\label{prop:spec_approx}
Suppose $\Am, \Bm \in \mathbb{S}^n_+$ and $\left( 1 - \gamma \right) \Am
\preceq
\Bm
\preceq
\left( 1 + \gamma \right) \Am$. Let $\xv, \yv, \zv, \bv \in \R^n$. Then the following hold:
\begin{enumerate}
    \item $(1-\gamma) \norm{\xv}_{\Am}^2 \leq \norm{\xv}_{\Bm}^2 \leq (1+\gamma)\norm{\xv}_{\Am}^2$
    \item If $\Am\xv = \bv$ and $\Bm\yv = \bv$, then $\norm{\xv - \yv}_{\Am}^2 \leq h(\gamma) \norm{\xv}_{\Am}^2$, where $h(\gamma) = \frac{\gamma^2}{(1-\gamma)^2}$. 
\end{enumerate}
\end{prop}

\begin{proof}
The first one is by definition of $\norm{\xv}_{\Am}^2 = \xv^{\top} \Am \xv$.

For the second one, first we claim that it is sufficient to prove that $\norm{\Am^\dag \bv - \Bm^\dag \bv}_{\Am}^2 \leq h(\gamma) \norm{\Am^\dag \bv}_{\Am}^2$. This is because in general, we have $\xv = \Am^\dag \bv + \uv$, and $\yv = \Bm^\dag \bv + \vv$, for some  $\uv \in \Null(\Am)$ and $\vv \in \Null(\Bm)$. Moreover, the condition 
$
\left( 1 - \gamma \right) \Am
\preceq
\Bm
\preceq
\left( 1 + \gamma \right) \Am
$
implies that $\Null(\Am) = \Null(\Bm)$. Hence, $\norm{\xv - \yv}_{\Am}^2 = \norm{\Am^\dag \bv - \Bm^\dag \bv}_{\Am}^2$, and $\norm{\xv}_{\Am}^2 =  \norm{\Am^\dag \bv}_{\Am}^2$. 

Next, we expand
$\norm{\Am^\dag \bv - \Bm^\dag \bv}_{\Am}^2 \leq h(\gamma) \norm{\Am^\dag \bv}_{\Am}^2$ into 
$$
(\Am^\dag \bv - \Bm^\dag \bv)^T \Am (\Am^\dag \bv - \Bm^\dag \bv) \leq h(\gamma) \bv^T \Am^\dag \bv,
$$
or equivalently,
$$
\bv^T \left(\Am^\dag - \Bm^\dag\right) \Am \left(\Am^\dag - \Bm^\dag\right) \bv \leq h(\gamma)   \bv^T \Am^\dag \bv.
$$
To prove the above inequality, it suffices to prove that 
\begin{equation}
    \label{eq:spec_goal}
     \left(\Am^\dag - \Bm^\dag\right) \Am \left(\Am^\dag - \Bm^\dag\right)  \preceq h(\gamma)  \Am^\dag.
\end{equation}
Multiplying the left and right sides of \Cref{eq:spec_goal} by $\Am^{\frac12}$, we get that (\ref{eq:spec_goal}) is  implied by
\begin{align}
\label{eq:spec_suff}
    \Am^{\frac12}\left(\Am^\dag - \Bm^\dag\right) \Am \left(\Am^\dag - \Bm^\dag\right)\Am^{\frac12}  \preceq h(\gamma) \Am^{\frac12} \Am^\dag \Am^{\frac12}.
\end{align}
Let $\Pi := \Am^{\frac12} \Am^\dag \Am^{\frac12}$ be the projection map onto the row space of $\Am$. Note that $\Pi = \Am^\dag\Am = \Am\Am^\dag$. Also, $\Pi = \Am^{\frac{\dag}{2}}\Am^{\frac{1}{2}} = \Am^{\frac{1}{2}}\Am^{\frac{\dag}{2}}$. These can be seen using the spectral decomposition. 
Now, the reason why \Cref{eq:spec_suff} implies \Cref{eq:spec_goal} is because if we can multiply both sides of (\ref{eq:spec_suff}) with one copy of $ \Am^{\frac{\dag}{2}}$ on the left and one copy of  $\Am^{\frac{\dag}{2}}$ on the right. Then \Cref{eq:spec_suff} becomes  $\Pi\left(\Am^\dag - \Bm^\dag\right) \Am \left(\Am^\dag - \Bm^\dag\right)\Pi \preceq h(\gamma) \Pi\Am^\dag \Pi.$ We have $\Pi (\Am^{\dag} - \Bm^{\dag}) = (\Am^{\dag} - \Bm^{\dag})\Pi  = \Am^{\dag} - \Bm^{\dag}$ because $\Am$ and $\Bm$ have the same null space. Similarly, $\Pi \Am^{\dag}  = \Am^{\dag} \Pi  = \Am^{\dag}$.

To prove (\ref{eq:spec_suff}), first rewrite it as
\begin{align*}
    \left(\Am^{\frac12}\left(\Am^\dag - \Bm^\dag\right) \Am^{\frac12} \right)^2   \preceq h(\gamma) \Pi,
\end{align*}
or equivalently
\begin{align*}
    \left(\Pi - \Am^{\frac12}\Bm^\dag\Am^\frac12 \right)^2   \preceq h(\gamma) \Pi.
\end{align*}
From the spectral approximation
$\left( 1 - \gamma \right) \Am
\preceq
\Bm
\preceq
\left( 1 + \gamma \right) \Am$, we deduce that
$$
\frac{1}{1 + \gamma} \Am^\dag
\preceq
\Bm^\dag
\preceq
\frac{1}{1-\gamma} \Am^\dag,
$$
which, when multiplying on the left and right by $\Am^{\frac12}$, implies that
$$
\frac{1}{1 + \gamma} \Pi
\preceq
\Am^{\frac12}\Bm^\dag\Am^{\frac12}
\preceq
\frac{1}{1-\gamma} \Pi.
$$
This in turn gives
$$
\frac{-\gamma}{1 + \gamma} \Pi
\preceq
\Pi - \Am^{\frac12}\Bm^\dag\Am^{\frac12}
\preceq
\frac{\gamma}{1-\gamma} \Pi.
$$
Observe that any eigenvector of $\Pi - \Am^{\frac12}\Bm^\dag\Am^{\frac12}$ is also an eigenvector of $\Pi$ (they share eigenspaces because $\Am$ and $\Bm$ have the same null spaces).
Moreover, the eigenvalues of $\Pi$ are 0 or 1. 
This implies that the eigenvalues of $\Pi - \Am^{\frac12}\Bm^\dag\Am^\frac12$ are all between $\frac{-\gamma}{1 + \gamma}$ and $\frac{\gamma}{1-\gamma}$. Hence, the eigenvalues of  $\left(\Pi - \Am^{\frac12}\Bm^\dag\Am^\frac12 \right)^2$ are all between 0 and  $ \frac{\gamma^2}{(1-\gamma)^2}$, and thus $\left(\Pi - \Am^{\frac12}\Bm^\dag\Am^\frac12 \right)^2   \preceq \frac{\gamma^2}{(1-\gamma)^2} \Pi$.

\end{proof}

%% file: previous.tex
\section{Previous Works on $p$-Norm Flows}
\label{sec:previous}

Here are some of the existing algorithms for computing minimum $p$-norm flows.
They include algorithms that use electrical flow solvers as a subroutine.  

\begin{enumerate}
    \item The following three algorithms compute an $(1+\epsilon)$-approximation to the optimal weighted $p$-norm flow, and they involve calling Laplacian solvers as a black box: 
    \begin{itemize}
        \item $O(pm^{\frac{4p-4}{3p-2} + o(1)}\log^2 \frac{1}{\epsilon})$ arithmetic operations for $2 \leq p < \mbox{poly}(m)$ \cite{AS20}
    (not considering the bit complexity of the operations), 
        \item $\tilde{O}(pm^{1 + \abs{\frac{1}{2}-\frac{1}{p}}} \log \frac{1}{\epsilon})$ time \cite{BCLL18} for $1 < p < \infty$,
        \item  $\tilde{O}(2^{\max\{p, \frac{1}{p - 1} \} } m^{1 + \frac{\abs{p - 2}}{2p + \abs{p - 2}}} \log \frac{1}{\epsilon})$ time for $1 < p < \infty$ \cite{AKPS19};
    \end{itemize}
    \item An algorithm with time $O(p(m^{1 + o(1)} + n^{4/3 + o(1)})\log^2 \frac{1}{\epsilon})$ has recently been achieved for $p=\omega(1)$ \cite{ABKS21:arxiv}. This algorithm uses graph sparsification in combination with the algorithms in the three bullet points above.
    \item An algorithm that runs in $\tilde{O}(m^{1.2} \log(\frac{1}{\epsilon}))$ time for $p=4$ \cite{Bullins20}, using higher-order acceleration.
    \item For unweighted graphs, a near-optimal flow can be computed in $pm^{\frac{p}{p-1} + o(1)}$ arithmetic operations \cite{AS20} for $2 \leq p \leq \mbox{poly}(m)$, or time $2^{O(p^{3/2})}m^{1 + \frac{7}{\sqrt{p-1}}+o(1)}\mbox{poly}(\log \frac{1}{\epsilon})$  for $p \geq 2$ using recursive preconditioning \cite{KPSW19}.
\end{enumerate}